\documentclass[10pt, a4paper]{article}
\usepackage{booktabs}
\usepackage{bm}
\usepackage{amsthm}
\usepackage{amsmath}
\usepackage{mathrsfs}
\usepackage{graphicx}
\usepackage{authblk}
\usepackage[dvipdfm]{geometry}
\usepackage{arydshln}
\usepackage{times}
\usepackage[justification=centering]{caption}
\usepackage[figuresright]{rotating}
\usepackage{float}

\makeatletter

\usepackage[ruled]{algorithm2e}

\usepackage[farskip=0pt]{subfig}
\usepackage{titlesec}
\usepackage{setspace}
\usepackage[colorlinks,linkcolor=blue,citecolor=red,anchorcolor=blue]{hyperref}

\titleformat*{\section}{\normalsize\bfseries}

\renewcommand\subsection{\@startsection
 {subsection}{2}{0mm}
 {-\baselineskip}
 {0.5\baselineskip}
 {\normalfont\normalsize\itshape\bfseries}}

\renewcommand\subsubsection{\@startsection
 {subsubsection}{2}{0mm}
 {-\baselineskip}
 {0.5\baselineskip}
 {\normalfont\normalsize\itshape\bfseries}}
\makeatother

\begin{document}

\newtheorem{definition}{Definition}
\newtheorem{theorem}{Theorem}
\newtheorem{remark}{Remark}

\title{\bf{The conformable fractional grey system model}\footnote{This is accepted version by ISA Transactions. The published version is available at \url{https://doi.org/10.1016/j.isatra.2019.07.009}}}      
\author[1,2*]{Xin Ma}
\author[1]{Wenqing Wu}
\author[3]{Bo Zeng}
\author[4]{Yong Wang}
\author[4]{Xinxing Wu}

\date{}
\affil[1]{\small School of Science, Southwest University of Science and Technology, Mianyang, China}
\affil[2]{\small State Key Laboratory of Oil and Gas Reservoir Geology and Exploitation,
                  Southwest Petroleum University, Chengdu, China}
\affil[3]{\small College of Business Planning, Chongqing Technology and Business University, Chongqing, China}
\affil[4]{\small School of Science, Southwest Petroleum University, Chengdu, China}

\affil[*]{\small Correspondence authors: Xin Ma, $Email:cauchy7203@gmail.com$}

\maketitle%
\abstract{

The fractional order grey models have appealed considerable interest of research in recent years due to its high effectiveness and flexibility in time series forecasting. However, the existing fractional order accumulation and difference are computationally complex, which leads to difficulties for theoretical analysis and applications. In this paper, new definitions of fractional accumulation and difference are proposed based on the definition of conformable fractional derivative, which are called the conformable fractional accumulation and difference. Then a novel conformable fractional grey model is proposed based on the conformable fractional accumulation and difference, and Brute Force method is introduced to optimize its fractional order. The feasibility and simplicity of the proposed model and the Brute Force method are shown in the numerical example. The conformable fractional grey model outperforms the existing fractional grey model and the autoregressive model in 1 to 3 step predictions with 21 benchmark data sets, and also outperforms the existing fractional grey model in predicting the natural gas consumption of 11 countries. The results indicate that the proposed conformable fractional grey model is more efficient in longer term prediction and non-smooth time series forecasting than the existing models.

\bf{Key words: Grey System Model; Conformable Fractional Calculus; Conformable Fractional Accumulation; Fractional Grey Model; CFGM Model;Natural Gas Consumption }}

\section{Introduction}
\label{sec1:intro}

    The grey models play a key role in the Grey System Theory pioneered by Deng in 1982. The main idea of grey forecasting was initially proposed by Deng in 1983, and the first grey model was developed in 1984 \cite{xie2017review}. The commonly used statistical prediction models, such as the empirical models \cite{zgx2018,maminda2018}, semi-parametric models \cite{knea},  hybrid models \cite{xinyi01,xinyi02}, and especially the machine learning models \cite{WU2019280,fan2019empirical,YANG2019942}, often need large amount of sample data to make convinced predictions. But the idea of grey models can often built on very small samples. In the early report by Deng \cite{deng1988}, it was shown that the basic GM(1, 1) model can be built with only 4 points to make acceptable predictions.  Recent theoretical and empirical studies have also shown that the grey prediction models can be efficient in time series forecasting with very few data \cite{smallsample}, and can even be more efficient than the machine learning models \cite{kgm1n}. With high effectiveness in time series forecasting with small samples, the grey models have recently been used in a wide variety of application fields, such as manufacturing industries \cite{zb2019apm}, energy marketing \cite{wzx2017}, energy economics \cite{zb2018} and environment issues \cite{wang2019modelling,meng2018prediction,xiong2019grey}, $etc$.

    Significant efforts have been made to improve accuracy of the grey model family in recent years. One of the most important methods is the reconstruction of the back ground values proposed by Tan \cite{tan2000}, which was mainly used for univariate grey models. Zeng $et\ al.$  have developed the method of dynamic back ground value to improve the multivariate grey model, which can be regarded as an extension of the method by Tan \cite{zbcie2018}. On the other hand, some researchers have noticed that the parameters of most grey models are obtained using the least squares method, which limited the flexibility of the grey models. Pei $et\ al.$ employed the nonlinear least squares for the parameter estimation of the nonlinear Bernoulli models, which has been presented to be more efficient than the existing models. And it has been pointed out by Ma and Liu that the inconsistency of the solution and the discrete form of the grey model is the essential reason which lead large errors to the GMC(1, n) model, and the optimized parameter transforms have been proposed. Then the GMCO(1, n) model with optimal parameters has been proved to unbiased to arbitrary linear dynamical grey systems \cite{maxingmco}. Another kind of methods mainly focuses on the modelling mechanism of the grey models. The discrete grey modelling technique (DGMT) is one of the most important new mechanisms, which was initially introduced by Xie $et\ al.$ in 2009 \cite{xie2009}. The DGMT has been applied to build several effective grey models, such as the NDGM model \cite{ngdm2013}. And recently it has also been extended to build the multivariate grey models, such as the DGM(1, N) \cite{dgm1n2015}, RDGM(1, n) \cite{rdgm2016}, CDGM(1, N) \cite{DING2019}, $etc$. Some other methods, such as the intelligent optimizers \cite{ds2018}, kernel machine learning \cite{maxin2019}, data grouping \cite{wzxdatagroup}, high dimensional tensors \cite{duantenor}, have also been introduced to build the grey models.

    However, most of the main stream grey models are essentially linear models. Actually, early in 2002 Deng \cite{dengbook} has pointed out in his Chinese book that nonlinearities widely exist in the real-world applications, so the linear grey models are not sufficient for more general cases, such as the processes of the underground fluid flow \cite{wy1,wy2}, the management of the unconventional petroleum reservoirs\cite{dingxf2018,hysspe,wy3}, complex economical systems \cite{LIANG2019315} and operational systems \cite{wwq2018math}, etc. In Deng's idea, it was implied that the nonlinear grey models should be developed on nonlinear differential equations \cite{dengbook}, thus he firstly introduced the Bernoulli equation, but leaving no discussions on the computational steps. Chen $et\ al.$ have initially presented the detailed computational steps of the nonlinear grey Bernoulli model \cite{chenngbm11}, and soon some improved works have been illustrated by other researchers \cite{wangzx2011}. Wang $et\ al.$ generalized such idea and developed several multivariate grey models by using the power terms \cite{wang2014,wang2017}. And in recent research the Bernoulli equation has been applied to build the multivariate grey models \cite{ma2019bernoulli}.

    In addition to the above method, the fractional order accumulation (FOA) is also one of the most important innovations for the grey models. The FOA was firstly introduced by Wu $et\ al.$ in 2013 to build a fractional grey model (FGM) \cite{wulifeng2013}.  Without using the nonlinear equations, the FOA was used as an nonlinear data preprocessing method to describe the nonlinear series for the grey models. With high performance of improving the grey models and innovative methodology, the FOA and FGM soon appealed considerable interest of research in recent 6 years, and have been widely used in the real-world applications. For instance, the basic FGM has been applied to predict the weapon system costs\cite{fgmweapon}, pollutant gas emission\cite{fgmgas}, air quality of Beijing-Tianjin-Hebei region and the lightly polluted day in Jing-Jin-Ji region in China  \cite{wubeijing,wujjj}. There are also many other fractional grey models proposed to solve specific application problems. A discrete fractional grey model was built to predict the CO$_2$ emission of China By Gao $et.\ al$\cite{mao2015co2}. Two formulations of self-adaptive grey models  were proposed by Zeng $et\ al.$  to predict the shale gas output China and the electricity consumption of China, respectively \cite{zeng2018buffer,zeng2017self}. The PSO-based FGM was proposed by Duan $et\ al.$  to predict the crude oil consumption of China \cite{duan2018}. A nonhomogeneous FGM model with optimized parameters FAGMO(1, 1, k) was introduced by Wu $et\ al.$ to forecast the nuclear energy consumption of China \cite{wwq2018}. A new fractional order time delayed term was introduced by Ma $et\ al.$, which can make the fractional grey models more flexible \cite{ma2019energy}. The FOA was also introduced to build the nonlinear Grey Bernoulli by Wu $et\ al.$ \cite{wu2019bernoulli} to forecast the renewable energy consumption of China. The GMC(1, n) with FOA was presented by Wu $et\ al.$ to predict the electricity consumption of Shandong Province of China \cite{wu2018gmc}, and it has been further improved using discrete grey modelling technique in our recent works \cite{ma2019apm}.

    According to Wu's results, the FOA can perform as an error reducer for the grey models \cite{reduceerror}, and the FGM is also effective in  time series forecasting with small samples. It was noticed that the basic structure of the FGM models would not be changed by FOA, thus it can also be used for building other commonly used grey models. For instance, Wu $et\ al.$ have rebuilt the nonhomogenous discrete grey model using the FOA \cite{fndgm}. The FOA has also been proved to be efficient to revise the relational analysis, thus the fractional order grey relational analysis has also been prosed \cite{frelation}. With the enlightening of the FOA, some researchers started to broaden the usage of the fractional calculus. One of the typical works is the grey model with fractional derivatives by Yang and Xue \cite{xdy2016}, which introduced the fractional order calculus to the continuous grey models framework, and it has also been applied to develop a novel interval model \cite{xdy2017}. Essentially, the FOA can be regarded as a preprocessing method for the grey models, and it is also a general form of the well-known 1-AGO in the traditional Grey System Theory. With the success of the FOA, some researchers also mentioned that the variation of the preprocessing methods can be effective to improve the grey models, thus some new operators have also been proposed in recent years, such as the weakening buffer operator (WBO) \cite{weakbuffer} and the inverse fractional order accumulation (IFOA) \cite{mwfgm}. Above all, the introduction of the FOA has made significant contributions to the development of Grey System Theory, thus it has also been listed as one of the most achievements in the new millennium \cite{liu2016new}.

    However, the definition introduced by Wu is just one case of the fractional order calculus and differencing. In the point of view of computational complexity, the commonly used definitions of the fractional accumulation and its corresponding fractional differencing in the existing grey models are not easy to implement, and it leads to high difficulties to the deeper theoretical analysis. As mentioned above, the fractional order calculus has also been introduced to the grey models by Yang and Xue \cite{xdy2016}, but the analytic solution of such fractional models contains infinite series, this is obviously not easy to use and analyze. Moreover, such difficulties would also hinder the development of the fractional order grey models and even the Grey System Theory.

    According to our investigations, a new definition of the fractional order derivative has been proposed by Khalil $et\ al.$ in 2014 \cite{cfd2014}, which is called Conformable Fractional Derivative. The new definition of derivative is much simpler than the ``old'' definitions of the fractional order derivatives, such as the Riemann-Liouville definition and Caputo definition. It was proved by Khalil $et\ al.$ the conformable fractional derivative has very good properties, and often can solve many problems which are difficult or impossible when using the ``old'' definitions. With such significant improvements, the conformable fractional derivative has soon appealed significant interest by the researchers in recent 4 years, and many valuable new findings have been presented. Hammad and Khalil \cite{cfd02} proposed the Fractional Fourier Series (FFS) based on the conformable fractional derivative, and it was shown that the FFS is quite efficient to solve the partial differential equations. With the new formulation of conformable fractional derivative, many important differential equations have been redefined \cite{cfd03,cfd04,cfd05}, and it is also important to see that the comprehensive analysis can be easily carried out with the conformable fractional derivative in these works.

    According to the reviews above, we are motivated to put forward the idea of using the conformable fractional derivative by Khalil $et\ al.$ \cite{cfd2014} to fix the issues of computational complexity for the existing fractional grey models. Thus, we firstly introduce the new definitions of the fractional order difference and accumulation based on the conformable fractional derivative, and name them as the Conformable Fractional Difference and Conformable Fractional Accumulation, and then use them to build the novel Conformable Fractional Grey Model (CFGM). Comprehensive numerical example, benchmark data sets validation and real-world case studies would also presented in order to compare the properties of the CFGM and the existing FGM.

    The rest of this paper is organized as follows: the Section 2 gives a brief overview of the definition of the Conformable Fractional Derivative by Khalil $et\ al.$, and presents the definitions of the Conformable Fractional Difference and Accumulation; Section 3 presented the modelling procedures of the new CFGM in a very brief way; a numerical example with detailed steps of the CFGM is presented in Section 4; validations based on the benchmark data sets for comparing the CFGM and FGM are presented Section 5; the cases studies of predicting the natural gas production in the 11 countries are shown in Section 6, along with comprehensive discussions on the properties of CFGM comparing to the FGM, and the conclusions and perspectives are shown in Section 7.

\section{The conformable fractional accumulation}
\label{sec2:pre}

   In this section, we firstly give a brief overview of the conformable fractional derivative and some important properties. Then the conformable fractional difference can be define in a very natural way, and the conformable fractional accumulation is just the inverse operator of it.

\subsection{Definition of the conformable fractional derivative}

    The conformable fractional derivative defined by Khalil, $et\ al.$ can be represented as follows:
    \begin{definition}(see \cite{cfd2014}) \label{def:cfdev}
    Given a differentiable function $f:[0,\infty) \rightarrow R$. Then the conformable fractional derivative  of $f$ with $\alpha$ order is defined as
    \begin{equation}
        T_{\alpha}(f)(t) = \lim_{\varepsilon \rightarrow 0} \frac{f(t+\varepsilon t ^{1-\alpha})- f(t)}{\varepsilon}
    \end{equation}
    for all $t>0,\alpha \in (0,1]$.
    \end{definition}

    Within the definition given above, we can easily obtain the following property of the the conformable derivative as

    \begin{theorem}(see \cite{cfd2014}) \label{th:cfdev}
    If the function $f$ is differentiable, then we have
    \begin{equation}
        T_{\alpha}(f)(t) = t^{1-\alpha} \frac{df(t)}{dt}.
    \end{equation}
    for all $t>0,\alpha \in (0,1]$.
    \end{theorem}
    \begin{proof}
        Let $h=\varepsilon t ^{1-\alpha}$, then $\varepsilon= h t^{\alpha-1}$. Therefore,
        $$
        \begin{aligned}
            T_{\alpha}(f)(t)& = \lim_{\varepsilon \rightarrow 0} \frac{f(t+\varepsilon t ^{1-\alpha})- f(t)}{\varepsilon} \\
                            & = t ^{1-\alpha} \lim_{h \rightarrow 0} \frac{f(t+h)- f(t)}{h} \\
                            & = t ^{1-\alpha} \frac{df(t)}{dt}.\\
        \end{aligned}
        $$
    \end{proof}

    The Theorem \ref{th:cfdev} is very important as it describes the relationship between the conformable fractional derivative and the common derivative. And in the next subsection it will be shown that this property would be quite useful for us to define the conformable fractional difference and accumulation.

    Without loss of generality, we can also define the higher order derivative and the similar relationship to the commonly defined derivatives. Firstly we notice that the Definition \ref{def:cfdev} is actually in the following form
    \begin{equation}
        T_{\alpha}(f)(t) = \lim_{\varepsilon \rightarrow 0} \frac{f(t+\varepsilon t ^{[\alpha]-\alpha})- f(t)}{\varepsilon} = t ^{[\alpha]-\alpha} \frac{df(t)}{dt},
    \end{equation}
    where $[\cdot]$ is the ceil function, $i.\ e.$ the $[\alpha]$ is the smallest integer not smaller than $\alpha$. Thus we can still use the similar formulation to define the higher order conformable derivative.

    \begin{remark}\label{remark:cfdev}(see \cite{cfd2014})
        The  $\alpha$ order ($\alpha \in (n,n+1]$) conformable derivative is defined as follows
        \begin{equation}
            \begin{array}{ll}
                T_{\alpha}(f)(t)  = \lim_{\varepsilon \rightarrow 0} \frac{f^{(n)}(t+\varepsilon t ^{ [\alpha] -\alpha})- f^{(n)}(t)}{\varepsilon}
            \end{array}
        \end{equation}
        for all $t>0,\alpha \in (n,n+1], n\in N^{+}$, and $f$ is $(n+1)-$differentiable. And also, set $h=\varepsilon t ^{1-\alpha}$ , then $\varepsilon = h t ^{\alpha-1}$, we have
        \begin{equation}\label{eq:hdif}
            \begin{array}{ll}
                T_{\alpha}(f)(t)  & = t ^{[\alpha] -\alpha} \lim_{h \rightarrow 0} \frac{f^{(n)}(t+h)- f^{(n)}(t)}{h} \\
                                    & = t ^{[\alpha] -\alpha} \frac{d^n f(t)}{dt^n}.\\
            \end{array}
        \end{equation}
        Then when $\alpha=n+1$, there is $[\alpha] -\alpha=0$, thus we have
        \begin{equation}
                T_{n+1}(f)(t)  = \lim_{\varepsilon \rightarrow 0} \frac{f^{(n)}(t+\varepsilon )- f^{(n)}(t)}{\varepsilon} = \frac{d^{n+1} f(t)}{dt^{n+1}}.
        \end{equation}
    \end{remark}

\subsection{Definition of the conformable fractional accumulation and difference}

    It is well known that the first order difference can be easily defined by the approximation of the first order derivative as
    $\Delta f(k) \approx \frac{f(t)}{t}|_{t=k}= \lim_{h \rightarrow 1} \frac{f(t+h)- f(t)}{h} |_{t=k} = f(k+1)-f(k)$. Considering the Theorem \ref{th:cfdev} it is very natural to give the definition of the conformable fractional difference as follows:
    \begin{definition} \label{def:cfd}
     The conformable fractional difference (CFD) of $f$ with $\alpha$ order is defined as
    \begin{equation}
        \Delta^{\alpha} f(k) = k ^{1-\alpha} \Delta f(k) = k ^{1-\alpha} [f(k)-f(k-1)].
    \end{equation}
    for all $k \in N^{+},\alpha \in (0,1]$.
    \end{definition}

    Within this definition, we are going to define the conformable fractional accumulation. At first, let's recall the definition of the first order accumulation as
    \begin{equation}
        \nabla f(k) = \sum_{j=1}^{k}f(j).
    \end{equation}

    It is known that the first order accumulation is the inverse operator of the first order difference, because
    \begin{equation}\label{eq:inv1}
       \Delta \nabla f(k) = \Delta \left( \sum_{j=1}^{k} f(j) \right) = \sum_{j=1}^{k}f(j) - \sum_{j=1}^{k-1}f(j)=f(k).
    \end{equation}

    Without loss of generality, we denote the conformable fractional accumulation as $\nabla^{\alpha}$, and it should also satisfy the following relationship:
    \begin{equation}\label{eq:r1}
       \Delta^{\alpha} \nabla^{\alpha} f(k)= f(k).
    \end{equation}

    Considering the Definition \ref{def:cfd} of CFD, we can rewrite the Eq.\eqref{eq:r1} as
    \begin{equation}\label{eq:r2}
       \Delta^{\alpha} \nabla^{\alpha} f(k)= k^{1-\alpha} \Delta \left( \nabla^{\alpha} f(k)\right) =f(k).
    \end{equation}

    Dividing Eq.\eqref{eq:r2} by $k^{1-\alpha}$, we have
    \begin{equation}\label{eq:r3}
      \Delta \left( \nabla^{\alpha} f(k)\right)  =   \frac{f(k)}{k^{1-\alpha}}.
    \end{equation}
    Taking the first order accumulation of Eq.\eqref{eq:r3}, we can obtain the definition of conformable fractional accumulation as
    \begin{definition} \label{def:cfa}
     The conformable fractional accumulation (CFA) of $f$ with $\alpha$ order is defined as
    \begin{equation}\label{eq:ex1}
        \nabla^{\alpha} f(k) = \nabla \left( \frac{f(k)}{k ^{1-\alpha}} \right) =  \sum_{j=1}^{k} \left( \frac{f(j)}{j ^{1-\alpha}} \right).
    \end{equation}
    for all $k \in N^{+},\alpha \in (0,1]$.
    \end{definition}

    Comparing to the higher order conformable derivative  defined in Remark \ref{remark:cfdev}, we can easily define the higher order CFD.
    \begin{definition}\label{def:highercfd}
        The $\alpha$ order ($\alpha \in (n,n+1]$) CFD is defined as
        \begin{equation}\label{eq:hcfd}
           \Delta^{\alpha} f(k)= k^{[\alpha]-\alpha}  \Delta^{n}f(k)
        \end{equation}
        for all $n \in N.$
    \end{definition}

    Obviously, when $\alpha=1$ it yields the $n+1$ order difference $\Delta^{n+1}$. And the Definition \ref{def:highercfd} is a uniform definition of the CFD as it holds for all nonnegative $\alpha$ including $\alpha=0$.

    Notice that the higher order CFA is still the inverse operator of the higher order CFD, $i.\ e.$
    $$\Delta^{\alpha} \nabla^{\alpha} f(k)= f(k)$$ for $\alpha \in (n,n+1]$.

    Recalling the Definition \ref{def:highercfd}, we have
    \begin{equation}\label{eq:inv2}
       k^{[\alpha]-\alpha} \Delta^{n} \nabla^{\alpha} f(k) = f(k).
    \end{equation}

    Similarly, we can also obtain the definition of the $\alpha$ order CFA by dividing \eqref{eq:inv2}  by $k^{[\alpha]-\alpha}$  and using the relationship $\nabla^{n}\Delta^{n} f(k)=f(k)$.
    \begin{definition}\label{def:highercfa}
    The $\alpha$ order ($\alpha \in (n,n+1]$) CFA is defined as
    \begin{equation}\label{eq:hcfa}
                  \nabla^{\alpha} f(k) = \nabla^{n} \left( \frac{f(k)}{k ^{[\alpha]-\alpha}} \right).
        \end{equation}
    \end{definition}

    When $\alpha=n+1$ the CFA yields the $(n+1)$ order accumulation $\nabla^{n+1}$. And also the Definition \ref{def:highercfa} is a uniform definition for the CFA as it holds for all nonnegative $\alpha$ including $\alpha=0$. With the uniform definition of CFA, we can easily deduce a recursive equality as
    \begin{equation}\label{eq:cfarec}
                  \nabla^{\alpha} f(k) = \nabla \left(\nabla^{n-1} \left( \frac{f(k)}{k ^{[\alpha]-\alpha}} \right)\right) = \sum_{j=1}^{k} \left(\nabla^{\alpha-1} f(j) \right), \alpha \geq 1,
    \end{equation}
    and $n = [\alpha]-1$. This formulation is very convenient for computer implementation.

\subsection{Comparison to the existing fractional order accumulation and difference}

    The fractional order accumulation (FOA) introduced by Wu \cite{wulifeng2013}, which is often used in the existing fractional grey models is usually defined as
    \begin{equation}\label{eq:fago}
              \nabla_{W}^{\alpha} f(k) = \sum_{j=1}^{k} \left(
              \begin{array}{c}
                 k-j+\alpha-1  \\ k-j
              \end{array}\right) f(j).
    \end{equation}

    And its inverse operation, the fractional order difference (FOD), is defined as
    \begin{equation}\label{eq:ifago}
              \Delta_{W}^{\alpha} f(k) = \sum_{j=1}^{k} \left(
              \begin{array}{c}
                 k-j-\alpha-1  \\ k-j
              \end{array}\right) f(j).
    \end{equation}

    The coefficients in Eqs. \eqref{eq:fago} and \eqref{eq:ifago} can be uniformly defined as
    \begin{equation}\label{eq:cofago}
              \left(
              \begin{array}{c}
                 k-j \pm \alpha-1  \\ k-j
              \end{array}\right) = \frac{(k-j \pm \alpha-1)(k-j \pm \alpha-1)\cdots(\pm \alpha+1)(\pm \alpha)}{(k-j)!}.
    \end{equation}

    The fractional order $\alpha$ can be arbitrary fractional numbers (actually it can be any real numbers with such formulations according to the definition in \cite{mao2016novel}).

    It is clear that the definitions of the FOA and FOD above are more complex. And it would be shown in the numerical example that the implementations of the CFA and CFD are quite simple.

\section{The Conformable Fractional Grey Model}
\label{sec:cfgm}
    Within the definitions of the CFA and CFD we can rebuild the classical grey system model GM(1, 1), and these procedures are presented in this section.
\subsection{Formulation of the conformable fractional grey model}
\label{sec:modelcfgm}

    With the original series $X^{(0)}= \left( x^{(0)}(1),x^{(0)}(2),...,x^{(0)}(N)\right)$, we firstly denote the $\alpha$ order CFA as \begin{equation}\label{eq:alphaago}
        X^{(\alpha)}= \left( x^{(\alpha)}(1),x^{(\alpha)}(2),...,x^{(\alpha)}(N) \right),
    \end{equation}
    where
    \begin{equation}
        x^{(\alpha)}(k) = \nabla^{\alpha} x^{(0)}(k) = \left\{
        \begin{aligned}
        &\sum_{j=1}^{k} \frac{x^{(0)}(j)}{j ^{[\alpha]-\alpha}},& 0<\alpha \leq 1,\\
        &\sum_{j=1}^{k}x^{(\alpha-1)}(j), & \alpha >1.
        \end{aligned}\right.
    \end{equation}

    Then the Conformable Fractional Grey Model is represented as
    \begin{equation}\label{eq:whiten}
        \frac{dx^{(\alpha)}(t)}{dt} + a x^{(\alpha)}(t) = b.
    \end{equation}

    In the rest of this paper, we abbreviate it as CFGM for convenience. And the differential equation \eqref{eq:whiten} is called its whitening equation. When $\alpha=1$, it yields the conventional GM(1, 1) in Liu's book \cite{liubook}. And if the fractional order accumulated series \eqref{eq:alphaago} is computed by the existing FOA in \eqref{eq:fago}, the CFGM model can be translated to the existing FGM model by Wu \cite{wulifeng2013}.

    The discrete form of the whitening equation can be obtained using the trapezoid formula as
    \begin{equation}\label{eq:dis}
        \left( x^{(\alpha)}(k)-x^{(\alpha)}(k-1)\right)  + \frac{a}{2} \left( x^{(\alpha)}(k) + x^{(\alpha)}(k-1)\right) = b.
    \end{equation}

    Similar procedures can be found in recent researches on fractional grey models.

\subsection{Parameters estimation}

    Within the discrete form \eqref{eq:dis} we can easily obtain the parameter estimation of the CFGM model using the least squares method with given samples and $\alpha$ as
    \begin{equation}\label{eq:lse}
        [\hat{a},\hat{b}]^T= (B^T B)^{-1}B^T Y,
    \end{equation}
    where
    \begin{equation}\label{eq:matrixby}
    B=
        \left[
          \begin{array}{cc}
            -\frac{1}{2} \left( x^{(\alpha)}(2) + x^{(\alpha)}(1)\right) & 1 \\
            -\frac{1}{2} \left( x^{(\alpha)}(3) + x^{(\alpha)}(2)\right) & 1 \\
            \vdots & \vdots \\
            -\frac{1}{2} \left( x^{(\alpha)}(N) + x^{(\alpha)}(N-1)\right) & 1 \\
          \end{array}
        \right]
         ,
         Y=
         \left[
            \begin{array}{c}
               x^{(\alpha)}(2)-x^{(\alpha)}(1) \\  x^{(\alpha)}(3)-x^{(\alpha)}(2) \\ \vdots \\  x^{(\alpha)}(N)-x^{(\alpha)}(N-1)
            \end{array}
         \right].
    \end{equation}

\subsection{Response function and restored values}

    Notice that the initial point $x^{(\alpha)}(1)= \nabla \frac{x^{(0)}(1)}{1^{[\alpha] - \alpha}}=x^{(0)}(1)$, the response function can be obtained by solving the whitening equation \eqref{eq:whiten} as
    \begin{equation}\label{eq:restemp}
        x^{(\alpha)}(t) = ( x^{(0)}(1) - \frac{b}{a} ) e^{-a(t-1)} + \frac{b}{a}.
    \end{equation}
    Thus the predicted values of the CFA series can be computed using the estimated parameters and the discrete form of the response function as
    \begin{equation}\label{eq:resp}
        \hat{x}^{(\alpha)}(k) = ( x^{(0)}(1) - \frac{\hat{b}}{\hat{a}} ) e^{-\hat{a}(k-1)} + \frac{\hat{b}}{\hat{a}}.
    \end{equation}

    Then restored values can be obtained using the CFD as
    \begin{equation}\label{eq:restore}
        \hat{x}^{(0)}(k) = \Delta ^\alpha \hat{x}^{(\alpha)}(k)= k^{[\alpha] - \alpha} \Delta ^n \hat{x}^{(\alpha)}(k), \alpha \in (n,n+1].
    \end{equation}
    Particularly, when $\alpha \in (0,1]$ the restored values can be written as
     \begin{equation}\label{eq:restore1}
        \hat{x}^{(0)}(k) = k^{1 - \alpha} \left(\hat{x}^{(\alpha)}(k)-\hat{x}^{(\alpha)}(k-1)\right).
    \end{equation}

\subsection{Computation steps}
\label{subsec:cmptstep}

    The computation steps of CFGM model with given sample and $\alpha$ can be summarized as follows:

    \emph{\textbf{Step 1:}} Compute the $\alpha$ order CFA series $\left( x^{(\alpha)}(1), x^{(\alpha)}(2),...,x^{(\alpha)}(N) \right)$ of the given sample $\left( x^{(0)}(1), x^{(0)}(2),...,x^{(0)}(N) \right)$;

    \emph{\textbf{Step 2:}} Compute the parameters $\hat{a}$ and $\hat{b}$ using Eq.\eqref{eq:lse};

    \emph{\textbf{Step 3:}} Compute the predicted values of CFA series $\hat{x}^{(\alpha)}(k)$ and the restored values using the response function \eqref{eq:resp} and the CFD \eqref{eq:restore} by $k$ from 1 to $n+p$, respectively. Where $n$ is the number of samples and $p$ is the number of prediction steps.

\subsection{Brute force method for selecting the optimal $\alpha$}
\label{sec:bf}
    It should be noticed that the above procedures are given with the assumption that the value $\alpha$ is given. Thus we can build a very simple optimization problem for the optimal $\alpha$ as

    $$\underset{\alpha}{\min} \quad \text{MAPE} = \sum_{j=1}^{N} \left|\frac{\hat{x}^{(0)}(j) - x^{(0)}(j)}{x^{(0)}(j)} \right|\times 100\% $$
    \begin{equation}\label{eq:opt}
    \begin{aligned}
     s.t.
        \left\{
          \begin{array}{l}
            [\hat{a},\hat{b}]^T = (B^T B)^{-1}B^T Y \\
                B=
        \left[
          \begin{array}{cc}
            -\frac{1}{2} \left( x^{(\alpha)}(2) + x^{(\alpha)}(1)\right) & 1 \\
            \vdots & \vdots \\
            -\frac{1}{2} \left( x^{(\alpha)}(N) + x^{(\alpha)}(N-1)\right) & 1 \\
          \end{array}
        \right]\\
         Y=
         \left[
            \begin{array}{c}
               x^{(\alpha)}(2)-x^{(\alpha)}(1) \\ \vdots \\  x^{(\alpha)}(k)-x^{(\alpha)}(k-1)
            \end{array}
         \right]\\
          \hat{x}^{(\alpha)}(k) = ( x^{(0)}(1) - \frac{\hat{b}}{\hat{a}} ) e^{-\hat{a}(k-1)} + \frac{\hat{b}}{\hat{a}} \\
          \hat{x}^{(0)}(k) = k^{[\alpha] - \alpha} \left( \hat{x}^{(\alpha)}(k) - \hat{x}^{(\alpha)}(k-1)\right). \quad k=2,3,...,N.\\  \end{array}
        \right.
        \end{aligned}
    \end{equation}

    This formulation is often used in the nonlinear grey models with tunable parameters. And its objective is to find out the optimal parameter $\alpha$ which minimizes the mean absolute percentage error (MAPE) of the model. The optimization problem is essentially a nonlinear programming with nonlinear objective function and nonlinear constraints. There are a lot of nonlinear optimizers used for the grey system models in recent publications, and heuristics are often adopted to solve similar problems in other fields \cite{Xiao2019AE}. However, in this paper we do not use the cutting-edge optimizers but the Brute Force method. The Brute Force method is a basic method for solving the optimization problems, it can solve almost all the problems but with low efficiency (with large number of iterations and low accuracy). Thus if the Brute Force method is available to select an appropriate $\alpha$ for the CFGM model, it is sufficient to prove its simplicity and applicability.

    In this paper we enumerate all the values in the interval $[0,2]$ with step $0.01$, then use the computational steps presented in subsection \ref{subsec:cmptstep} and select the $\alpha$ corresponds to the minimum MAPE as the optimal value. The detailed algorithm is presented in Algorithm \ref{al:alg1}.

    \begin{algorithm}
    \small
     \caption{Brute force method to compute the fractional order $\alpha$ of CFGM\label{al:alg1}}
      \KwIn{Original series $\Big( x^{(0)}(1),x^{(0)}(2),...,x^{(0)}(N)\Big)$}
      \KwOut{Optimal $\alpha^*$}
      Initialize   $\alpha^*=0$, MAPE$_{min}=inf$\;

        \For{$\alpha$=0 to 2; Step=0.01}
        {
            Construct $B$ and $Y$ using Eq.\eqref{eq:matrixby}\:

            Compute $\hat{a},\hat{b}$ using Eq. \eqref{eq:lse}\;
            \For{$k$=1 to n; Step=1}
            {
                Compute $\hat{x}^{(\alpha)}(k)$  using Eq.\eqref{eq:resp} \;
                Compute $\hat{x}^{(0)}(k)$ using Eq.\eqref{eq:restore} \;
            }
            Compute MAPE using the objective function in Eq.\eqref{eq:opt}\;
            \If{MAPE $<$ MAPE$_{min}$}
            {
                MAPE$_{min} \leftarrow$ MAPE\;
                $\alpha^*$ $\leftarrow$ $\alpha$;
           }
        }
        return $\alpha^{*}$\;
    \end{algorithm}

\section{Numerical example}

    In this section we present a numerical example to show the computational steps of the  CFGM model with the raw data
    $X^{(0)}= \Big(55.70,59.01,62.10,62.27,60.81,58.39,55.43,52.18,48.80,45.41 \Big)$. In the following contexts, we use the first 5 points to build the CFGM model, leaving last 5 points for testing. The decimals listed in this section are often truncated in order to make them easy to display. And all the numbers are set to be ``Real number'' in Matlab R2018a in all computational steps in this paper.

\subsection{Computing the CFA of the original series}
\label{sec:cfa}

\begin{table}[!htb]
\caption{ Computational details for the CFA with $\alpha=1.1$  \label{t:ex1_cfa}}
\centering
\scriptsize
\begin{tabular}{rrrrrr}
\toprule
$k$		& $x^{(0)}(k)$	&	$k^{[\alpha]-\alpha}$ 	&	 $\frac{x^{(0)}(k)}{k^{[\alpha]-\alpha}}$	&	 $x^{(\alpha-1)}(k)=\sum_{j=1}^{k}\frac{x^{(0)}(j)}{j^{[\alpha]-\alpha}}$	&	 $x^{(\alpha)}(k)=\sum_{j=1}^{k}x^{(\alpha-1)}(j)$	\\
\\
\hline
\\
1	&	55.7	&	1.00 	&	55.70 	&	55.70 	&	55.70 	\\
2	&	59	&	1.87 	&	31.62 	&	87.32 	&	143.02 	\\
3	&	62.7	&	2.69 	&	23.33 	&	110.64 	&	253.66 	\\
4	&	61.3	&	3.48 	&	17.60 	&	128.25 	&	381.91 	\\
5	&	61.4	&	4.26 	&	14.42 	&	142.67 	&	524.58 	\\
\bottomrule
\end{tabular}
\end{table}
    Computation of the CFA of the original series is the first step to build the CFGM model. For convenience, we firstly set $\alpha=1.1$ to illustrate the computational steps of CFA, and the detailed steps are listed in Table \ref{t:ex1_cfa}.

    In the first row of Table \ref{t:ex1_cfa} we present the formulations in the computational steps, and the corresponding values are listed in the following rows. We firstly list the values of $x^{(0)}(k)$ in the second column, and the values of $k^{[\alpha]-\alpha}$ in the third column, respectively. Recalling the uniform Definition \ref{def:highercfa} of the higher order CFA, we need firstly to compute the values of $\frac{x^{(0)}(k)}{k^{[\alpha]-\alpha}}$, which are listed in the fourth column in Table \ref{t:ex1_cfa}. Noticing that $\alpha = 1.1 \in (1,2]$, then $[\alpha]-\alpha=0.1$, the computational formulation for the CFA should be $$x^{(1.1)}(k)= \nabla^{2} \left( \frac{x^{(0)}(k)}{k^{0.1}} \right).$$

    This implies that we need to computing the second order accumulation of $\frac{x^{(0)}(k)}{k^{[\alpha]-\alpha}}$. Noticing that the first order accumulation of $\frac{x^{(0)}(k)}{k^{[\alpha]-\alpha}}$ is actually $$x^{(\alpha-1)}(k)= \nabla \left( \frac{x^{(0)}(k)}{k^{0.1}} \right),$$ then we can write $x^{(\alpha)}(k)=\sum_{j=1}^{k}x^{(\alpha-1)}(j)$,  thus we list the values of $x^{(\alpha-1)}(k)$ before the values of $x^{(\alpha)}(k)$ in the fifth and  sixth column, respectively.

    The above descriptions indicate that the CFA can be easily implemented. Actually, the Table \ref{t:ex1_cfa} is implemented in the software Microsoft Excel 2010.

    Using the similar computational steps, we can easily obtain the values of CFA with $$\alpha=0.1,0.2,...,2.$$ The subfigures in Fig. \ref{fig:ex_cfa} presents the plots of CFA with $\alpha=0.1,0.2,...,1$ and $\alpha=1.1,1.2,...,2.$ It can be seen that the CFA series approaches to the 1-AGO series when $\alpha$ approaches to 1, and it approaches to 2-AGO series when $\alpha$ approaches to 2. It can also be seen that for each point the CFA value $x^{(\alpha)}(k)$ becomes larger with larger $\alpha$, and the growing speed also increases with larger $\alpha$.

\begin{figure}[!htb]

\subfloat[]{
  \includegraphics[width=0.45\textwidth]{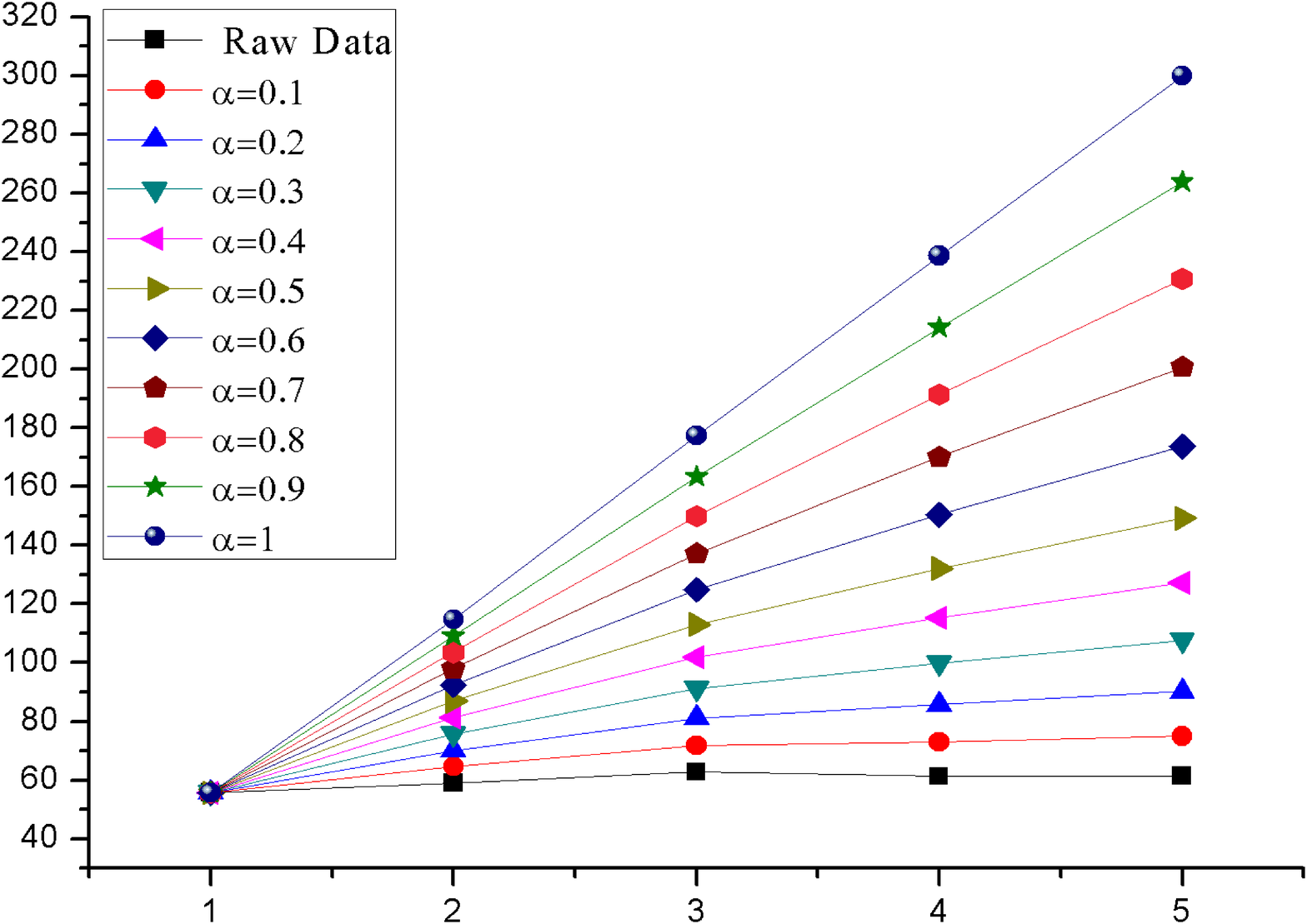}}\hfill
\subfloat[]{
  \includegraphics[width=0.45\textwidth]{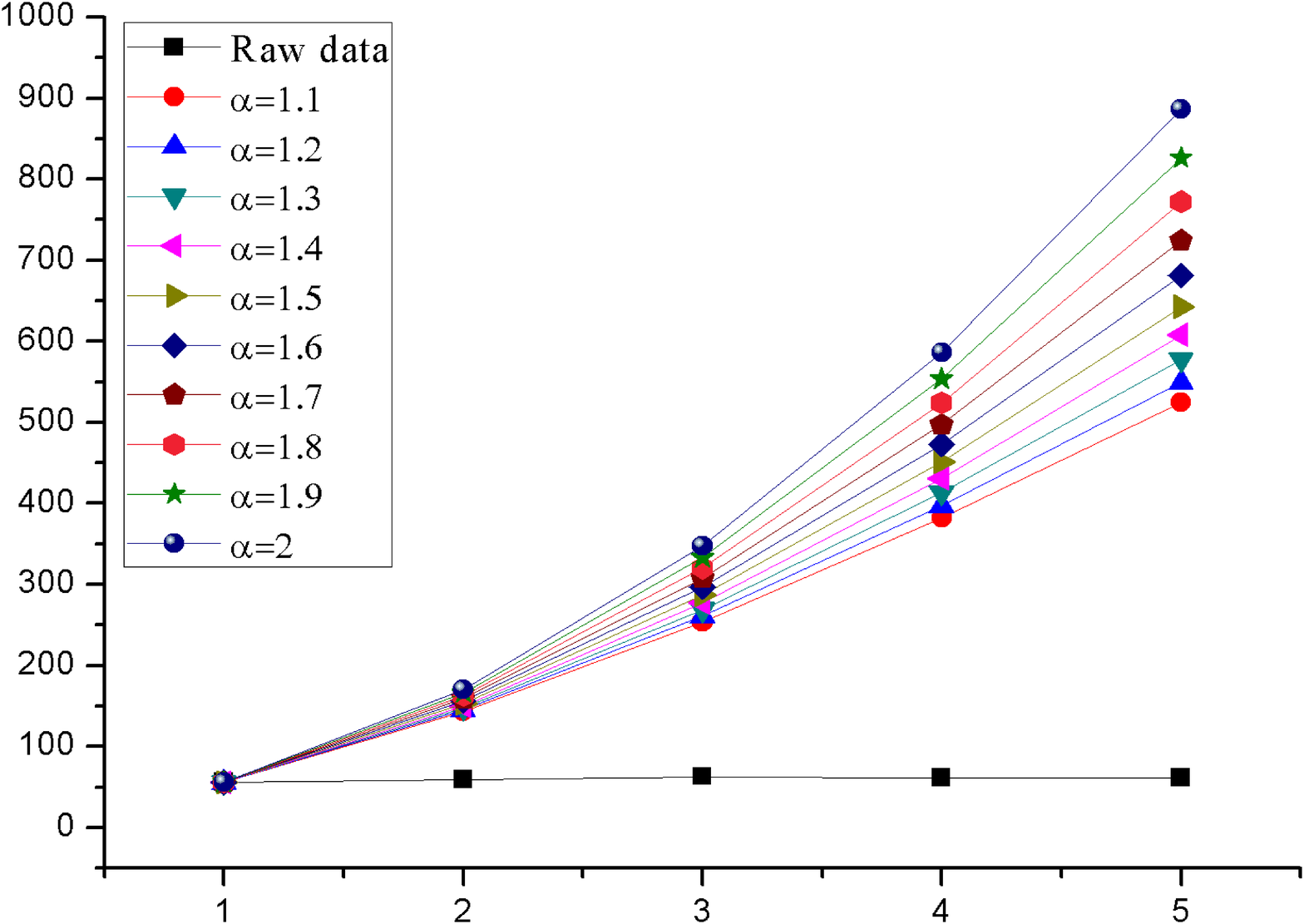}}\hfill

\caption{The CFA series with different $\alpha$.
 \label{fig:ex_cfa}}
\end{figure}

\subsection{Modelling the CFGM}

    In this step we choose the order as  $\alpha=0.59$, thus the CFA series can be obtained using the similar steps presented in Subsection \ref{sec:cfa} as $X^{(0.59)}= \Big(55.70,100.10,140.07,174.79,206.53 \Big).$

    The matrices $B$ and $Y$ in \eqref{eq:lse} can be constructed as
    \begin{equation}
    B=
        \left[
          \begin{array}{cc}
            -77.9024 	&	1	\\
            -120.0858 	&	1	\\
            -157.4283 	&	1	\\
            -190.6591 	&	1	\\
          \end{array}
        \right]
         ,
         Y=
         \left[
            \begin{array}{c}
                44.4048 	\\
                39.9621 	\\
                34.7229 	\\
                31.7388 	\\
            \end{array}
         \right].
    \end{equation}

    Then we obtain the parameters using the least squares solution \eqref{eq:lse} as
    \begin{equation}\label{eq:params}
      [\hat{a},\hat{b}]^T= (B^T B)^{-1}B^T Y = [0.1152,53.4382]^T.
    \end{equation}

    By substituting the parameters $\hat{a},\hat{b}$ into the response function \eqref{eq:resp} we have
    \begin{equation}\label{eq:ex_resp}
        \hat{x}^{(\alpha)}(k) = -408.17e^{-0.1152(k-1)} + 463.87.
    \end{equation}
    Then the restored can be obtained using \eqref{eq:ex_resp} by $k$ from 1 to 10 as

   \begin{equation}
   \hat{X}^{(0.59)}= \Big( 55.7,100.11,139.69,174.96,206.39,234.4,259.37,281.61,301.43,319.1\Big).
   \end{equation}

   Then the restored values can be obtained using the CFD in \eqref{eq:restore} as
   \begin{equation}
   \hat{X}^{(0)}= \Big(55.7,59.01,62.10,62.27,60.81,58.39,55.43,52.18,48.80,45.41 \Big).
   \end{equation}

   The predicted values of $\hat{X}^{(0.59)}$ and  $\hat{X}^{(0)}$ are plotted in Fig.\ref{fig:ex_predict}. 
\begin{figure}[!htb]

\subfloat[]{
  \includegraphics[width=0.45\textwidth]{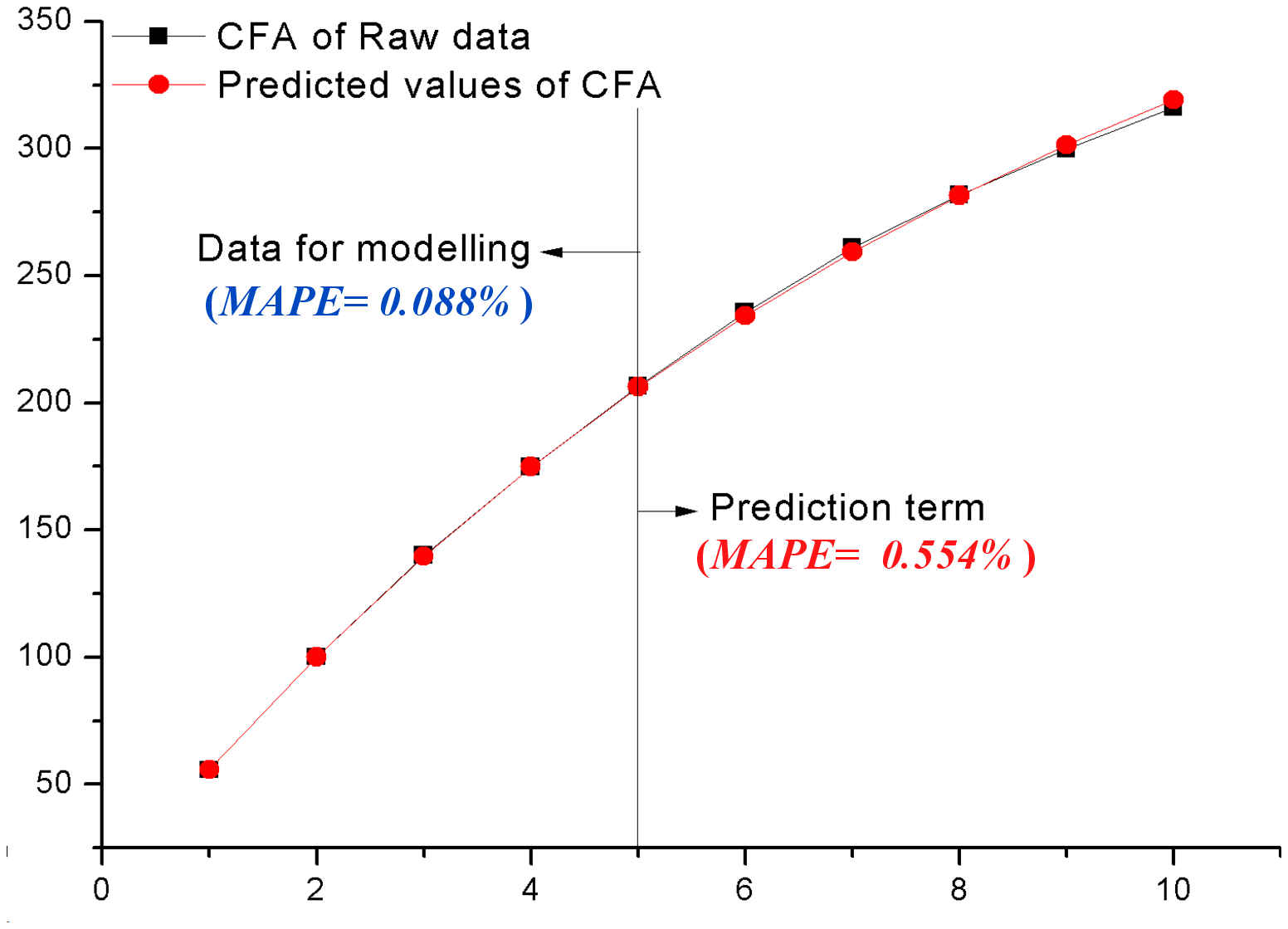}}\hfill
\subfloat[]{
  \includegraphics[width=0.45\textwidth]{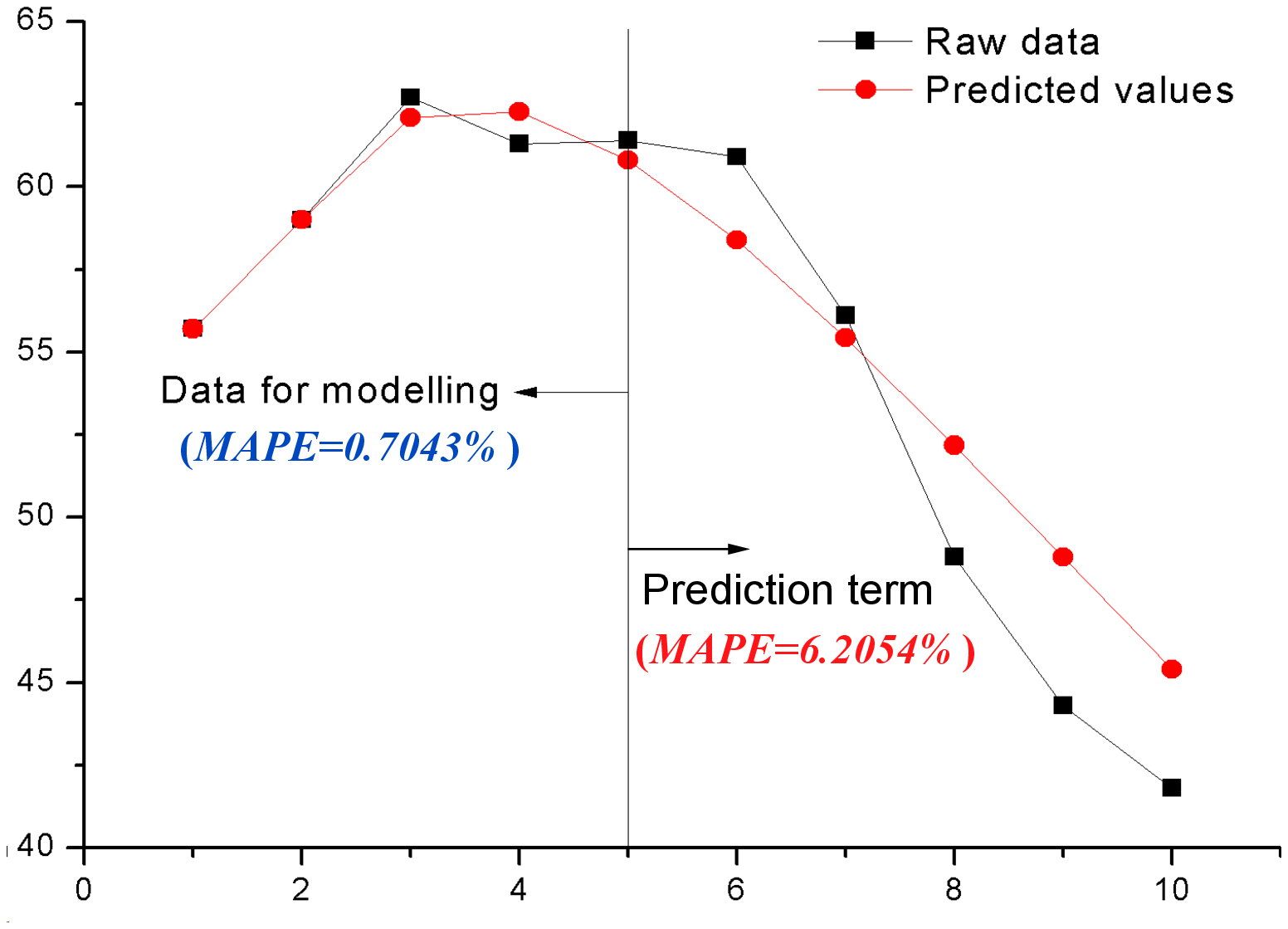}}\hfill

\caption{Predicted values by CFGM model with $\alpha=0.59$. (a) CFA series; (b)The restored values.
 \label{fig:ex_predict}}
\end{figure}

\subsection{Selecting the optimal $\alpha$}
\label{subsec:alpha}

\begin{figure}[!htb]
\begin{center}
\subfloat[]{
  \includegraphics[width=0.55\textwidth]{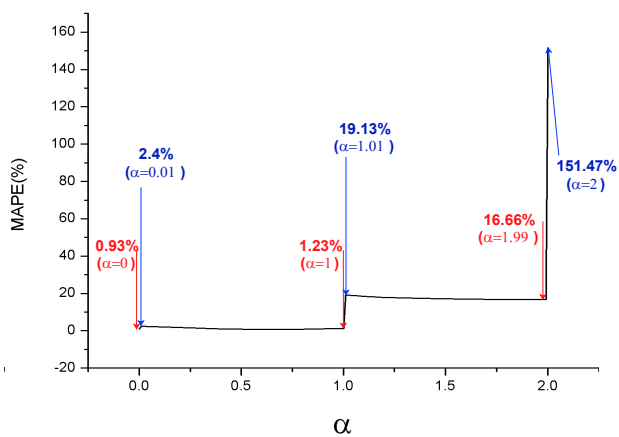}}
\end{center}
\subfloat[]{
  \includegraphics[width=0.45\textwidth]{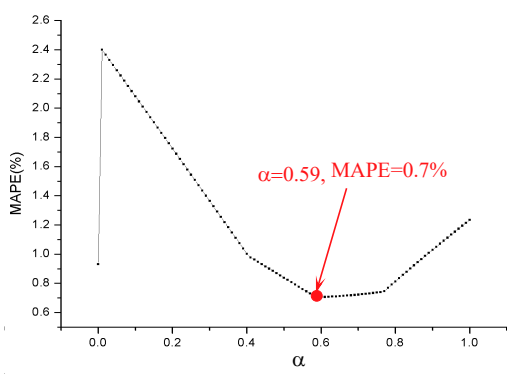}}\hfill
\subfloat[]{
  \includegraphics[width=0.45\textwidth]{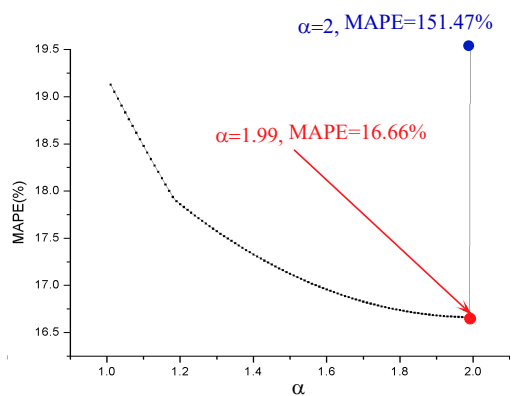}}\hfill
\caption{MAPEs of CFGM model with $\alpha$ in $[0,2]$ by step 0.01. (a) The overall picture of MAPEs with $\alpha$ in $[0,2]$; (b) The picture of MAPEs with $\alpha$ in $[0,1]$; (c) The picture of MAPEs with $\alpha$ in $[1.01,2]$.
 \label{fig:ex_mapes}}
\end{figure}

    The Brute Force strategy used in this paper is quite simple, which is actually repeating the computational steps in the above subsections, the fitting MAPEs with $\alpha$ in the interval $[0,2]$ with step $0.01$ are plotted in Fig. \ref{fig:ex_mapes}. It can be seen that the values of MAPE often jump when the $\alpha$ is near an integer, thus the subfiures in Fig. \ref{fig:ex_mapes} are also presented in order to provide a clearer picture. And in this example, the optimal $\alpha$ is obtained at $\alpha=0.59$, which is the value we used in the above subsection. And it is shown that the optimal $\alpha$ is quite easy to obtain, which indicates that Brute Force strategy is available.

\section{Validation of CFGM  with benchmark data sets}

\subsection{The benchmark data sets}

    The data sets used in this section are collected from the Time Series Data Library by Rob Hyndman \footnote{Available at the website of Time Series Data Library https://datamarket.com/data/list/?q=provider:tsdl}. In order to validate the effectiveness of CFGM  in different scenarios, we choose 21 data sets in different sizes and different background. Detailed information of the data sets used in this section is listed in Table \ref{t:datainfo}.

\begin{table}[!htb]
\centering
\tiny
\caption{Information of the data sets collected from the Time Series Data Library \label{t:datainfo}}
\begin{tabular}{p{0.1\columnwidth}p{0.6\columnwidth}p{0.1\columnwidth}}
\toprule
No	&	Data Title	&	 Instance	\\
\hline
1	&	Annual U.S. suicide rate (per 100,1000) 1920 – 1969	&	50	\\
2	&	Births per 10,000 of 23 year old women, U.S., 1917-1975	&	59	\\
3	&	Women unemployed (1000's) U.K. 1st of month Jan. ’67-July. ’72	&	67	\\
4	&	Real-estate loans (billions) (monthly: Jan.1973-Oct.1978)	&	70	\\
5	&	Annual sheep population (1000s) in England \& Wales 1867 – 1939	&	73	\\
6	&	Monthly gambling expenditure in Victoria, Australia, July 1999 to November 2006. Units are Millions of dollars per day. Smoking ban introduced in gaming venues in September 2002	&	89	\\
7	&	Highest mean monthly level, Lake Michigan, 1860 to Dec 1955	&	96	\\
8	&	Monthly logged flour price indices over the 9-years period 1972-1980 at the commodity exchanges of Buffalo, Minneapolis, and Kansas City	&	100	\\
9	&	Mean July average water surface elevation, in feet, for harbour beach,Michigan, on Lake huron, station 5014,1875 – 1972	&	127	\\
10	&	Monthly number of employed persons in Australia: thousands. Feb 1978 – Apr 1991	&	159	\\
11	&	Weekly closings of the Dow-Jones industrial average, July 1971 – August 1974	&	162	\\
12	&	Wisconsin employment time series, fabricated metals, Jan. 1961 – Oct. 1975	&	178	\\
13	&	Number of deaths and serious injuries in UK road accidents each month. Jan 1969 – Dec 1984. Seatbelt law introduced in Feb 1983 (indicator in second column)	&	192	\\
14	&	Civilian labour force in Australia each month: thousands of persons. Feb 1978 – Aug 1995	&	211	\\
15	&	Monthly U.S air passenger miles January 1960 through December 1977	&	216	\\
16	&	IBM common stock closing prices: daily, 29th June 1959 to 30th June 1960 	&	255	\\
17	&	Monthly U.S. auto registration (thousands) 1947 – 1968	&	264	\\
18	&	Exchange rate of Australian dollar: A for 1 US dollar. Monthly average: Jul 1969 - Aug 1995	&	314	\\
19	&	IBM common stock closing prices: daily, 17th May 1961 – 2nd November 1962	&	369	\\
20	&	Monthly mean water levels in meters, Lake of the wood at warroad, 1916-1965	&	600	\\
21	&	Dutch weekly survey on left-right political orientation (1978-1996), means on scale of 1-7 for (groups of parties), trend and number of respondents per party.	&	988	\\
\bottomrule
\end{tabular}
\end{table}

\subsection{Criteria for evaluating the overall performance of the prediction models}

    The modelling accuracy criteria we used in this paper are the mean squares error (MSE), mean absolute error (MAE) and mean absolute percentage error (MAPE) with standard deviation (STD) for all subcases.

    Firstly we define the error for each point as
    \begin{equation}\label{eq:error}
        \varepsilon_k =\hat{x}^{(0)}(k) - x^{(0)}(k),
    \end{equation}

    Then the MSE, MAE and MAPE with the corresponding STD are defined in Table \ref{t:criteria}.

    \begin{table}[!htb]
    \centering
    \scriptsize
    \caption{Modelling accuracy evaluation criteria \label{t:criteria}}
    \begin{tabular}{ccc}
    \toprule
    	AE 		&			MAE		&		STD		\\
    $	\varepsilon_{k}^{A}=|\varepsilon_{k}|	$	&		$	\overline{\varepsilon_{k}^{A}}=\frac{1}{N}\sum_{k} \varepsilon_{k}^{A}	$	&	$	\sqrt{\frac{1}{N}\sum_{k} \left(\varepsilon_{k}^{A}-\overline{\varepsilon_{k}^{A}}\right)^2}	$	\\
    \hline \\
    	SE		&			MSE		&		STD		\\
    $	\varepsilon_{k}^{S}=\varepsilon_{k}^2	$	&		$	\overline{\varepsilon_{k}^{S}}=\frac{1}{N}\sum_{k} \varepsilon_{k}^{S}	$	&	$	\sqrt{\frac{1}{N}\sum_{k} \left(\varepsilon_{k}^{S}-\overline{\varepsilon_{k}^{S}}\right)^2}	$	\\
    \hline \\
    	APE		&			MAPE		&		STD		\\
    $	\varepsilon_{k}^{P}=\left|\frac{\varepsilon_{k}}{x_{i}^{(0)}(j)}\right|\times 100\	$	&		$	\overline{\varepsilon_{k}^{P}}=\frac{1}{N}\sum_{k} \varepsilon_{k}^{P}	$	&	$	\sqrt{\frac{1}{N}\sum_{k} \left(\varepsilon_{k}^{P}-\overline{\varepsilon_{k}^{P}}\right)^2}	$	\\
    \bottomrule
    \end{tabular}
    \end{table}

\subsection{Validation scheme and results}

    In this section we use the classical and widely adopted 1-step, 2-step and 3-step prediction test to validate the performance of CFGM. The results are also compared to the existing FGM model and the autoregressive model (AR).  As described in \cite{yang2019}, the 1-step prediction, the predicted value $\hat{x}^{(0)}(k)$ is obtained according to the previous points $\left( x^{(0)}(k-1), x^{(0)}(k-2),...,x^{(0)}(k-\tau)\right)$, where $\tau$ is the input number which is set as 5 in this section. Then the 2-step and 3-step predictions are executed in the similar way.

    Results of 1-step, 2-step and 3-step prediction based on different data sets are listed in Table \ref{t:valiresults}, numerics in bold stand for the best results. It is clearly shown that all the criteria of CFGM are smaller than FGM and AR, which indicates the CFGM has the highest accuracy in 1-step, 2-step and 3-step predictions and also has the highest stability.

    A summary of the results is also listed in Table \ref{t:rsum}, which presents the maximum values and average values of the criteria MAE, MSE and MAPE. The maximum values of the indicators just show the worst case of the models, and average values represent the overall performance of the models in all the cases. The results shown in Table \ref{t:rsum} still coincides that in Table \ref{t:valiresults}, which indicates the performance of CFGM is significantly better than FGM and AR. It is interesting to see that in most cases, the AR has the worst performance. Noticing that the sample sizes of the subcases are all very small, it is clear that the CFGM and FGM are better in small sample forecasting.

    As shown in Table \ref{t:rsum}, errors of CFGM and FGM all increase with more steps of prediction. But such increase of CFGM is much slower. For one more step prediction, the MAPE of CFGM only increases 0.4\%, and MAE increases no larger than 3, and the MSE always stays in the same order of magnitudes; but for FGM, the MAPE increases more than two times, the MAE increases more than 14, and MSE increases almost ten times. In the worst case, the maximum errors of the models increases faster. It can be seen that the increase of maximum errors by CFGM is still acceptable, and in the 3-step prediction, the maximum MAPE of CFGM is 19.17\%, which is still acceptable. But for the FGM, the maximum errors increases very fast. In 2-step and 3-step predictions, the maximum MAPE by FGM are as large as 117.5847\% and 465.9076\% , respectively, which are not acceptable at all.
    Above all, it is clear that the CFGM outperforms the existing FGM and AR model in the validation with the 21 benchmark data sets.

\begin{table}[htbp]
\caption{Validation results of CFGM, FGM and AR with the benchmark data sets\label{t:valiresults}}
\tiny
\centering
\setlength{\tabcolsep}{0.4mm}{

\resizebox{\textwidth}{!}{
	\begin{tabular}{ccccccccccccc}

	\toprule
	&		&		CFGM						&	FGM			&	AR			&		CFGM						&	FGM			&	AR			&		CFGM						&	FGM			&	AR				\\
Number	&	Ceriteria	&		1-Step						&	1-Step			&	1-Step			&		2-Step						&	2-Step			&	2-Step			&		3-Step						&	3-Step			&	3-Step				\\
\hline
1	&	MSE	&	\bf{	89.44	}	$\pm$	\bf{	153.26	}	&	159.15	$\pm$	561.87	&	129.55	$\pm$	285	&	\bf{	186.03	}	$\pm$	\bf{	280.34	}	&	1006.03	$\pm$	5204.3	&	8586.96	$\pm$	5110.35	&	\bf{	308.45	}	$\pm$	\bf{	424.26	}	&	5982.84	$\pm$	33838.66	&	5725.16	$\pm$	3405.96		\\
 	&	MAE	&	\bf{	7.17	}	$\pm$	\bf{	6.17	}	&	7.6	$\pm$	10.07	&	7.62	$\pm$	8.45	&	\bf{	10.1	}	$\pm$	\bf{	7.58	}	&	12.9	$\pm$	21.69	&	214.4	$\pm$	110.62	&	\bf{	12.36	}	$\pm$	\bf{	8.52	}	&	20.85	$\pm$	44.4	&	143.01	$\pm$	73.59		\\
 	&	MAPE	&	\bf{	5.87	}	$\pm$	\bf{	5.24	}	&	6.38	$\pm$	9.54	&	6.26	$\pm$	7.6	&	\bf{	8.24	}	$\pm$	\bf{	6.19	}	&	10.75	$\pm$	19.2	&	165.64	$\pm$	91.79	&	\bf{	9.99	}	$\pm$	\bf{	6.44	}	&	17.43	$\pm$	38.61	&	110.48	$\pm$	61.08		\\
 																																																				\\
2	&	MSE	&	\bf{	225.98	}	$\pm$	\bf{	612.43	}	&	338.36	$\pm$	1004.72	&	1325.28	$\pm$	7724.56	&	\bf{	468.29	}	$\pm$	\bf{	1006.44	}	&	1084.19	$\pm$	2655.89	&	1.6E+05	$\pm$	1.3E+05	&	\bf{	716.39	}	$\pm$	\bf{	1294.86	}	&	2977.16	$\pm$	8056.03	&	104258.27	$\pm$	84914.65		\\
 	&	MAE	&	\bf{	9.79	}	$\pm$	\bf{	11.41	}	&	11.36	$\pm$	14.47	&	15.6	$\pm$	32.89	&	\bf{	14.17	}	$\pm$	\bf{	13.75	}	&	18.85	$\pm$	20.06	&	460.57	$\pm$	332.48	&	\bf{	17.54	}	$\pm$	\bf{	14.82	}	&	26.86	$\pm$	28.84	&	307.66	$\pm$	222.34		\\
 	&	MAPE	&	\bf{	5.45	}	$\pm$	\bf{	6.09	}	&	6.4	$\pm$	7.67	&	8.98	$\pm$	18.91	&	\bf{	7.91	}	$\pm$	\bf{	7.28	}	&	10.69	$\pm$	11.22	&	281.24	$\pm$	198.14	&	\bf{	9.75	}	$\pm$	\bf{	7.76	}	&	15.36	$\pm$	17.38	&	188.01	$\pm$	132.75		\\
 																																																				\\
3	&	MSE	&	\bf{	39.99	}	$\pm$	\bf{	75.39	}	&	46.5	$\pm$	114.14	&	156.52	$\pm$	511.49	&	\bf{	81.85	}	$\pm$	\bf{	99.21	}	&	97.24	$\pm$	140.96	&	6630.15	$\pm$	3184.94	&	\bf{	143.22	}	$\pm$	\bf{	161.84	}	&	199.8	$\pm$	267.6	&	4420.08	$\pm$	2115.64		\\
 	&	MAE	&	\bf{	4.88	}	$\pm$	\bf{	4.02	}	&	5.01	$\pm$	4.63	&	7.59	$\pm$	9.94	&	\bf{	6.98	}	$\pm$	\bf{	4.26	}	&	7.2	$\pm$	4.66	&	203.5	$\pm$	104.37	&	\bf{	8.92	}	$\pm$	\bf{	4.8	}	&	9.76	$\pm$	5.42	&	135.75	$\pm$	69.24		\\
 	&	MAPE	&	\bf{	5.35	}	$\pm$	\bf{	4.07	}	&	5.4	$\pm$	4.45	&	8.14	$\pm$	9.65	&	\bf{	7.69	}	$\pm$	\bf{	4.59	}	&	7.87	$\pm$	4.97	&	235.99	$\pm$	121.51	&	\bf{	9.84	}	$\pm$	\bf{	5.24	}	&	10.71	$\pm$	5.91	&	157.43	$\pm$	80.64		\\
 																																																				\\
4	&	MSE	&	\bf{	0.11	}	$\pm$	\bf{	0.15	}	&	0.14	$\pm$	0.22	&	0.55	$\pm$	3.37	&	\bf{	0.34	}	$\pm$	\bf{	0.56	}	&	0.46	$\pm$	1.27	&	55.66	$\pm$	56.21	&	\bf{	0.81	}	$\pm$	\bf{	1.72	}	&	1.22	$\pm$	4.49	&	37.17	$\pm$	37.56		\\
	&	MAE	&	\bf{	0.27	}	$\pm$	\bf{	0.2	}	&	0.29	$\pm$	0.24	&	0.37	$\pm$	0.65	&	\bf{	0.42	}	$\pm$	\bf{	0.3	}	&	0.45	$\pm$	0.38	&	11.51	$\pm$	8.41	&	\bf{	0.58	}	$\pm$	\bf{	0.41	}	&	0.62	$\pm$	0.56	&	7.69	$\pm$	5.63		\\
 	&	MAPE	&	\bf{	0.42	}	$\pm$	\bf{	0.3	}	&	0.45	$\pm$	0.34	&	0.58	$\pm$	1.06	&	\bf{	0.64	}	$\pm$	\bf{	0.42	}	&	0.68	$\pm$	0.49	&	19.04	$\pm$	13.36	&	\bf{	0.9	}	$\pm$	\bf{	0.55	}	&	0.92	$\pm$	0.67	&	12.7	$\pm$	8.92		\\
 																																																				\\
5	&	MSE	&	\bf{	1.2E+04	}	$\pm$	\bf{	1.3E+04	}	&	1.3E+04	$\pm$	2.1E+04	&	1.4E+05	$\pm$	8.5E+05	&	\bf{	2.8E+04	}	$\pm$	\bf{	2.8E+04	}	&	4.8E+04	$\pm$	1.1E+05	&	5.6E+06	$\pm$	2.6E+06	&	\bf{	5.1E+04	}	$\pm$	\bf{	4.3E+04	}	&	1.5E+05	$\pm$	4.9E+05	&	3.8E+06	$\pm$	1.8E+06		\\
 	&	MAE	&	\bf{	88.33	}	$\pm$	\bf{	62.27	}	&	89.07	$\pm$	71.38	&	165.31	$\pm$	328.96	&	\bf{	131.49	}	$\pm$	\bf{	73.56	}	&	145.9	$\pm$	109.84	&	4872.77	$\pm$	2714.81	&	\bf{	175.15	}	$\pm$	\bf{	79.96	}	&	214.6	$\pm$	167.51	&	3249.42	$\pm$	1808.43		\\
 	&	MAPE	&	\bf{	4.84	}	$\pm$	\bf{	3.46	}	&	4.93	$\pm$	4.27	&	9.27	$\pm$	18.46	&	\bf{	7.22	}	$\pm$	\bf{	3.98	}	&	8.11	$\pm$	6.52	&	269.4	$\pm$	156.28	&	\bf{	9.64	}	$\pm$	\bf{	4.32	}	&	12.05	$\pm$	10.29	&	179.64	$\pm$	104.12		\\
 	&																																																			\\
6	&	MSE	&	\bf{	0.14	}	$\pm$	\bf{	0.32	}	&	0.17	$\pm$	0.54	&	0.47	$\pm$	2.03	&	\bf{	0.23	}	$\pm$	\bf{	0.36	}	&	0.31	$\pm$	0.62	&	50.85	$\pm$	28.57	&	\bf{	0.33	}	$\pm$	\bf{	0.42	}	&	0.51	$\pm$	0.74	&	33.9	$\pm$	19.05		\\
 	&	MAE	&	\bf{	0.29	}	$\pm$	\bf{	0.24	}	&	0.31	$\pm$	0.27	&	0.35	$\pm$	0.59	&	\bf{	0.35	}	$\pm$	\bf{	0.24	}	&	0.4	$\pm$	0.26	&	13.99	$\pm$	7.16	&	\bf{	0.42	}	$\pm$	\bf{	0.24	}	&	0.51	$\pm$	0.26	&	9.33	$\pm$	4.77		\\
 	&	MAPE	&	\bf{	4.44	}	$\pm$	\bf{	3.73	}	&	4.75	$\pm$	4.35	&	5.46	$\pm$	9.65	&	\bf{	5.43	}	$\pm$	\bf{	3.7	}	&	6.18	$\pm$	4.06	&	219.24	$\pm$	112.66	&	\bf{	6.51	}	$\pm$	\bf{	3.73	}	&	7.74	$\pm$	3.98	&	146.2	$\pm$	75.05		\\
 																																																				\\
7	&	MSE	&	\bf{	1.03	}	$\pm$	\bf{	1.42	}	&	1.09	$\pm$	1.59	&	58.88	$\pm$	548.86	&	\bf{	2.27	}	$\pm$	\bf{	3.4	}	&	2.77	$\pm$	3.96	&	2432.55	$\pm$	9811.99	&	\bf{	3.78	}	$\pm$	\bf{	5.08	}	&	5.2	$\pm$	6.46	&	1621.71	$\pm$	6541.35		\\
 	&	MAE	&	\bf{	0.8	}	$\pm$	\bf{	0.63	}	&	0.82	$\pm$	0.64	&	1.52	$\pm$	7.52	&	\bf{	1.11	}	$\pm$	\bf{	0.79	}	&	1.2	$\pm$	0.81	&	27.72	$\pm$	42.08	&	\bf{	1.39	}	$\pm$	\bf{	0.89	}	&	1.59	$\pm$	0.98	&	18.49	$\pm$	28.06		\\
 	&	MAPE	&	\bf{	0.98	}	$\pm$	\bf{	0.78	}	&	1.01	$\pm$	0.79	&	1.89	$\pm$	9.4	&	\bf{	1.38	}	$\pm$	\bf{	0.99	}	&	1.49	$\pm$	1.01	&	34.12	$\pm$	51.78	&	\bf{	1.72	}	$\pm$	\bf{	1.11	}	&	1.96	$\pm$	1.23	&	22.76	$\pm$	34.52		\\
 																																																				\\
8	&	MSE	&	\bf{	0	}	$\pm$	\bf{	0.01	}	&	0	$\pm$	0.01	&	0.01	$\pm$	0.06	&	\bf{	0.01	}	$\pm$	\bf{	0.01	}	&	0.01	$\pm$	0.01	&	2.7	$\pm$	1.73	&	\bf{	0.02	}	$\pm$	\bf{	0.02	}	&	0.02	$\pm$	0.02	&	1.8	$\pm$	1.15		\\
 	&	MAE	&	\bf{	0.05	}	$\pm$	\bf{	0.04	}	&	0.05	$\pm$	0.04	&	0.06	$\pm$	0.1	&	\bf{	0.07	}	$\pm$	\bf{	0.04	}	&	0.08	$\pm$	0.05	&	3.21	$\pm$	1.81	&	\bf{	0.09	}	$\pm$	\bf{	0.05	}	&	0.1	$\pm$	0.06	&	2.14	$\pm$	1.21		\\
 	&	MAPE	&	\bf{	1.02	}	$\pm$	\bf{	0.8	}	&	1.06	$\pm$	0.81	&	1.15	$\pm$	1.94	&	\bf{	1.43	}	$\pm$	\bf{	0.86	}	&	1.51	$\pm$	0.89	&	63.76	$\pm$	36.27	&	\bf{	1.85	}	$\pm$	\bf{	1.01	}	&	2.02	$\pm$	1.09	&	42.51	$\pm$	24.17		\\
																																																				\\
9	&	MSE	&	\bf{	1.11	}	$\pm$	\bf{	1.68	}	&	1.23	$\pm$	1.97	&	62.9	$\pm$	665.48	&	\bf{	2.41	}	$\pm$	\bf{	3.74	}	&	3.04	$\pm$	4.65	&	1.8E+04	$\pm$	2.7E+04	&	\bf{	4.11	}	$\pm$	\bf{	5.58	}	&	5.69	$\pm$	7.59	&	1.2E+04	$\pm$	1.8E+04		\\
 	&	MAE	&	\bf{	0.82	}	$\pm$	\bf{	0.67	}	&	0.84	$\pm$	0.72	&	1.58	$\pm$	7.77	&	\bf{	1.15	}	$\pm$	\bf{	0.81	}	&	1.24	$\pm$	0.89	&	83.5	$\pm$	99.4	&	\bf{	1.47	}	$\pm$	\bf{	0.89	}	&	1.66	$\pm$	1.02	&	55.68	$\pm$	66.26		\\
 	&	MAPE	&	\bf{	0.14	}	$\pm$	\bf{	0.11	}	&	0.15	$\pm$	0.12	&	0.27	$\pm$	1.34	&	\bf{	0.2	}	$\pm$	\bf{	0.14	}	&	0.21	$\pm$	0.15	&	14.41	$\pm$	17.15	&	\bf{	0.25	}	$\pm$	\bf{	0.15	}	&	0.29	$\pm$	0.18	&	9.61	$\pm$	11.43		\\
 																																																				\\
10	&	MSE	&	\bf{	1.2E+04	}	$\pm$	\bf{	16629.4	}	&	1.9E+04	$\pm$	2.6E+04	&	2.0E+05	$\pm$	1.7E+06	&	\bf{	1.7E+04	}	$\pm$	\bf{	2.5E+04	}	&	3.1E+04	$\pm$	3.6E+04	&	4.1E+07	$\pm$	4.3E+07	&	\bf{	2.5E+04	}	$\pm$	\bf{	3.3E+04	}	&	5.0E+04	$\pm$	5.8E+04	&	2.7E+07	$\pm$	2.9E+07		\\
	&	MAE	&	\bf{	83.02	}	$\pm$	\bf{	69.31	}	&	109.82	$\pm$	85.79	&	129.23	$\pm$	426.39	&	\bf{	90.46	}	$\pm$	\bf{	67.52	}	&	128.84	$\pm$	76.76	&	7447.33	$\pm$	5628.52	&	\bf{	107.97	}	$\pm$	\bf{	64.68	}	&	156.72	$\pm$	86.01	&	4973.25	$\pm$	3772.87		\\
	&	MAPE	&	\bf{	1.21	}	$\pm$	\bf{	0.99	}	&	1.61	$\pm$	1.26	&	1.84	$\pm$	5.83	&	\bf{	1.32	}	$\pm$	\bf{	0.96	}	&	1.89	$\pm$	1.11	&	112	$\pm$	80.41	&	\bf{	1.57	}	$\pm$	\bf{	0.9	}	&	2.29	$\pm$	1.24	&	74.77	$\pm$	53.86		\\
																																																				\\
11	&	MSE	&	\bf{	740.12	}	$\pm$	\bf{	1202.66	}	&	919.22	$\pm$	2381.99	&	4181.88	$\pm$	3.5E+04	&	\bf{	1355.23	}	$\pm$	\bf{	1.7E+03	}	&	2.4E+03	$\pm$	1.1E+04	&	1.9E+06	$\pm$	9.5E+05	&	\bf{	2.3E+03	}	$\pm$	\bf{	2.9E+03	}	&	7.6E+03	$\pm$	5.3E+04	&	1.3E+06	$\pm$	6.3E+05		\\
	&	MAE	&	\bf{	20.92	}	$\pm$	\bf{	17.39	}	&	22.14	$\pm$	20.71	&	27.42	$\pm$	58.56	&	\bf{	27.3	}	$\pm$	\bf{	17.65	}	&	31.06	$\pm$	25.73	&	2384.16	$\pm$	1034.44	&	\bf{	34.64	}	$\pm$	\bf{	19.66	}	&	42.09	$\pm$	41.84	&	1589.55	$\pm$	689.6		\\
	&	MAPE	&	\bf{	2.32	}	$\pm$	\bf{	1.95	}	&	2.46	$\pm$	2.35	&	3.03	$\pm$	6.18	&	\bf{	3.03	}	$\pm$	\bf{	2.02	}	&	3.46	$\pm$	2.98	&	258.35	$\pm$	111.11	&	\bf{	3.85	}	$\pm$	\bf{	2.24	}	&	4.69	$\pm$	4.76	&	172.25	$\pm$	74.07		\\
																																																				\\
12	&	MSE	&	\bf{	0.74	}	$\pm$	\bf{	1.74	}	&	0.91	$\pm$	1.97	&	272.65	$\pm$	3348.71	&	\bf{	1.44	}	$\pm$	\bf{	2.65	}	&	1.87	$\pm$	3.21	&	5.5E+05	$\pm$	4.7E+05	&	\bf{	2.45	}	$\pm$	\bf{	3.93	}	&	3.41	$\pm$	5.19	&	3.7E+05	$\pm$	3.1E+05		\\
	&	MAE	&	\bf{	0.61	}	$\pm$	\bf{	0.6	}	&	0.67	$\pm$	0.68	&	2.52	$\pm$	16.32	&	\bf{	0.84	}	$\pm$	\bf{	0.64	}	&	0.93	$\pm$	0.7	&	551.77	$\pm$	426.63	&	\bf{	1.07	}	$\pm$	\bf{	0.71	}	&	1.22	$\pm$	0.81	&	367.85	$\pm$	284.42		\\
	&	MAPE	&	\bf{	1.51	}	$\pm$	\bf{	1.5	}	&	1.65	$\pm$	1.7	&	6.19	$\pm$	39.83	&	\bf{	2.05	}	$\pm$	\bf{	1.59	}	&	2.28	$\pm$	1.74	&	1339.99	$\pm$	1020.67	&	\bf{	2.62	}	$\pm$	\bf{	1.74	}	&	2.99	$\pm$	2.03	&	893.35	$\pm$	680.43		\\
																																																				\\
13	&	MSE	&	\bf{	9.3E+04	}	$\pm$	\bf{	2.0E+05	}	&	1.2E+05	$\pm$	2.4E+05	&	3.2E+09	$\pm$	4.3E+10	&	\bf{	1.8E+05	}	$\pm$	\bf{	3.1E+05	}	&	2.8E+05	$\pm$	5.6E+05	&	9.1E+12	$\pm$	1.1E+13	&	\bf{	3.1E+05	}	$\pm$	\bf{	6.2E+05	}	&	6.1E+05	$\pm$	1.5E+06	&	6.1E+12	$\pm$	7.6E+12		\\
	&	MAE	&	\bf{	223.26	}	$\pm$	\bf{	207.61	}	&	248.43	$\pm$	233.19	&	5134.13	$\pm$	5.7E+04	&	\bf{	294.55	}	$\pm$	\bf{	223.51	}	&	346.81	$\pm$	293.62	&	1.6E+06	$\pm$	1.8E+06	&	\bf{	370.55	}	$\pm$	\bf{	246.21	}	&	463.47	$\pm$	380.01	&	1.1E+06	$\pm$	1.2E+06		\\
	&	MAPE	&	\bf{	13.55	}	$\pm$	\bf{	12.75	}	&	15.11	$\pm$	14.42	&	247.56	$\pm$	2529.3	&	\bf{	18.01	}	$\pm$	\bf{	14.54	}	&	21.42	$\pm$	19.03	&	9.1E+04	$\pm$	1.0E+05	&	\bf{	22.76	}	$\pm$	\bf{	16.81	}	&	29.04	$\pm$	25.53	&	6.1E+04	$\pm$	6.8E+04		\\
																																																				\\
14	&	MSE	&	\bf{	2.0E+04	}	$\pm$	\bf{	3.4E+04	}	&	3.4E+04	$\pm$	6.0E+04	&	1.4E+07	$\pm$	2.0E+08	&	\bf{	2.5E+04	}	$\pm$	\bf{	3.3E+04	}	&	4.9E+04	$\pm$	6.2E+04	&	4.0E+10	$\pm$	7.4E+09	&	\bf{	4.1E+04	}	$\pm$	\bf{	5.4E+04	}	&	8.2E+04	$\pm$	1.0E+05	&	2.6E+10	$\pm$	4.9E+09		\\
	&	MAE	&	\bf{	107.34	}	$\pm$	\bf{	94.66	}	&	140.74	$\pm$	121.05	&	374.38	$\pm$	3724.95	&	\bf{	119.13	}	$\pm$	\bf{	78.32	}	&	169.88	$\pm$	101.02	&	1.5E+05	$\pm$	3.1E+04	&	\bf{	144.66	}	$\pm$	\bf{	79.55	}	&	206.97	$\pm$	115.2	&	1.0E+05	$\pm$	2.0E+04		\\
	&	MAPE	&	\bf{	1.38	}	$\pm$	\bf{	1.17	}	&	1.82	$\pm$	1.53	&	5.54	$\pm$	57.96	&	\bf{	1.54	}	$\pm$	\bf{	0.97	}	&	2.19	$\pm$	1.25	&	2344.64	$\pm$	457.75	&	\bf{	1.87	}	$\pm$	\bf{	0.97	}	&	2.66	$\pm$	1.41	&	1563.1	$\pm$	305.16		\\
																																																				\\
15	&	MSE	&	\bf{	2.48	}	$\pm$	\bf{	5.07	}	&	2.69	$\pm$	5.85	&	7.19	$\pm$	67.52	&	\bf{	4.69	}	$\pm$	\bf{	8.35	}	&	7.35	$\pm$	37.95	&	1254.51	$\pm$	1384.48	&	\bf{	7.98	}	$\pm$	\bf{	14.27	}	&	28.83	$\pm$	269.03	&	836.33	$\pm$	922.97		\\
	&	MAE	&	\bf{	1.06	}	$\pm$	\bf{	1.17	}	&	1.12	$\pm$	1.2	&	1.23	$\pm$	2.38	&	\bf{	1.42	}	$\pm$	\bf{	1.32	}	&	1.54	$\pm$	1.79	&	64	$\pm$	61.88	&	\bf{	1.82	}	$\pm$	\bf{	1.48	}	&	2.14	$\pm$	3.08	&	42.66	$\pm$	41.25		\\
	&	MAPE	&	\bf{	12.33	}	$\pm$	\bf{	10.68	}	&	13.55	$\pm$	11.01	&	13.7	$\pm$	25.65	&	\bf{	15.91	}	$\pm$	\bf{	10.68	}	&	17.74	$\pm$	13.89	&	877.19	$\pm$	628	&	\bf{	20.34	}	$\pm$	\bf{	12.36	}	&	23.83	$\pm$	21.34	&	584.8	$\pm$	418.58		\\
																																																				\\
16	&	MSE	&	\bf{	48.02	}	$\pm$	\bf{	96.55	}	&	58.38	$\pm$	114.45	&	2725.29	$\pm$	2.8E+04	&	\bf{	91.86	}	$\pm$	\bf{	161.29	}	&	123.2	$\pm$	185.72	&	3.2E+06	$\pm$	1.2E+06	&	\bf{	149.04	}	$\pm$	\bf{	211.49	}	&	216.51	$\pm$	283.4	&	2.14E+06	$\pm$	7.8E+05		\\
	&	MAE	&	\bf{	5.03	}	$\pm$	\bf{	4.77	}	&	5.56	$\pm$	5.24	&	11.24	$\pm$	50.98	&	\bf{	6.71	}	$\pm$	\bf{	5.35	}	&	7.81	$\pm$	5.52	&	2380.5	$\pm$	967.11	&	\bf{	8.46	}	$\pm$	\bf{	5.54	}	&	10.14	$\pm$	6.3	&	1587.01	$\pm$	644.71		\\
	&	MAPE	&	\bf{	1.14	}	$\pm$	\bf{	1.05	}	&	1.26	$\pm$	1.15	&	2.6	$\pm$	12.29	&	\bf{	1.52	}	$\pm$	\bf{	1.18	}	&	1.77	$\pm$	1.21	&	557.55	$\pm$	225.98	&	\bf{	1.92	}	$\pm$	\bf{	1.21	}	&	2.3	$\pm$	1.38	&	371.7	$\pm$	150.64		\\
17																																																				\\
	&	MSE	&	\bf{	1.6E+04	}	$\pm$	\bf{	2.8E+04	}	&	3.8E+04	$\pm$	3.0E+05	&	3.9E+05	$\pm$	4.0E+06	&	\bf{	3.1E+04	}	$\pm$	\bf{	5.5E+04	}	&	6.8E+05	$\pm$	9.9E+06	&	3.0E+08	$\pm$	2.3E+08	&	\bf{	5.5E+04	}	$\pm$	\bf{	1.0E+05	}	&	2.6E+07	$\pm$	4.2E+08	&	2.0E+08	$\pm$	1.5E+08		\\
	&	MAE	&	\bf{	89.93	}	$\pm$	\bf{	88.04	}	&	106.95	$\pm$	162.59	&	149.52	$\pm$	609.73	&	\bf{	120.03	}	$\pm$	\bf{	100.26	}	&	184.55	$\pm$	564.76	&	2.1E+04	$\pm$	1.6E+04	&	\bf{	146.86	}	$\pm$	\bf{	116.97	}	&	415.29	$\pm$	2952.63	&	14106.72	$\pm$	10452.27		\\
	&	MAPE	&	\bf{	16.4	}	$\pm$	\bf{	14.82	}	&	19.92	$\pm$	34.16	&	30.08	$\pm$	137.13	&	\bf{	21.88	}	$\pm$	\bf{	16.3	}	&	35.09	$\pm$	117.58	&	4654.61	$\pm$	3357.34	&	\bf{	26.72	}	$\pm$	\bf{	19.17	}	&	71.85	$\pm$	465.91	&	3103.12	$\pm$	2238.23		\\
																																																				\\
18	&	MSE	&	\bf{	0	}	$\pm$	\bf{	0	}	&	0	$\pm$	0.01	&	887.66	$\pm$	1.6E+04	&	\bf{	0	}	$\pm$	\bf{	0.01	}	&	0.01	$\pm$	0.06	&	1.5E+07	$\pm$	9.0E+06	&	\bf{	0	}	$\pm$	\bf{	0.01	}	&	0.03	$\pm$	0.38	&	9.7E+06	$\pm$	6.0E+06		\\
	&	MAE	&	\bf{	0.02	}	$\pm$	\bf{	0.03	}	&	0.02	$\pm$	0.03	&	1.75	$\pm$	29.74	&	\bf{	0.03	}	$\pm$	\bf{	0.03	}	&	0.03	$\pm$	0.05	&	2322.12	$\pm$	1418.39	&	\bf{	0.04	}	$\pm$	\bf{	0.04	}	&	0.05	$\pm$	0.09	&	1548.08	$\pm$	945.59		\\
	&	MAPE	&	\bf{	2.52	}	$\pm$	\bf{	2.98	}	&	2.57	$\pm$	3.18	&	161.41	$\pm$	2737.65	&	\bf{	3.43	}	$\pm$	\bf{	3.53	}	&	3.63	$\pm$	4.33	&	2.1E+05	$\pm$	1.3E+05	&	\bf{	4.39	}	$\pm$	\bf{	4.15	}	&	4.93	$\pm$	7.43	&	142385.27	$\pm$	87011.59		\\
																																																				\\
19	&	MSE	&	\bf{	89.22	}	$\pm$	\bf{	173.51	}	&	112.17	$\pm$	303.96	&	1.1E+04	$\pm$	1.8E+05	&	\bf{	186.99	}	$\pm$	\bf{	336.55	}	&	290.94	$\pm$	1186.61	&	6.4E+06	$\pm$	1.1E+07	&	\bf{	360.26	}	$\pm$	\bf{	794.27	}	&	704.89	$\pm$	4695.08	&	4.2E+06	$\pm$	7.5E+06		\\
	&	MAE	&	\bf{	6.95	}	$\pm$	\bf{	6.4	}	&	7.38	$\pm$	7.59	&	17.78	$\pm$	105.15	&	\bf{	9.52	}	$\pm$	\bf{	7.57	}	&	10.49	$\pm$	9.56	&	3362.82	$\pm$	2745.55	&	\bf{	12.41	}	$\pm$	\bf{	9.05	}	&	14.21	$\pm$	13.19	&	2241.89	$\pm$	1830.37		\\
	&	MAPE	&	\bf{	1.59	}	$\pm$	\bf{	1.66	}	&	1.71	$\pm$	2.02	&	4.02	$\pm$	26.03	&	\bf{	2.15	}	$\pm$	\bf{	1.94	}	&	2.41	$\pm$	2.56	&	695.02	$\pm$	630.61	&	\bf{	2.79	}	$\pm$	\bf{	2.31	}	&	3.25	$\pm$	3.56	&	463.35	$\pm$	420.41		\\
																																																				\\
20	&	MSE	&	\bf{	0.03	}	$\pm$	\bf{	0.07	}	&	0.04	$\pm$	0.08	&	59.93	$\pm$	1229.99	&	\bf{	0.08	}	$\pm$	\bf{	0.16	}	&	0.1	$\pm$	0.21	&	2.4E+05	$\pm$	2.0E+05	&	\bf{	0.17	}	$\pm$	\bf{	0.41	}	&	0.22	$\pm$	0.5	&	1.6E+05	$\pm$	1.3E+05		\\
 	&	MAE	&	\bf{	0.12	}	$\pm$	\bf{	0.13	}	&	0.13	$\pm$	0.14	&	0.76	$\pm$	7.7	&	\bf{	0.18	}	$\pm$	\bf{	0.17	}	&	0.2	$\pm$	0.19	&	492.23	$\pm$	309.77	&	\bf{	0.25	}	$\pm$	\bf{	0.22	}	&	0.28	$\pm$	0.26	&	328.23	$\pm$	206.42		\\
 	&	MAPE	&	\bf{	0.04	}	$\pm$	\bf{	0.04	}	&	0.04	$\pm$	0.04	&	0.23	$\pm$	2.38	&	\bf{	0.06	}	$\pm$	\bf{	0.05	}	&	0.06	$\pm$	0.06	&	152.14	$\pm$	95.77	&	\bf{	0.08	}	$\pm$	\bf{	0.07	}	&	0.09	$\pm$	0.08	&	101.45	$\pm$	63.82		\\
																																																				\\
21	&	MSE	&	\bf{	0.02	}	$\pm$	\bf{	0.04	}	&	0.03	$\pm$	0.04	&	5.86	$\pm$	171.26	&	\bf{	0.04	}	$\pm$	\bf{	0.05	}	&	0.05	$\pm$	0.07	&	1.0E+04	$\pm$	3.7E+04	&	\bf{	0.06	}	$\pm$	\bf{	0.07	}	&	0.1	$\pm$	0.24	&	6792.98	$\pm$	24608.24		\\
 	&	MAE	&	\bf{	0.12	}	$\pm$	\bf{	0.1	}	&	0.14	$\pm$	0.1	&	0.21	$\pm$	2.41	&	\bf{	0.15	}	$\pm$	\bf{	0.09	}	&	0.18	$\pm$	0.09	&	68.18	$\pm$	94.52	&	\bf{	0.18	}	$\pm$	\bf{	0.09	}	&	0.23	$\pm$	0.11	&	45.45	$\pm$	63.01		\\
 	&	MAPE	&	\bf{	2.6	}	$\pm$	\bf{	2	}	&	2.94	$\pm$	2.21	&	4.64	$\pm$	54.27	&	\bf{	3.18	}	$\pm$	\bf{	1.9	}	&	3.81	$\pm$	1.98	&	1457.61	$\pm$	2117.85	&	\bf{	3.84	}	$\pm$	\bf{	1.94	}	&	4.79	$\pm$	2.32	&	971.74	$\pm$	1411.89		\\																											

	\bottomrule
\end{tabular}}}
\end{table}

\begin{table}[!htb]
\centering
\tiny
\caption{Summary of prediction accuracy of CFGM, FGM and AR model \label{t:rsum}}
\begin{tabular}{rrrrrrrrrrr}
\toprule
	&		&	CFGM	&	FGM	&	AR	&	CFGM	&	FGM	&	AR	&	CFGM	&	FGM	&	AR	\\
Ceriteria	&		&	1-Step	&	1-Step	&	1-Step	&	2-Step	&	2-Step	&	2-Step	&	3-Step	&	3-Step	&	3-Step	\\
\hline\\
MAPE	&	Maximum	&	\textbf{16.4000}	&	19.9200 	&	247.5600 	&	\textbf{21.8800} 	&	35.0900 	&	2.1000E+05	&	\textbf{26.7200} 	&	71.8500 	&	1.4239E+05	\\
	&	Average	&	\textbf{4.0524} 	&	4.5319 	&	24.8971 	&	\textbf{5.4390} 	&	6.8210 	&	1.4988E+04	&	\textbf{6.8190}	&	10.5329 	&	1.0122E+04	\\
	&		&		&		&		&		&		&		&		&		&		\\
MAE	&	Maximum	&	\textbf{223.2600}	&	248.4300 	&	5134.1300 	&	\textbf{294.5500} 	&	346.8100 	&	1.6000E+06	&	\textbf{370.5500} 	&	463.4700 	&	1.1000E+06	\\
	&	Average	&	\textbf{30.9895}	&	36.1167 	&	288.1033 	&	\textbf{39.8171} 	&	50.9781 	&	8.5522E+04	&	\textbf{49.8490 }	&	75.6838 	&	5.8608E+04	\\
	&		&		&		&		&		&		&		&		&		&		\\
MSE	&	Maximum	&	\textbf{9.3000E+04}	&	1.2000E+05	&	3.2000E+09	&	\textbf{1.8000E+05}	&	6.8000E+05	&	9.1000E+12	&	\textbf{3.1000E+05}	&	2.6000E+07	&	6.1000E+12	\\
	&	Average	&	\textbf{7.3447E+03}	&	1.0745E+04	&	1.5308E+08	&	\textbf{1.3494E+04}	&	5.2048E+04	&	4.3526E+11	&	\textbf{2.3143E+04}	&	1.2814E+06	&	2.9173E+11	\\
	\bottomrule
\end{tabular}
\end{table}

\newpage

\section{Applications and analysis}

    In this section, we carry out the case study of predicting the annual natural gas production of 11 countries to show the performance in real world applications comparing to the existing fractional grey model FGM \cite{wulifeng2013}.

\subsection{Background and data collection}

\begin{table}[!htb]
\centering
\scriptsize
\caption{Annual natural gas (NG) production ($10^9 m^3$) of the 11 countries from 2008 to 2016 \label{t:rawdata}}
\begin{tabular}{rrrrrrrrrrr}
\toprule
No.&YEAR&2008&2009&2010&2011&2012&2013&2014&2015&2016
\\
\hline
\\
1&UAE&50.2&48.8&51.3&52.3&54.3&54.6&54.2&60.2&61.9
\\
2&Brazil&14&11.9&14.6&16.7&19.3&21.3&22.7&23.1&23.5
\\
3&Bolivia&14.3&12.3&14.2&15.6&17.8&20.3&21&20.3&19.7
\\
4&Denmark&10&8.4&8.2&6.6&5.7&4.8&4.6&4.6&4.5
\\
5&Netherlands&66.5&62.7&70.5&64.1&63.8&68.6&57.9&43.3&40.2
\\
6&Qatar&77&89.3&131.2&145.3&157&177.6&174.1&178.5&181.2
\\
7&Nigeria&36.2&26&37.3&40.6&43.3&36.2&45&50.1&44.9
\\
8&Turkmenistan&66.1&36.4&42.4&59.5&62.3&62.3&67.1&69.6&66.8
\\
9&Brunei&12.2&11.4&12.3&12.8&12.6&12.2&11.9&11.6&11.2
\\
10&Italy&8.4&7.3&7.6&7.7&7.8&7&6.5&6.2&5.3
\\
11&India&30.5&37.6&49.3&44.5&38.9&32.1&30.5&29.3&27.6
\\
\bottomrule
\end{tabular}
\end{table}

    Clean energy is one of most important resources, which will be vital to the global economics and natural environment in the future. Natural gas (NG) is one kind of clean energy with low price and high efficiency. And now many countries are using NG as one of the most important fuels.

    However, most recent researches only focus on the future trend of the largest gas producers such as the USA, China, Russia, $etc$ \cite{timedelayed}. In this paper, we selected 11 countries as mid-sized gas producers. The annual production data are collected from 2008 to 2016, which are available from the BP Statistical Review of World Energy \footnote{Available at the website of British Petroleum Company https://www.bp.com/en/global/corporate/energy-economics/statistical-review-of-world-energy.html} and listed in Table \ref{t:rawdata}. The countries listed in Table \ref{t:rawdata} are actually mid or small size economic entity, changes of energy market, industrial development, domestic policies $etc.$ may all significantly effect the natural gas consumption of these countries. Thus we can see that the time series of their natural gas consumption are all not stable. Moreover, with high speed changes of global economics and energy markets, trustable data of these countries are very limited. As a consequence,
    it is difficult to accurately predict the natural gas consumption of such countries using the traditional energy prediction models with very small samples.

\begin{figure}[!htb]
\centering
  \includegraphics[width=0.75\textwidth]{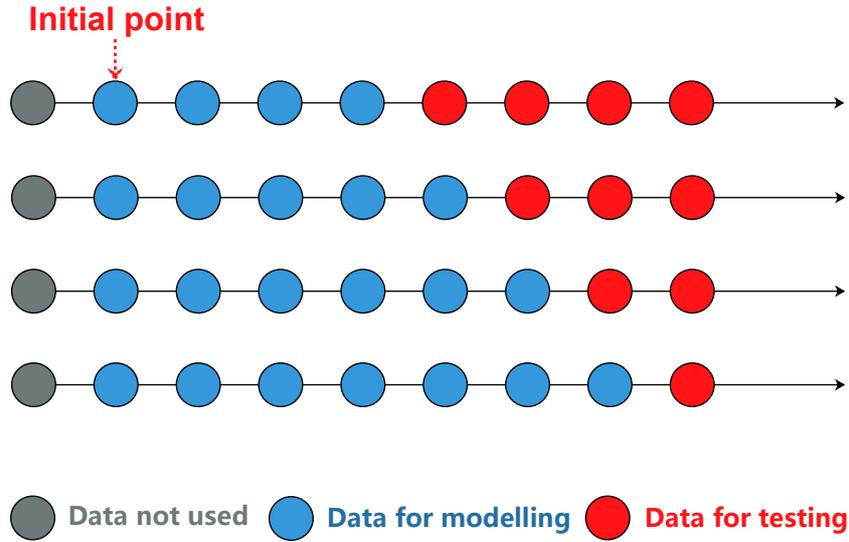}
\caption{The time series cross validation
 \label{fig:timeseries}}
\end{figure}

\subsection{Overall performance in comparison to the existing fractional grey model}

    In order to provide a comprehensive comparison of CFGM to the existing FGM model, we use the time series cross validation (TSCV) in this section, which has been quite efficient to evaluate the overall performance and robustness of the times series models in \cite{timeseries}.

    The main idea of TSCV is illustrated in Fig. \ref{fig:timeseries}. For each case listed in Table \ref{t:rawdata} the prediction models would be built on different subcases shown in \ref{fig:timeseries} with different initial points and different numbers of sample data. The Brute Force method is also used to select the optimal $\alpha$ for the CFGM and FGM in this subsection. The $\alpha$ for CFGM is searched in the $[0,2]$ by step 0.01, and that for FGM is searched in a wider range $[-2,2]$ by step 0.01.

    According to the TSCV shown in Fig. \ref{fig:timeseries}, the CFGM and FGM model would be built on 15 subcases in each case \footnote{As there are 9 points in each case, the subcases should be with 4 to 8 points for modelling when the initial point is the first point in the original series, and then there are 5 subcases. Similarly, we could know that there are 4,3,2,1 subcases when the 2$^{nd}$, 3$^{rd}$, 4$^{th}$,5$^{th}$ point are used as initial point, respectively. So we have 5+4+3+2+1=15 subcases  in each case.}, and this means we need to build these models for $15\times11=165$ times in total.

    The overall evaluation criteria for fitting accuracy are listed in Table \ref{t:retuls_fit}. The MSE, MAE and MAPE of CFGM are smaller than FGM model in 9, 7 and 7 cases, respectively, which indicates that the overall fitting accuracy of CFGM model is better than FGM. Meanwhile, it can be seen that most STDs of CFGM model are also smaller than FGM model, which indicates that the stability of CFGM model is also better than FGM model. On the other hand, it should also be noticed that the fitting accuracy of CFGM is quite close to the FGM model although when FGM has higher accuracy.

\begin{table}[!htb]
\centering
\scriptsize
\caption{Fitting performance of CFGM and FGM in time series cross validation \label{t:retuls_fit}}
\begin{tabular}{rrrrr}
\toprule
	&		&	MSE			&	MAE		&	MAPE(\%)		\\
\hline
\\
Case1	&	CFGM	&	0.6762 	$\pm$	1.5977 	&	0.4503 	$\pm$	1.5696 	&	0.8151 	$\pm$	3.0833 	\\	
	    &	FGM	&	\bf{0.4706 	$\pm$	1.2755} 	&	\bf{0.3531 	$\pm$	1.6242} 	&	\bf{0.6404 	$\pm$	3.1979} 	\\	
\\
Case2	&	CFGM	&	\bf{0.0313 	$\pm$	0.0724} 	&	\bf{0.1028 	$\pm$	2.3606}	&	\bf{0.5340 	$\pm$	17.5909}	\\	
	    &	FGM	&	0.1037 	$\pm$	0.2520 	&	0.1881 	$\pm$	2.2940 	&	1.0942 	$\pm$	17.2361 	\\	
\\
Case3	&	CFGM	&	\bf{0.1323 	$\pm$	0.3300}	&	\bf{0.2017 	$\pm$	1.9584} 	&	\bf{1.0867 	$\pm$	14.2824} 	\\	
	    &	FGM	&	0.1817 	$\pm$	0.3947 	&	0.2563 	$\pm$	1.9661 	&	1.4204 	$\pm$	14.5303 	\\	
\\
Case4	&	CFGM	&	0.0485 	$\pm$	0.0811 	&	0.1552 	$\pm$	1.0819 	&	2.6744 	$\pm$	12.4777 	\\	
	    &	FGM	&	\bf{0.0355 	$\pm$	0.0855}	&	\bf{0.1012 	$\pm$	1.1046}	&	\bf{1.6036 	$\pm$	12.7874} 	\\	
\\
Case5	&	CFGM	&	\bf{8.4423 	$\pm$	14.7093} 	&	2.0447 	$\pm$	4.0094 	&	3.3530 	$\pm$	6.1291 	\\	
	    &	FGM	&	9.4880 	$\pm$	18.2279 	&	\bf{2.0254 	$\pm$	4.4964} 	&	\bf{3.2386 	$\pm$	6.9701} 	\\	
\\
Case6	&	CFGM	&	\bf{12.8673 	$\pm$	20.5661} 	&	\bf{2.4822 	$\pm$	22.2533}	&	\bf{1.6616 	$\pm$	26.8327}	\\	
	    &	FGM	&	23.5152 	$\pm$	53.7708 	&	2.9278 	$\pm$	22.1466 	&	2.2180 	$\pm$	26.8304 	\\	
\\
Case7	&	CFGM	&	\bf{6.6457 	$\pm$	13.2297 }	&	\bf{1.6339 	$\pm$	4.3122} 	&	\bf{4.2040 	$\pm$	15.7513} 	\\	
	    &	FGM	&	7.4043 	$\pm$	12.3512 	&	1.8317 	$\pm$	4.4426 	&	4.9719 	$\pm$	16.5651 	\\	
\\
Case8	&	CFGM	&	\bf{4.0174 	$\pm$	6.7740} 	&	\bf{1.3110 	$\pm$	9.3218} 	&	\bf{2.3761 	$\pm$	20.7557} 	\\	
	    &	FGM	&	8.7518 	$\pm$	15.9248 	&	1.8508 	$\pm$	9.1328 	&	3.5410 	$\pm$	20.5943 	\\	
\\
Case9	&	CFGM	&\bf{	0.0031 	$\pm$	0.0062} 	&	\bf{0.0336 	$\pm$	0.4015} 	&	\bf{0.2724 	$\pm$	3.4182}	\\	
	    &	FGM	&	0.0141 	$\pm$	0.0328 	&	0.0672 	$\pm$	0.3734 	&	0.5487 	$\pm$	3.1928 	\\	
\\
Case10	&	CFGM	&	\bf{0.0128 	$\pm$	0.0262} 	&	\bf{0.0703 	$\pm$	0.4215} 	&	\bf{0.9519 	$\pm$	5.3180} 	\\	
	    &	FGM	&	0.0220 	$\pm$	0.0494 	&	0.0834 	$\pm$	0.4057 	&	1.1255 	$\pm$	5.1085 	\\	
\\
Case11	&	CFGM	&	\bf{3.4843 	$\pm$	6.9808} 	&	1.1751 	$\pm$	5.1944 	&	3.1228 	$\pm$	14.3985 	\\	
	    &	FGM	&	4.1196 	$\pm$	12.0940 	&	\bf{0.9921 	$\pm$	5.3509} 	&	\bf{2.4508 	$\pm$	14.6638} 	\\	
\\	
\centering														
Total	&		&	9:2			&	7:4	&			7:4			\\	
\bottomrule
\end{tabular}
\end{table}

\newpage

\begin{table}[!htb]
\centering
\scriptsize
\caption{Prediction performance of CFGM and FGM in Cross validation \label{t:retuls_pred}}
\begin{tabular}{rrrrr}
\toprule
	&		&	MSE			&	MAE		&	MAPE(\%)		\\
\hline
\\
Case1	&	CFGM	&	\bf{30.4939 	$\pm$	39.3554} 	&	4.4096 	$\pm$	3.3726 	&	\bf{7.3025 	$\pm$	5.4229} 	\\	
	&	FGM	&	34.5486 	$\pm$	39.5743 	&	\bf{4.8863 	$\pm$	3.3146} 	&	8.0961 	$\pm$	5.3269 	\\	
\\
Case2	&	CFGM	&	\bf{8.3676 	$\pm$	18.3188} 	&	\bf{2.0370 	$\pm$	2.0838} 	&	\bf{8.8024 	$\pm$	8.8677} 	\\	
	&	FGM	&	22.6661 	$\pm$	42.6136 	&	3.3110 	$\pm$	3.4710 	&	14.2821 	$\pm$	14.7859 	\\	
\\
Case3	&	CFGM	&	\bf{18.8652 	$\pm$	26.8104} 	&	\bf{3.3413 	$\pm$	2.8156} 	&	\bf{16.7184 	$\pm$	14.2824} 	\\	
	&	FGM	&	25.9883 	$\pm$	39.2119 	&	3.8781 	$\pm$	3.3572 	&	19.4072 	$\pm$	17.0292 	\\	
\\
Case4	&	CFGM	&	\bf{1.6338 	$\pm$	1.8054} 	&	\bf{1.1026 	$\pm$	0.6560} 	&	\bf{24.1863 	$\pm$	14.5933} 	\\	
	&	FGM	&	6.2599 	$\pm$	18.7491 	&	1.4047 	$\pm$	2.1007 	&	30.7946 	$\pm$	46.4865 	\\	
\\
Case5	&	CFGM	&	\bf{196.2533 	$\pm$	280.8344} 	&	\bf{11.1167 	$\pm$	8.6493} 	&	\bf{24.6414 	$\pm$	21.2041} 	\\	
	&	FGM	&	524.1852 	$\pm$	721.9064 	&	18.3410 	$\pm$	13.9038 	&	42.7472 	$\pm$	34.8749 	\\	
\\
Case6	&	CFGM	&	\bf{425.8787 	$\pm$	767.7857} 	&	\bf{15.0169 	$\pm$	14.3620} 	&	\bf{8.4184 	$\pm$	7.9611 }	\\	
	&	FGM	&	1371.6504 	$\pm$	2805.3065 	&	24.7774 	$\pm$	27.9288 	&	13.8827 	$\pm$	15.5000 	\\	
\\
Case7	&	CFGM	&	\bf{153.6544 	$\pm$	342.6568} 	&	\bf{9.1558 	$\pm$	8.4781} 	&	\bf{20.2662 	$\pm$	18.7234} 	\\	
	&	FGM	&	3100.9623 	$\pm$	13156.8750 	&	26.9264 	$\pm$	49.4552 	&	59.1887 	$\pm$	108.8847 	\\	
\\
Case8	&	CFGM	&	\bf{2443.2359 	$\pm$	8753.3487} 	&	\bf{21.6147 	$\pm$	45.1016} 	&	\bf{32.2738 	$\pm$	66.9320}	\\	
	&	FGM	&	4992.1914 	$\pm$	19111.8027 	&	32.7297 	$\pm$	63.5318 	&	48.7507 	$\pm$	94.2844 	\\	
\\
Case9	&	CFGM	&	\bf{0.4373 	$\pm$	0.7388} 	&	\bf{0.4786 	$\pm$	0.4630} 	&	\bf{4.1491 	$\pm$	4.0524} 	\\	
	&	FGM	&	0.7558 	$\pm$	1.7631 	&	0.5654 	$\pm$	0.6700 	&	4.8841 	$\pm$	5.8270 	\\	
\\
Case10	&	CFGM	&	\bf{1.0815 	$\pm$	1.7394} 	&	\bf{0.7954 	$\pm$	0.6797} 	&	\bf{13.5638 	$\pm$	12.3565} 	\\	
	&	FGM	&	1.2261 	$\pm$	1.9004 	&	0.8540 	$\pm$	0.7151 	&	14.7381 	$\pm$	13.0126 	\\	
\\
Case11	&	CFGM	&	\bf{34.2671 	$\pm$	38.8337 }	&	\bf{5.0070 	$\pm$	3.0770} 	&	\bf{17.0717 	$\pm$	10.3767} 	\\	
	&	FGM	&	1041.3573 	$\pm$	4094.3538 	&	15.7837 	$\pm$	28.5575 	&	54.8024 	$\pm$	102.3495 	\\	
\\
Total	&		&	11:0			&	10:1	&			11:0			\\	
\bottomrule
\end{tabular}
\end{table}

    The overall evaluation criteria for prediction accuracy are listed in Table \ref{t:retuls_pred}. The MSE, MAE and MAPE of CFGM are smaller than FGM model in all cases. Thus it is very clear that the CFGM model is very competitive to the FGM model in prediction accuracy. Meanwhile, also almost all STDs of CFGM are smaller than FGM, except the STD of AE in case 1 and STD of APE in case 1. Thus it is also shown that stability of CFGM model is better than FGM. On the other hand, it can be seen that the upper bounds of the criteria of CFGM are quite smaller than that of FGM. $e.g.$ the maximum MAPE of CFGM is 32.2738\%, and that of FGM is as large as 59.8024\%. Thus, it is sufficient to show that the CFGM model is more effective in prediction than CFGM with more acceptable robustness.

    And it should also be noticed that the $\alpha$ for FGM model is searched in a wider range than CFGM model, thus it is implied that it would be easier for the CFGM to select the optimal $\alpha$ than FGM in the applications.

\subsection{About the parameter $\alpha$}
\label{sec:alpha}

    In the numerical example \ref{subsec:alpha}  it has been shown that the value of $\alpha$ would effect the modelling accuracy of the CFGM model. In this subsection we will analyze the ranges which contain the optimal $\alpha$ in most cases.

    The proportion of optimal $\alpha$ for CFGM and FGM are shown in the subfigures in Fig.\ref{fig:alphas}. It is shown that 84\% optimal $\alpha$ are obtained in $(0,1)$ for CFGM, and 15\% optimal $\alpha$ are obtained at 0. In total, we can see that there are 99\% values are obtained in $[0,1)$. And only a few optimal $\alpha$ (just 1.22\%) are obtained in $(1,2]$ .

    As for the FGM model, most optimal $\alpha$ are obtained in $(0,1)$, which take 72\% in all the cases. But there still exist 25\% points obtained in $(1,2)$.
\begin{figure}[!htb]
\subfloat[]{
  \includegraphics[width=0.45\textwidth]{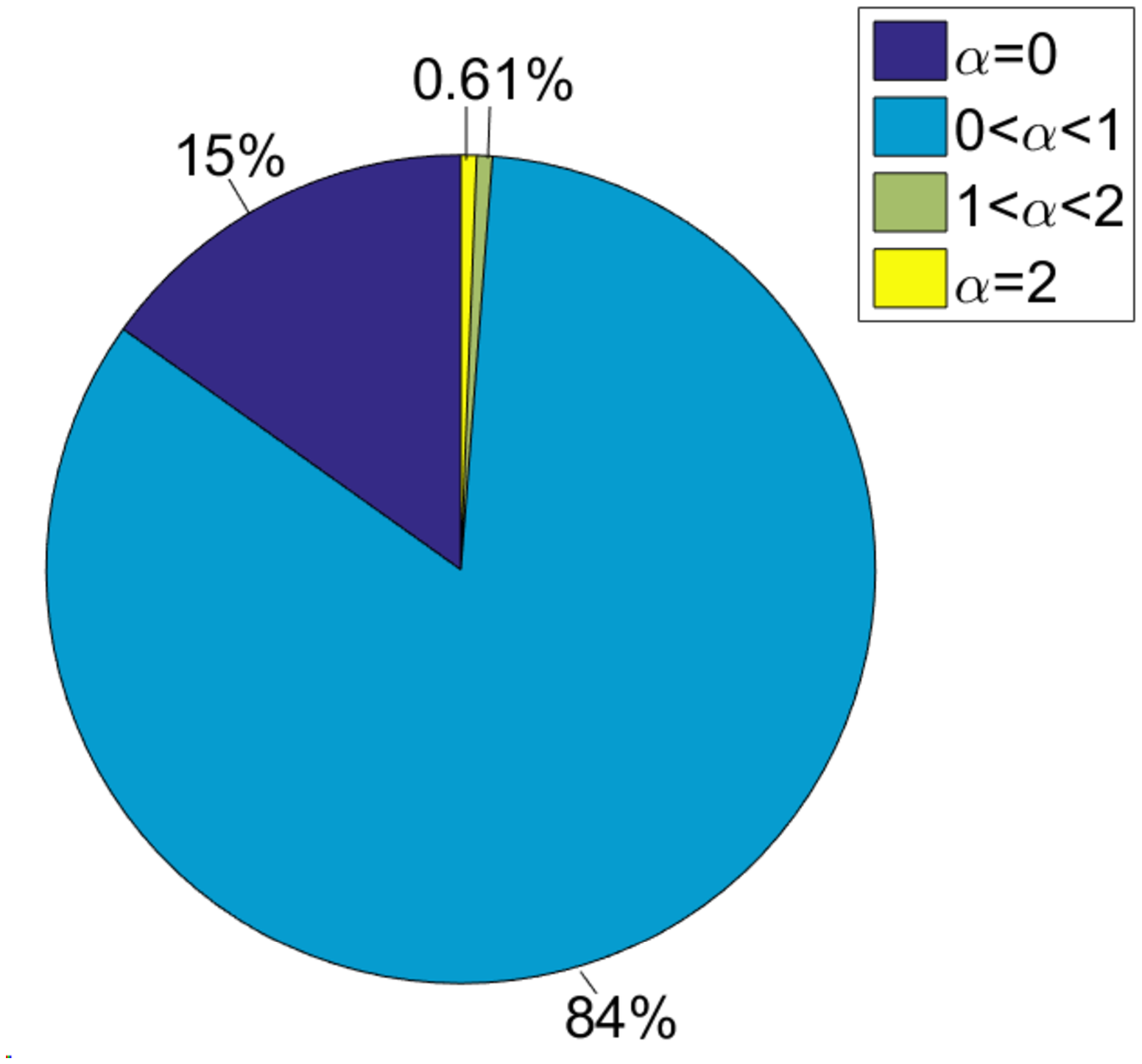}}\hfill
\subfloat[]{
  \includegraphics[width=0.45\textwidth]{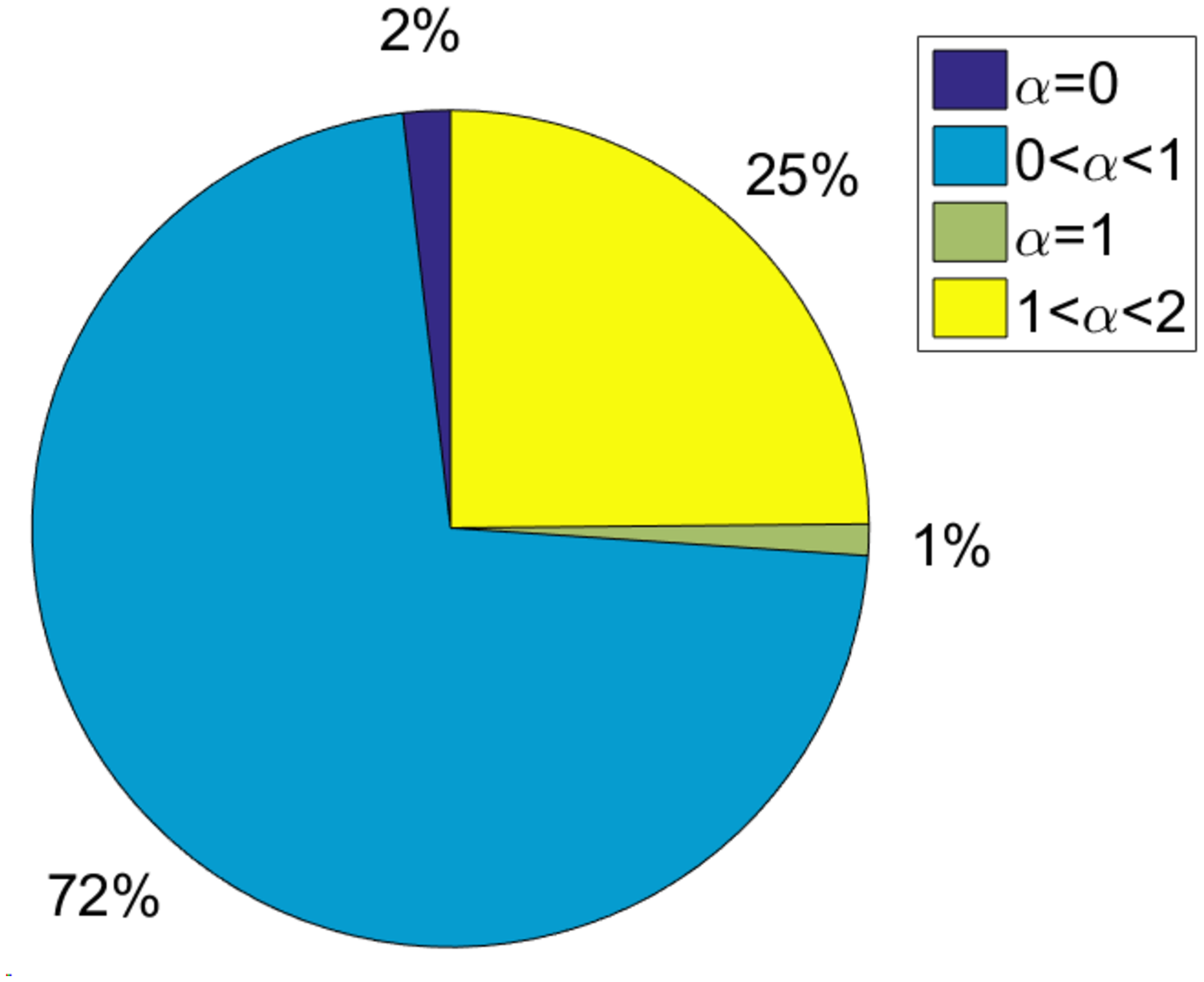}}\hfill

\caption{The proportion of optimal $\alpha$ in all cases of CFGM and FGM. (a) Optimal $\alpha$ of CFGM model; (b) Optimal $\alpha$ of FGM model.
 \label{fig:alphas}}
\end{figure}
    According to the results shown above, it is clear to see that the FGM model needs a wider range to select the optimal $\alpha$ than CFGM. And it is shown that we almost only need to search the optimal $\alpha$ in the range of $[0,1)$ for the CFGM model in the applications. This finding indicates that it is easier to optimize the $\alpha$ for CFGM. Moever, it is also very usefull for us to design more precise algorithm or try to use other optimizers, as the available searching range has been shown in this finding.

\subsection{Some typical cases and analysis}

    For better explanation, we choose some typical cases to compare the features of CFGM model to the FGM model. The three cases chosen in this subsection are cases 7-9, where the maximum errors of CFGM and FGM appear in case 8 and 7, respectively, and they perform closely in case 9, as shown in Table \ref{t:retuls_pred}.

\subsubsection{The non-smooth series with one inflexion point}

\begin{figure}[!htb]
\subfloat[]{
  \includegraphics[width=0.45\textwidth]{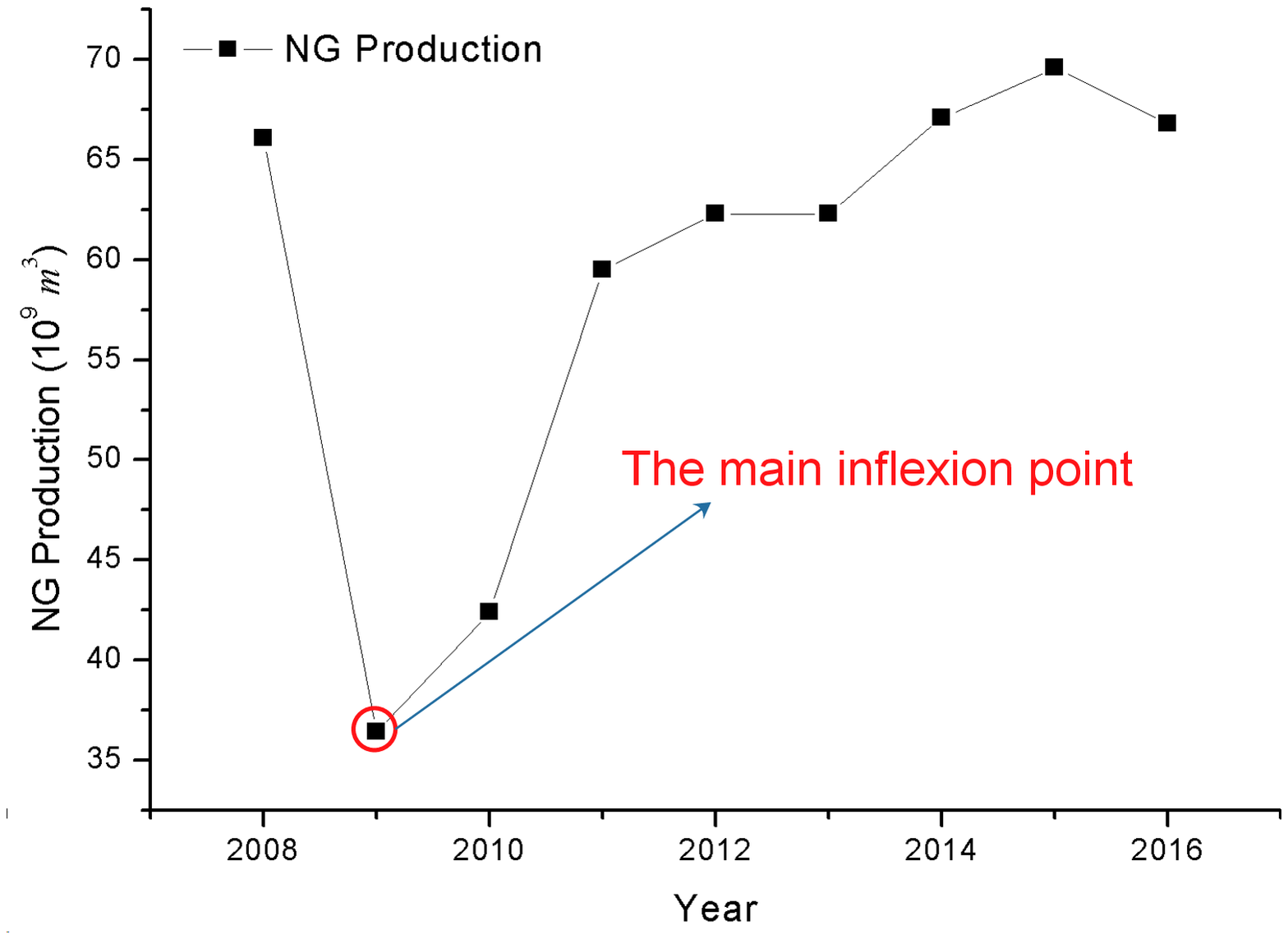}}\hfill
\subfloat[]{
  \includegraphics[width=0.45\textwidth]{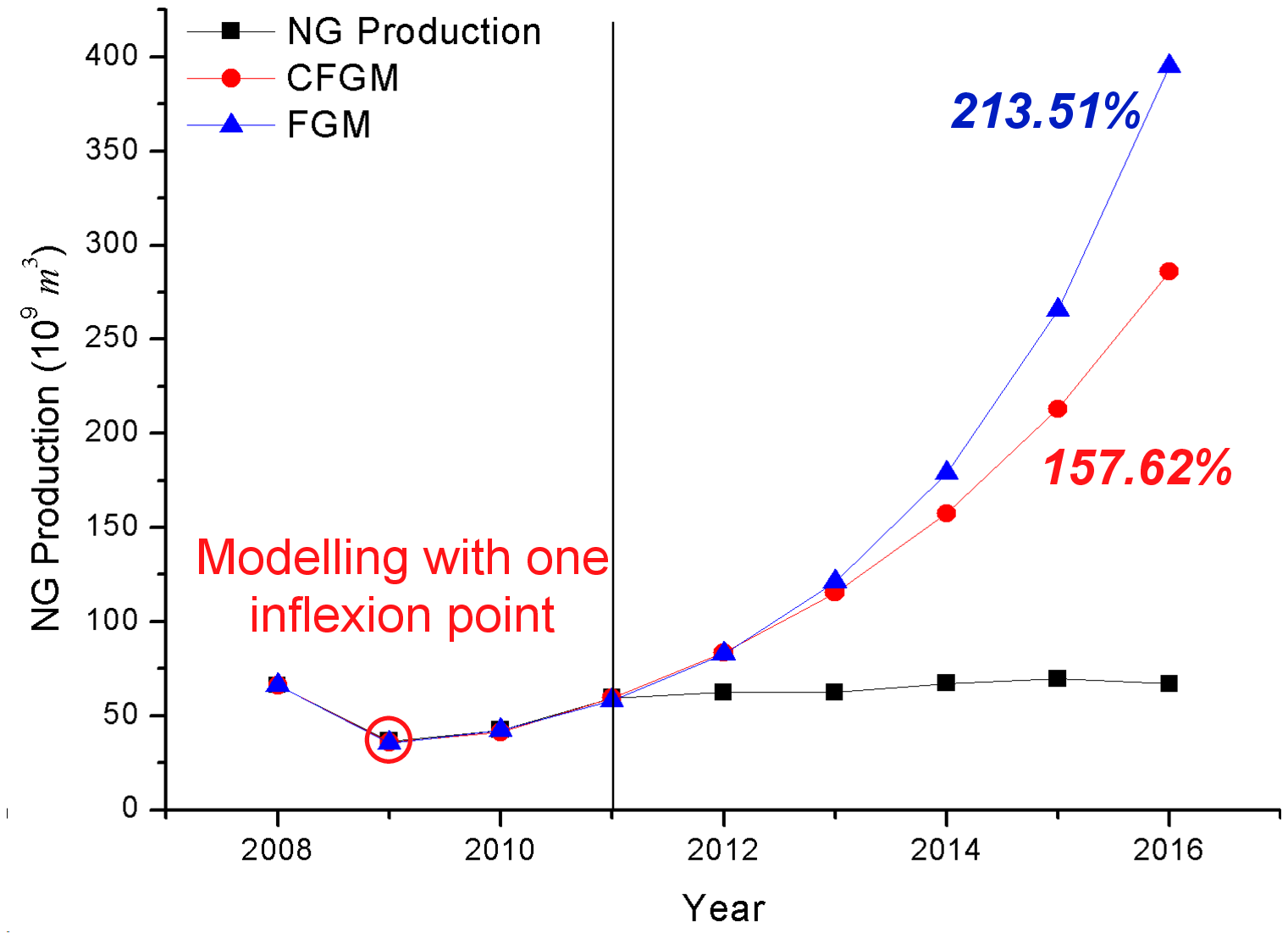}}\hfill
\subfloat[]{
  \includegraphics[width=0.45\textwidth]{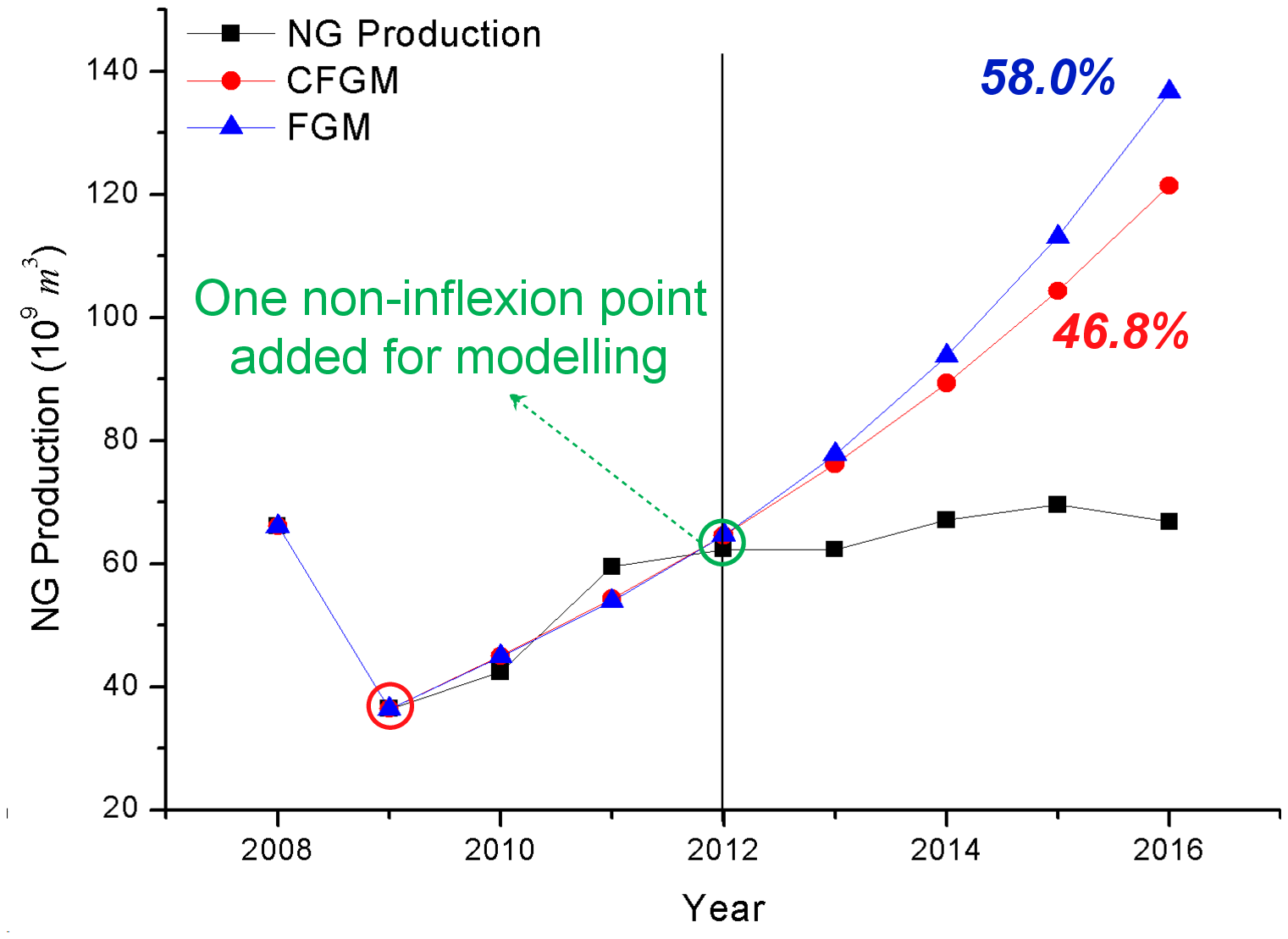}}\hfill
\subfloat[]{
  \includegraphics[width=0.45\textwidth]{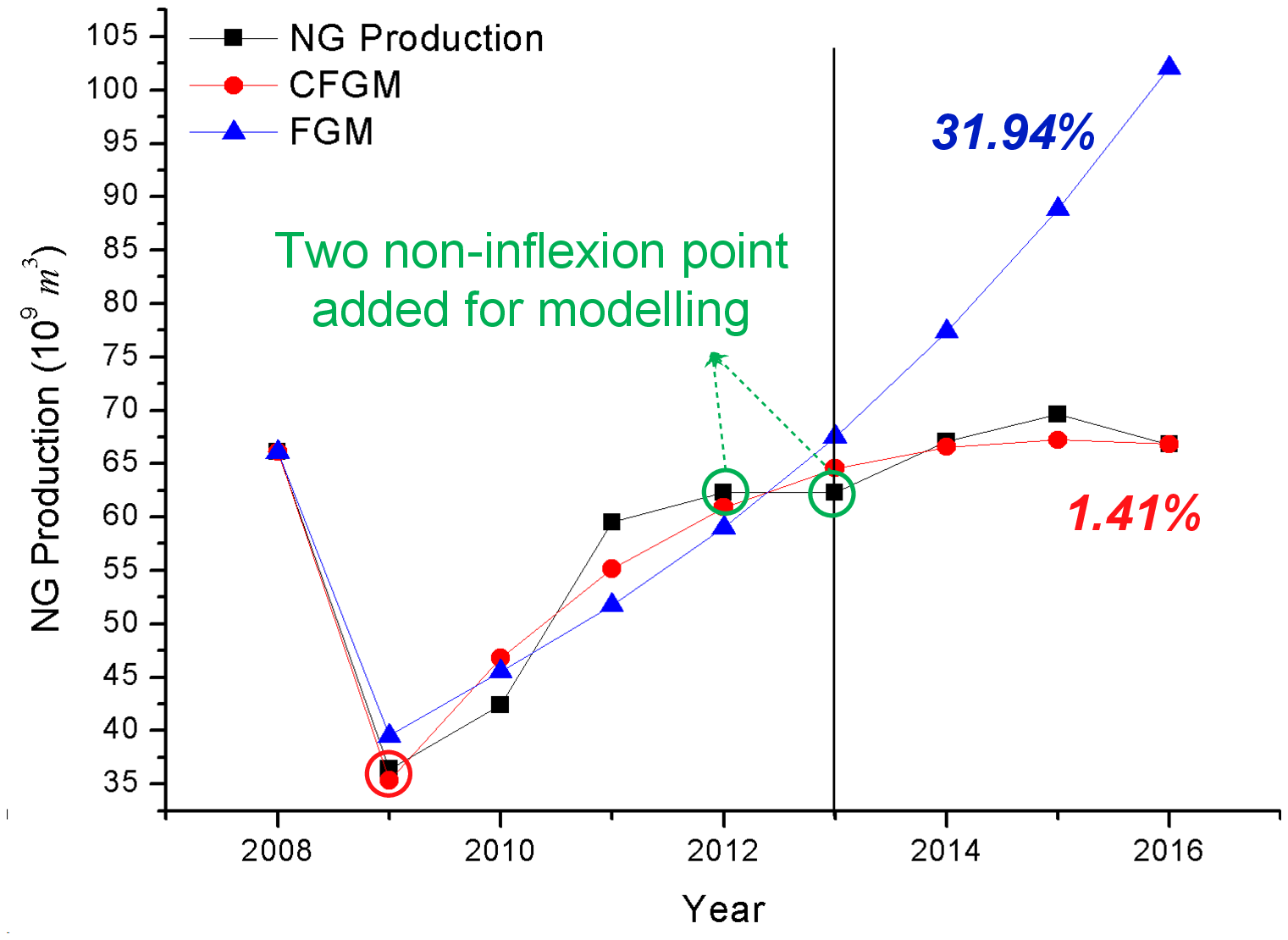}}\hfill
\caption{The prediction results in case 8 with 1 main inflexion point. (a) The raw data of NG production of Turkmenistan. (b)  4 points for modelling with 1 inflexion points. (c)  5 points for modelling with 1 new non-inflexion point.  (d) 5 points for modelling with 2 new non-inflexion points.
 \label{fig:case8}}
\end{figure}

    The CFGM model has the maximum errors in case 8, with MAPE as 32.2738\%.  We firstly plotted the raw data of case 8, the NG production of Turkmenistan, in subfigure (a) in Fig. \ref{fig:case8}. It can be shown that the series started to change its direction at the second point. Intuitively, we call such point as the inflexion point in the rest of this paper. And in this case, the second point is its main inflexion point.

    With such an inflexion point, we can see that the CFGM and FGM model all perform poorly in subfigure (b) in  Fig. \ref{fig:case8}, with quite large MAPE, which implies that it is  such large errors which lead to the poor overall performance in this case. In subfigure (c), one non-inflexion point added for modelling, and then it can be seen that the accuracy of these models is improved significantly. It is very interesting to see that when two non-inflexion points added for modelling, the performance of CFGM becomes much better immediately. However, the performance of FGM model is still not well at all.

    With this case, we can see that the inflexion point is very important to the CFGM model, which would make it inefficient in the applications. However, the inference of the inflexion point can be weakened when more non-inflexion points added.

    Moreover, it can be seen that the CFGM model is more sensitive to the newly added non-inflexion point than the FGM model, which makes it more effective in this case.

\begin{figure}[!htb]
\begin{center}
\subfloat[]{
  \includegraphics[width=0.45\textwidth]{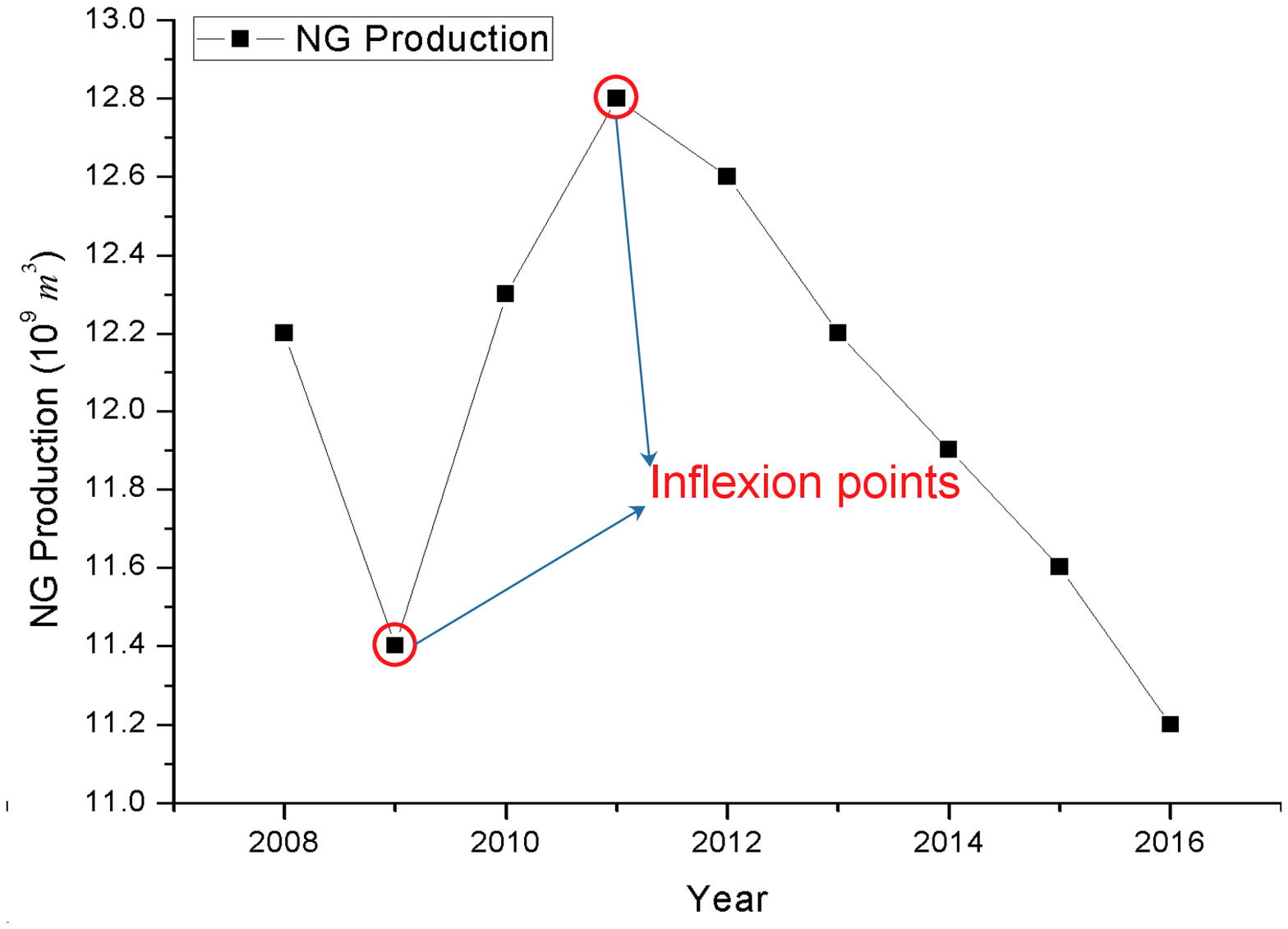}}
\end{center}
\subfloat[]{
  \includegraphics[width=0.45\textwidth]{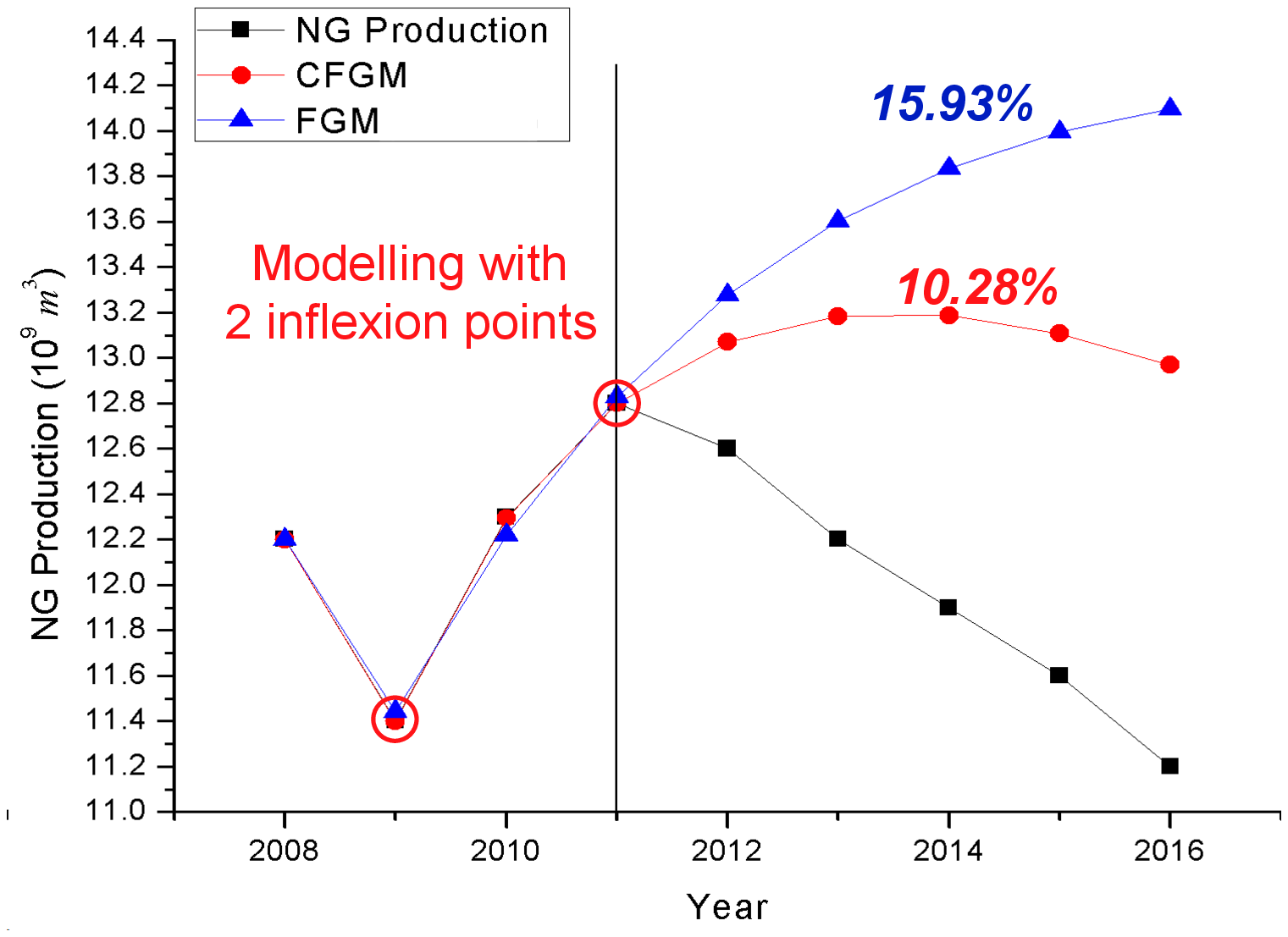}}\hfill
\subfloat[]{
  \includegraphics[width=0.45\textwidth]{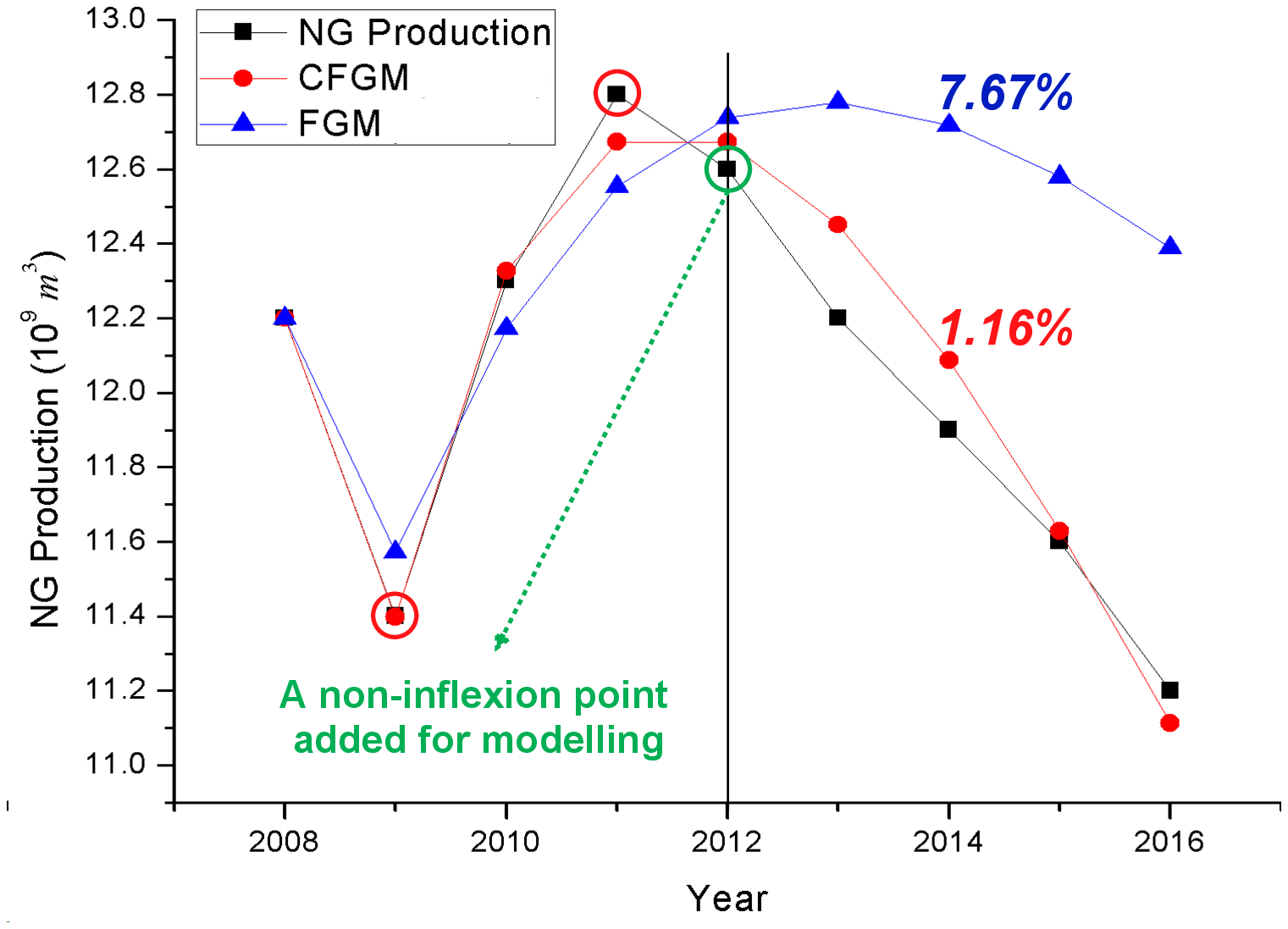}}\hfill

\caption{The prediction results for case 9 with 2 inflexion points. (a) The raw data of NG production of Brunei. (b) 4 points for modelling with 2 inflexion points. (c)  5 points for modelling with 1 new non-inflexion point.
 \label{fig:case9}}
\end{figure}

\subsubsection{The non-smooth series with two inflexion points}

    In case 9, the CFGM has similar performance to the FGM, although it has better accuracy than FGM. The original series of case 9 has two inflexion points as shown in subfigure (a) in Fig. \ref{fig:case9}. Thus we can see that if we use only four points to build the models, they only reflect the overall trend of the direction from the second to fourth points, as shown in subfigure (b) in Fig. \ref{fig:case9}. However, it is clear to see that the CFGM correctly catches the overall trend when only one non-inflexion point added for modelling, while the change of FGM is not significant, as shown in subfigure (c). This case clearly indicates that the CFGM is much more sensitive to the new information than FGM model.

\subsubsection{The non-smooth series with four inflexion points}

\begin{figure}[!htb]

\subfloat[]{
  \includegraphics[width=0.45\textwidth]{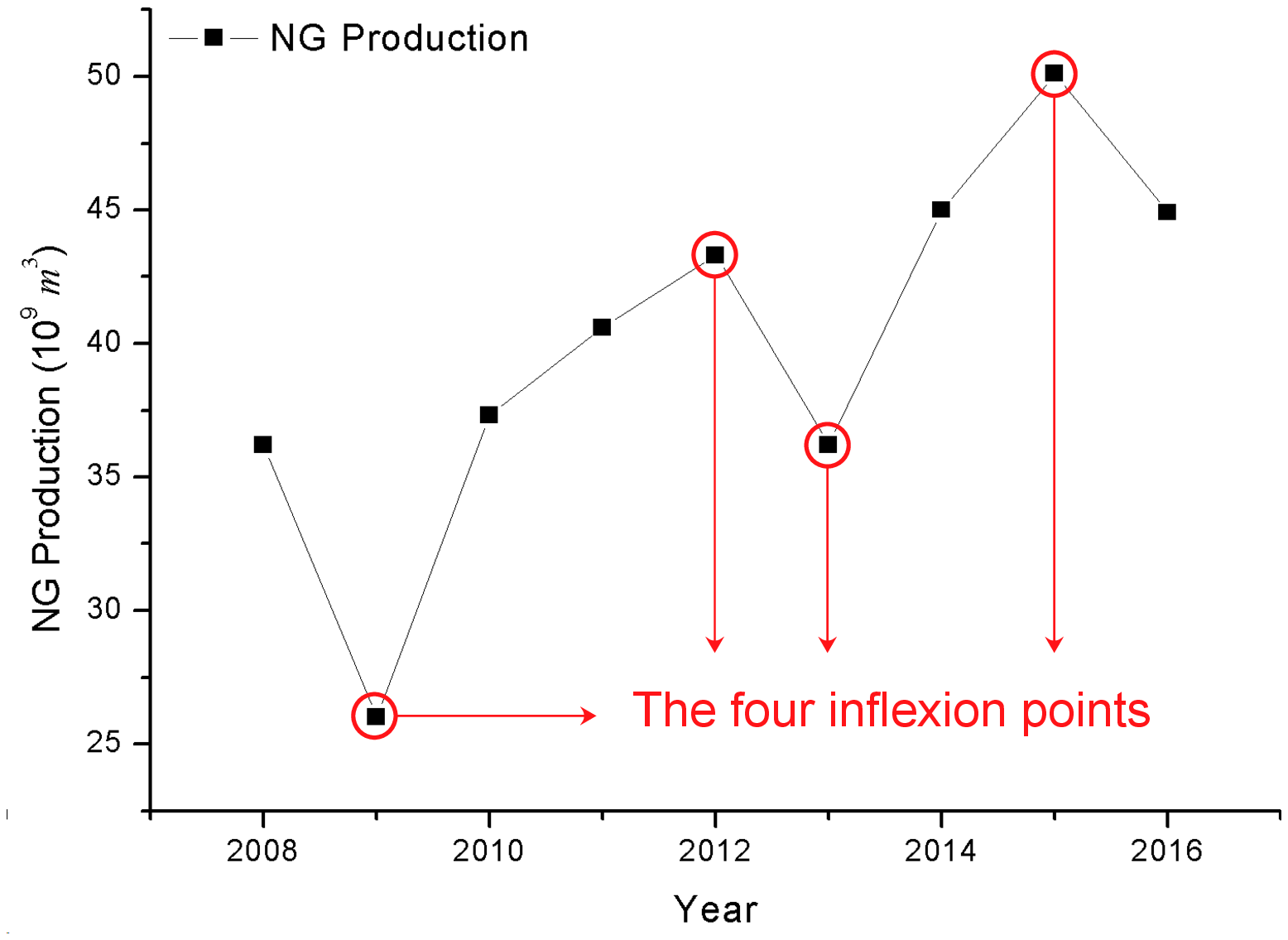}}\hfill
\subfloat[]{
  \includegraphics[width=0.45\textwidth]{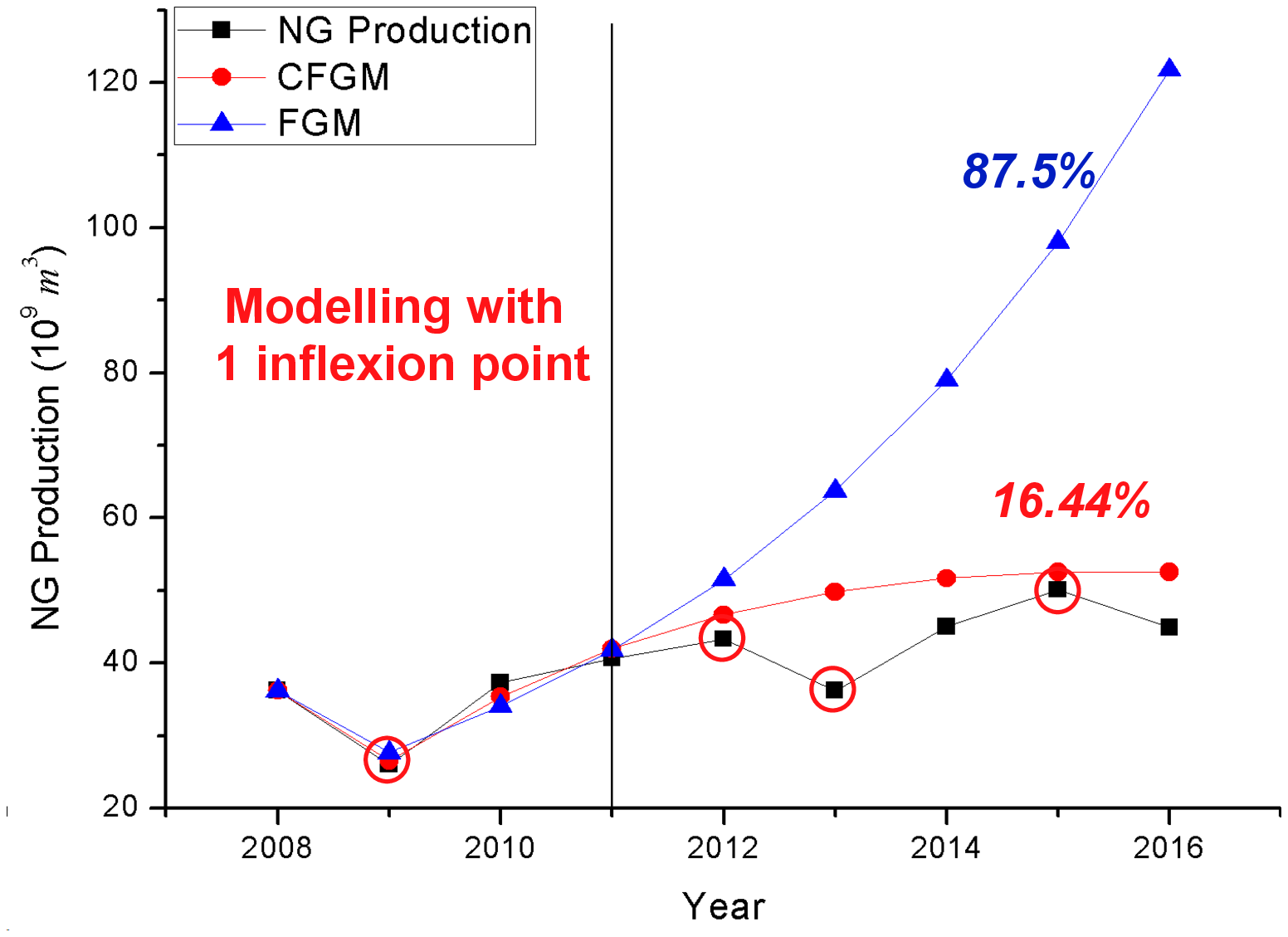}}\hfill
\subfloat[]{
  \includegraphics[width=0.45\textwidth]{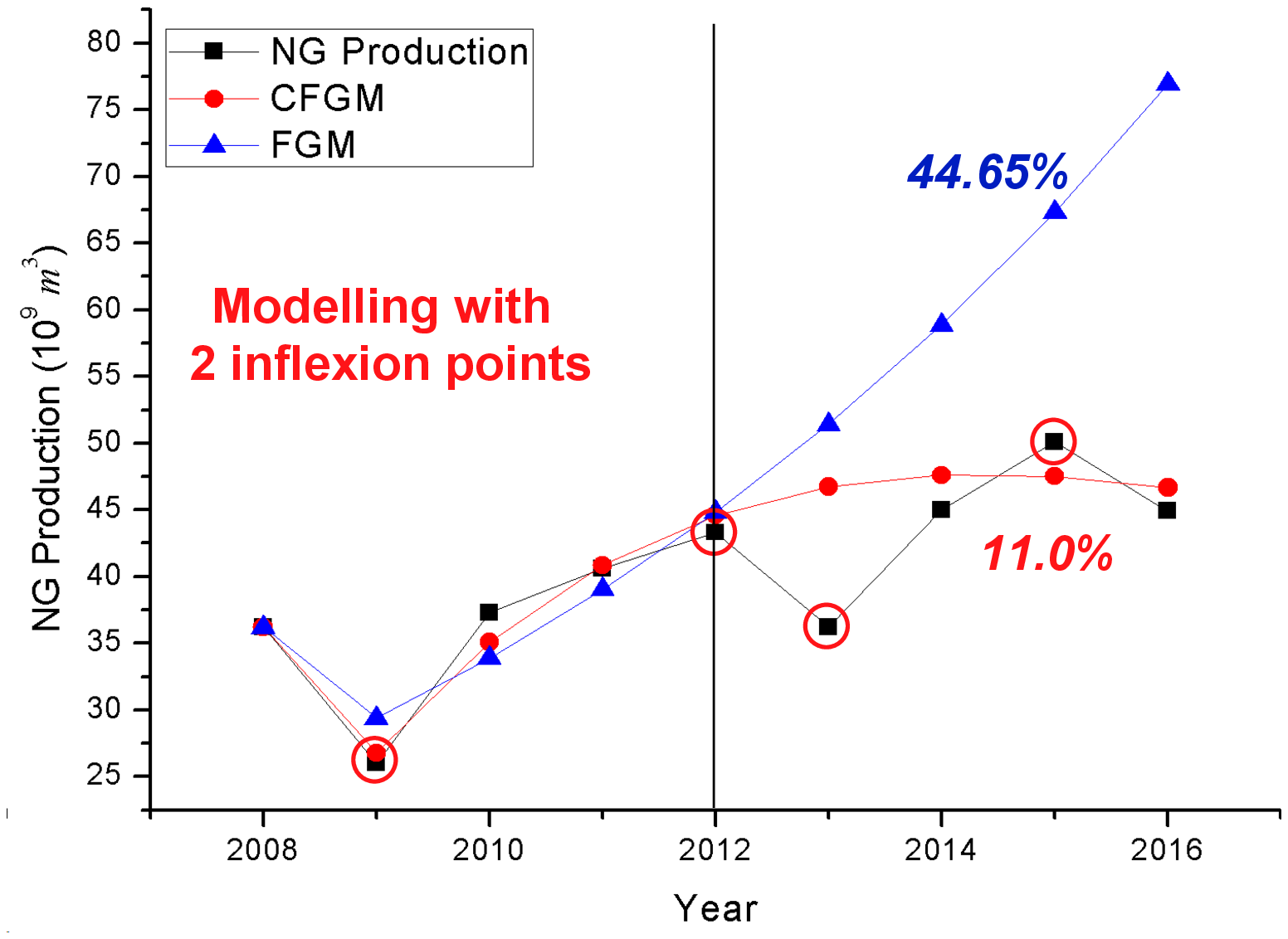}}\hfill
\subfloat[]{
  \includegraphics[width=0.45\textwidth]{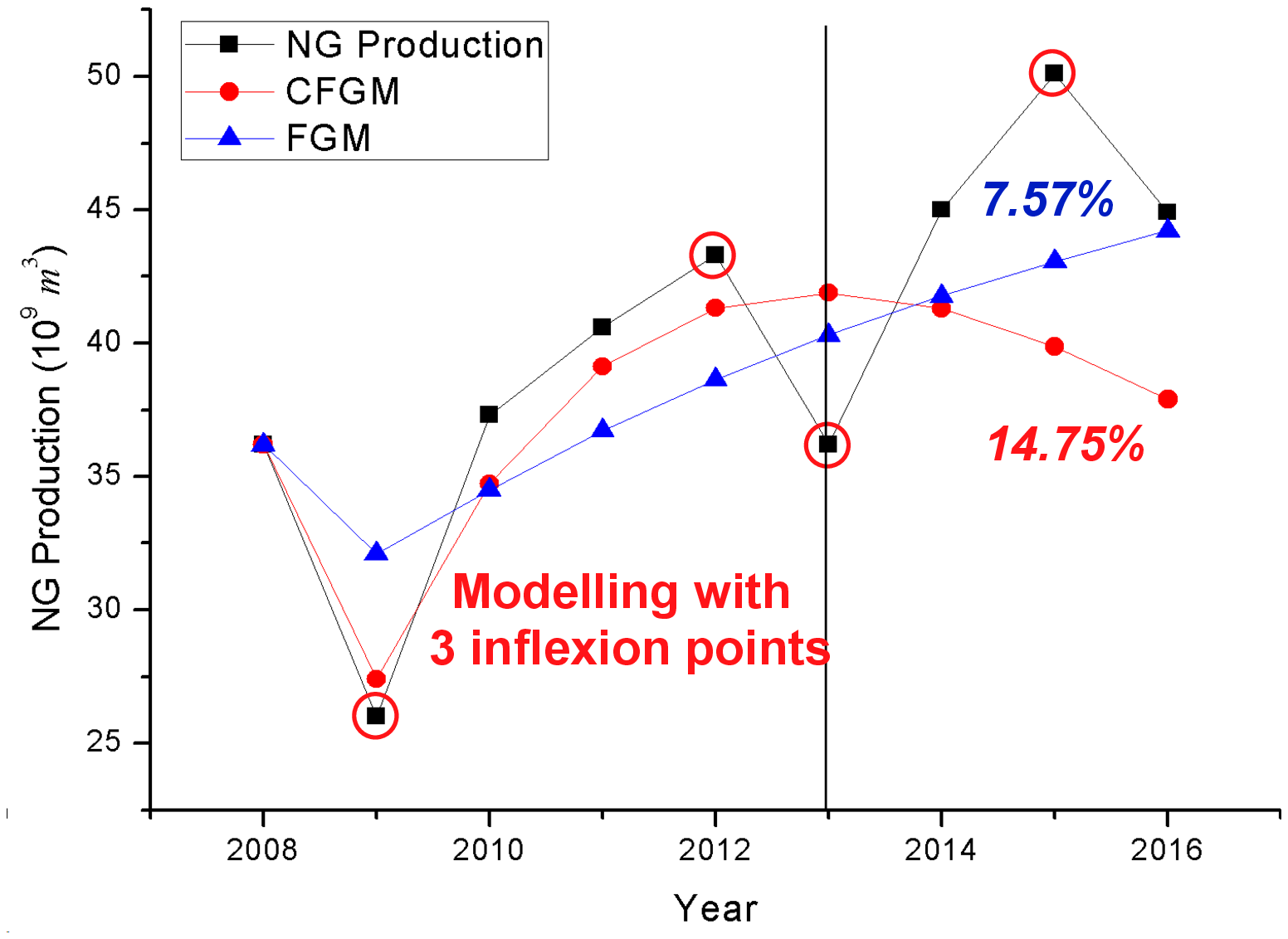}}\hfill
\subfloat[]{
  \includegraphics[width=0.45\textwidth]{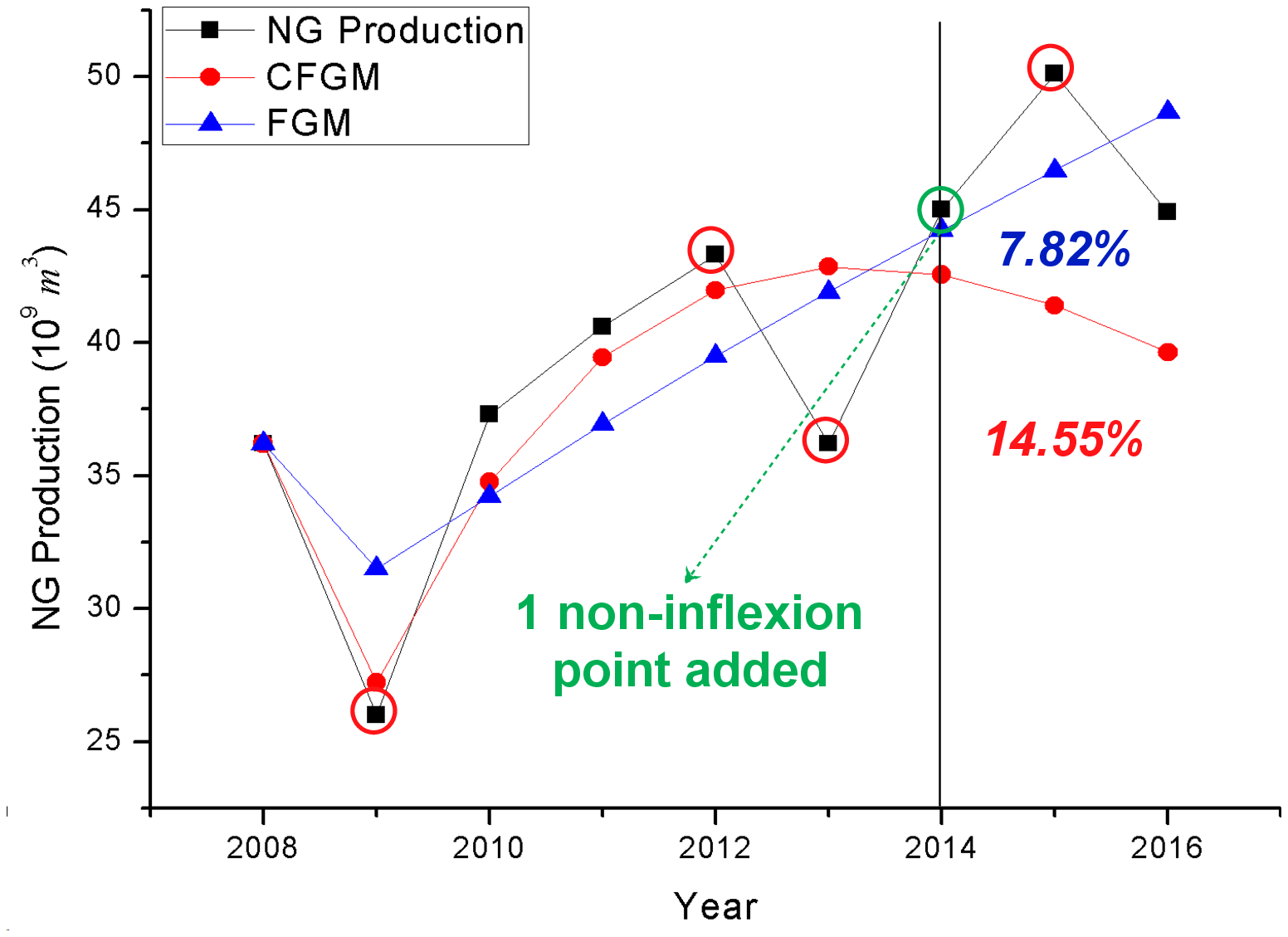}}\hfill
\subfloat[]{
  \includegraphics[width=0.45\textwidth]{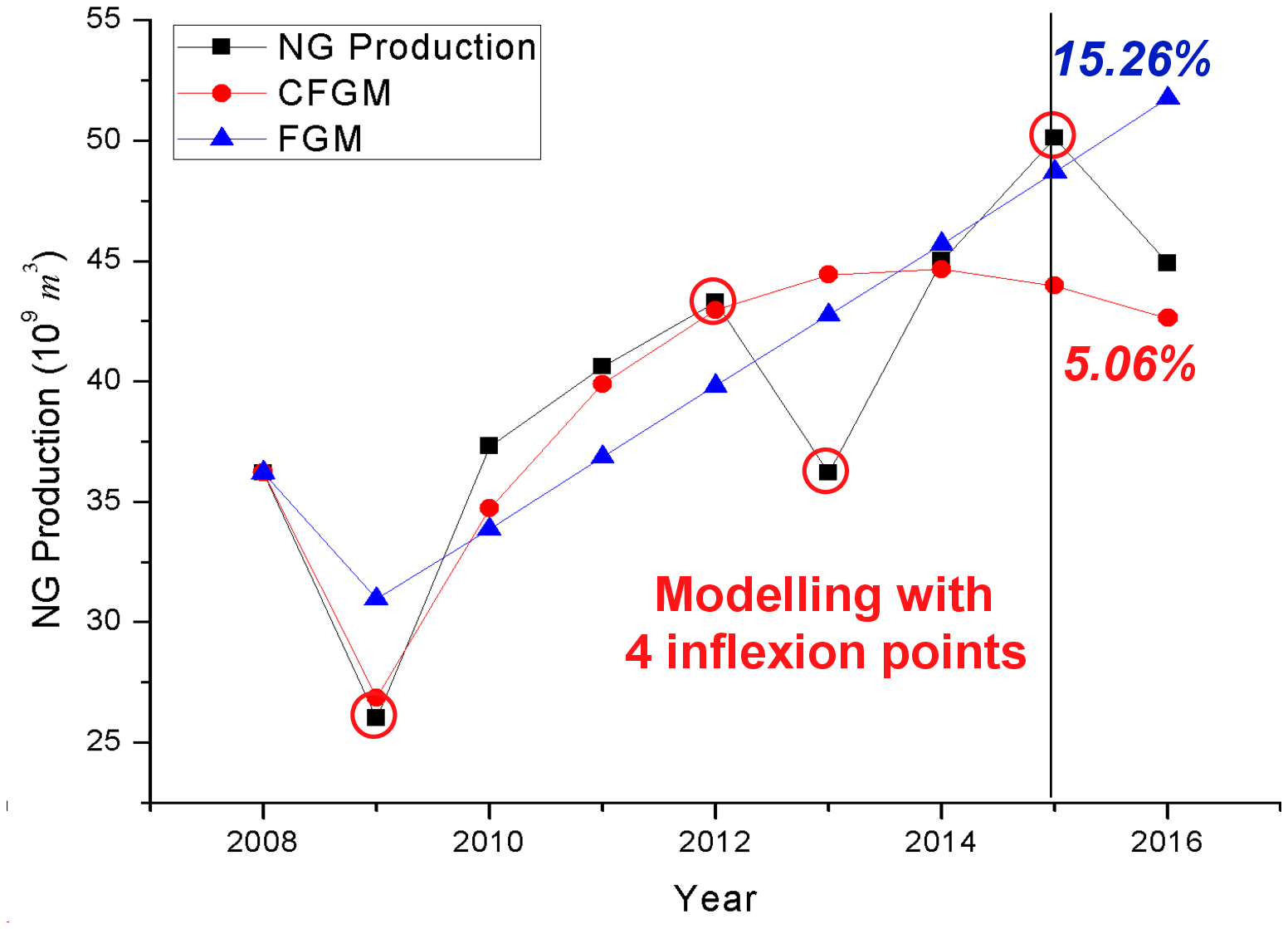}}\hfill
\caption{The prediction results in case 7 with 4 inflexion points. (a) The raw data of NG production of Nigeria; (b) 4 points for modelling with 1 inflexion point; (c) The results of CFGM and FGM with 5 points for modelling, including 2 inflexion points; (d)  6 points for modelling with 3 inflexion points; (e) 7 points for modelling with 1 new non-inflexion point; (f) 8 points for modelling with 4 inflexion points.
 \label{fig:case7}}
\end{figure}

    The case 7 is the most special case as it has four inflexion points out of the total 9 points, as shown in subfigure (a) in Fig. \ref{fig:case7}. In another point of view, it can be seen that the direction is changed with every three or less points, thus it approaches to the oscillatory series.

    With only one inflexion point for modelling, the CFGM and FGM all have poor performance, and the FGM overestimate the trend of the original series, as shown in subfigure (b) in Fig. \ref{fig:case7}. Their performances become better when a new inflexion point added, because the last four points for modelling roughly reflect the overall trend of the last series, shown in subfigure (c).

    However, the situations shown in subfigures (d) and (e) indicate that the performance of CFGM become worse with newly added points, whatever the inflexion or non-inflexion point. And we can see in these subfigures that the effects of the newly added points to the CFGM is quite significant. But for the FGM model, it always follows the overall increasing trend of the original series, even when the last point added shown in the subfigure (f).

    This case presents quite a clear picture of the sensitivities of the CFGM and FGM, which indicates the CFGM is much more sensitive to the new information than FGM model. However, it also indicates that such sensitivity may also lead to worse performance to the CFGM model.

\subsubsection{Further discussions on the typical cases}

    Combining the results shown above, we further analyze the properties of the CFGM and FGM. As mentioned in Subsection \ref{sec:modelcfgm}, the CFGM and FGM share the same modelling procedures except the definitions of fractional order accumulation and difference.  Actually all the two definitions will make the original series smoother so that the response function can be very close to the accumulated series based on the least squares method. As shown in Table. \ref{t:fagoseries}, the fitting errors for the accumulated series by CFGM and FGM are all very small, and the fitting errors for the accumulated series of CFGM and FGM are very close in most cases. The predicted values of the original series are generated using the fractional order difference, thus it is clear that the significant difference between the CFGM and FGM are caused by the formulation of the fractional order differences.

\begin{table}[!htb]
\caption{ MAPEs for fitting the accumulated series by CFGM and FGM in the three typical cases  \label{t:fagoseries}}
\centering
\tiny
\begin{tabular}{rrrrrrrrr}
\toprule
Case 8	&		&		&	Case 9	&		&		&	Case 7	&		&		\\
Modelling Points	&		&	Fitting MAPE	&	Modelling Points	&		&	Fitting MAPE	&	Modelling Points	&		&	Fitting MAPE	\\
\hline
4	&	CFGM	&	0.0074	&	4	&	CFGM	&	0.3294	&	4	&	CFGM	&	2.6321	\\
	&	FGM	&	0.1026	&		&	FGM	&	0.8331	&		&	FGM	&	1.0107	\\
	&		&		&		&		&		&		&		&		\\
5	&	CFGM	&	0.0518	&	5	&	CFGM	&	0.5417	&	5	&	CFGM	&	2.5067	\\
	&	FGM	&	0.3434	&		&	FGM	&	0.5724	&		&	FGM	&	1.4792	\\
	&		&		&		&		&		&		&		&		\\
6	&	CFGM	&	0.0776	&		&		&		&	6	&	CFGM	&	6.0371	\\
	&	FGM	&	0.5626	&		&		&		&		&	FGM	&	3.4061	\\
	&		&		&		&		&		&		&		&		\\
	&		&		&		&		&		&	7	&	CFGM	&	5.8968	\\
	&		&		&		&		&		&		&	FGM	&	2.6772	\\
	&		&		&		&		&		&		&		&		\\
	&		&		&		&		&		&	8	&	CFGM	&	6.0592	\\
	&		&		&		&		&		&		&	FGM	&	2.2572	\\
\bottomrule
\end{tabular}
\end{table}

    For convenience, we denote the accumulated series by FOA as $x_W^{(\alpha)}(k)$ and the predicted values of FOA series as $\hat{x}_W^{(\alpha)}(k)$, and the accumulated series by CFA as  $x^{(\alpha)}(k)$ and the predicted values of CFA series as $\hat{x}^{(\alpha)}(k)$, and the errors as $e^{(\alpha)}_W(k) = \hat{x}_W^{(\alpha)}(k) - x_W^{(\alpha)}(k)$ and $e^{(\alpha)}(k) = \hat{x}^{(\alpha)}(k) - x^{(\alpha)}(k)$. Then we can estimate the errors of the original series by FGM and CFGM.

    Recalling the definition of the  fractional order difference in \eqref{eq:ifago}, the error of the restored values $x^{(0)}(k)$ of FGM can be written as
    \begin{equation}
        \begin{aligned}
           e_W^{(0)}(k) & = \hat{x} _W^{(0)}(k) - x^{(0)}(k) \\
                      & = \sum_{j=1}^{k} \left(
              \begin{array}{c}
                 k-j-\alpha-1  \\ k-j
              \end{array}\right) \left( \hat{x}_W^{(\alpha)}(k) - x_W^{(\alpha)}(k)\right) \\
              & = \sum_{j=1}^{k} \left(
              \begin{array}{c}
                 k-j-\alpha-1  \\ k-j
              \end{array}\right)e^{(\alpha)}_W(k).
        \end{aligned}
    \end{equation}

    We can see that the errors of the restored values by FGM are actually the IFOA series of the errors $e^{(\alpha)}_W(k)$. In another words, the error of restored value by FGM at each point is effected by all the errors of the FOA series. This property of FGM sometimes improve its applicability as it take use of each point as presented by Wu. However, it may also enlarge the error of FGM in some cases. For example, if the error of the FOA series at one point is very large, the errors after this point will be all effected by it. And this is one reason the FGM can not make accurate prediction in the above three cases.

    But for the CFGM, the error of the restored values is actually

    \begin{equation}
        \begin{aligned}
           e^{(0)}(k)   & =  \hat{x}^{(0)}(k) - x^{(0)}(k) \\
                        & =  k^{[\alpha] - \alpha} \Delta ^n  \left( \hat{x}^{(\alpha)}(k) - x^{(\alpha)}(k) \right) \\
                        & =  k^{[\alpha] - \alpha} \Delta ^n e^{(\alpha)}(k).
        \end{aligned}
    \end{equation}

    Particularly, when $\alpha \in (0,1]$ the error is actually $e^{(0)}(k)=k^{1 - \alpha} \left( e^{(\alpha)}(k)-e^{(\alpha)}(k-1) \right)$, this means it is effected by the errors of the CFA series with the two points; and when $\alpha \in (1,2]$ the error of restored value is effected by the errors of the CFA series with three points. This indicates that the error of restored value by CFGM at each point is only effected by the closest few points. As shown in the Subsection \ref{sec:alpha}, most values of $\alpha$ are in the interval (0,1), this implies the error of restored value by CFGM at each point is often effected by the former two points. Thus CFGM can be more sensitive to respond to the new points.

    With the analysis above, it is clear to see the reason why the CFGM can outperform the FGM in the case studies. As mentioned in the statements of background, the time series of natural gas consumption in all the cases are all not stable, thus the sensitivity of CFGM would make it more efficient in the subcases with non-smooth series, and then the overall performance of CFGM can be much better than FGM.

\section{Conclusions}

In this paper we proposed a novel definition of conformable fractional difference (CFD) and conformable fractional accumulation (CFA), and then used them to build the novel conformable fractional grey model (CFGM). With the CFA and CFD, the CFGM proposed in this paper is much easier to implement and the optimal fractional order $\alpha$ is also very easy to obtain using the simple Brute Force method as shown in the numerical example.

Benchmark data sets obtained from the Time Series Data Library are used to validate the effectiveness of CFGM, in comparison with  FGM and AR model, based on the classical 1-step, 2-step and 3-step prediction tests. The results show that the CFGM is significantly more efficient than FGM, especially in longer term predictions. The real world applications in predicting the natural gas consumptions in 11 countries are also carried out to compare the effectiveness of CFGM and FGM. The time series cross validation also indicates that CFGM outperforms the FGM. The results also indicate that the fractional order of CFGM model needs narrower range (often in $[0,1)$) to tune, which implies it might be easier to tune the fractional order of CFGM. Analysis of three typical time series shows that  the restored errors of CFGM is only effected by a few closest points, which makes it more sensitive to the new information, thus it can be more efficient in predicting the unstable time series with several inflexion points. With widely existed unstable time series with small size, the CFGM can be expected to be applied in a wider range of application fields.

Moreover, introduction of the CFGM can be regarded as a new methodology of grey models. The CFD and CFA can take place of the fractional order operations used in the existing works, and then the framework of existing fractional grey models can be reconstructed. And also a new framework of conformable grey models can be built in the similar way to the CFGM. Our future works will be mainly oriented in such new frameworks. Otherwise, bigger picture can also be considered in the future works, like applying the CFD and CFA to the difference equations, the statistical models (like ARIMA), discrete dynamical systems simulation and control, $etc.$, and more details, $e.\ g.$ the sampling period for simulation and controlling, should be analyzed in the future.

\section{Conflict of Interests}
    The authors declare that there is no conflict of interests regarding the publication of this paper.

\section{Acknowledgment}
    This research was supported by the National Natural Science
Foundation of China (Grant No. 71771033, 71571157), Humanities and Social Science Project of Ministry of Education of China (No.19YJCZH119), the Doctoral Research Foundation of Southwest University of Science and Technology (no. 16zx7140, 15zx7141), the Open Fund (PLN201710) of State Key Laboratory of Oil and Gas Reservoir Geology and Exploitation (Southwest Petroleum University), and National Statistical Scientific Research Project (2018LY42).

\bibliographystyle{unsrt}
\bibliography{cfgm}

\begin{thebibliography}{10}

\bibitem{xie2017review}
Naiming Xie and Ruizhi Wang.
\newblock A historic review of grey forecasting models.
\newblock {\em Journal of Grey System}, 29(4), 2017.

\bibitem{zgx2018}
Gaoxun Zhang, Yi~Zheng, Honglei Zhang, and Xinchen Xie.
\newblock L{\'e}vy process-driven asymmetric heteroscedastic option pricing
  model and empirical analysis.
\newblock {\em Discrete Dynamics in Nature and Society}, 2018, 2018.

\bibitem{maminda2018}
Minda Ma, Wei Cai, Weiguang Cai, and Liang Dong.
\newblock {Whether carbon intensity in the commercial building sector decouples
  from economic development in the service industry? Empirical evidence from
  the top five urban agglomerations in China}.
\newblock {\em Journal of Cleaner Production}, 222:193 -- 205, 2019.

\bibitem{knea}
Xin Ma and Zhibin Liu.
\newblock Predicting the oil production using the novel multivariate nonlinear
  model based on {Arps} decline model and kernel method.
\newblock {\em Neural Computing and Applications}, 29(2):579--591, 2018.

\bibitem{xinyi01}
Pei Du, Jianzhou Wang, Wendong Yang, and Tong Niu.
\newblock A novel hybrid model for short-term wind power forecasting.
\newblock {\em Applied Soft Computing}, 80:93 -- 106, 2019.

\bibitem{xinyi02}
Pei Du, Jianzhou Wang, Zhenhai Guo, and Wendong Yang.
\newblock Research and application of a novel hybrid forecasting system based
  on multi-objective optimization for wind speed forecasting.
\newblock {\em Energy Conversion and Management}, 150:90--107, 2017.

\bibitem{WU2019280}
Li-Feng Wu, Guo-Min Huang, Jun-Liang Fan, Fu-Cang Zhang, Xiu-Kang Wang, and
  Wen-Zhi Zeng.
\newblock Potential of kernel-based nonlinear extension of {Arps} decline model
  and gradient boosting with categorical features support for predicting daily
  global solar radiation in humid regions.
\newblock {\em Energy Conversion and Management}, 183:280 -- 295, 2019.

\bibitem{fan2019empirical}
Junliang Fan, Lifeng Wu, Fucang Zhang, Huanjie Cai, Wenzhi Zeng, Xiukang Wang,
  and Haiyang Zou.
\newblock {Empirical and machine learning models for predicting daily global
  solar radiation from sunshine duration: A review and case study in China}.
\newblock {\em Renewable and Sustainable Energy Reviews}, 100:186--212, 2019.

\bibitem{YANG2019942}
Wendong Yang, Jianzhou Wang, Haiyan Lu, Tong Niu, and Pei Du.
\newblock {Hybrid wind energy forecasting and analysis system based on divide
  and conquer scheme: A case study in China}.
\newblock {\em Journal of Cleaner Production}, 222:942 -- 959, 2019.

\bibitem{deng1988}
Julong Deng.
\newblock {\em Essential topics on grey system: theory and applications}.
\newblock China Ocean Press, 1988.

\bibitem{smallsample}
Lifeng Wu, Sifeng Liu, Ligen Yao, and Shuli Yan.
\newblock The effect of sample size on the grey system model.
\newblock {\em Applied Mathematical Modelling}, 37(9):6577--6583, 2013.

\bibitem{kgm1n}
Xin Ma and Zhi-bin Liu.
\newblock The kernel-based nonlinear multivariate grey model.
\newblock {\em Applied Mathematical Modelling}, 56:217--238, 2018.

\bibitem{zb2019apm}
Bo~Zeng, Huiming Duan, and Zhou Yufeng.
\newblock A new multivariable grey prediction model with structure
  compatibility.
\newblock {\em Applied Mathematical Modelling}, 75:385--397, 2019.

\bibitem{wzx2017}
Zheng-Xin Wang, Hong-Hao Zheng, Ling-Ling Pei, and Tong Jin.
\newblock Decomposition of the factors influencing export fluctuation in
  {China's} new energy industry based on a constant market share model.
\newblock {\em Energy Policy}, 109:22--35, 2017.

\bibitem{zb2018}
Bo~Zeng, Yongtao Tan, Hui Xu, Jing Quan, Luyun Wang, and Xueyu Zhou.
\newblock Forecasting the electricity consumption of commercial sector in {Hong
  Kong} using a novel grey dynamic prediction model.
\newblock {\em Journal of Grey System}, 30(1):159--174, 2018.

\bibitem{wang2019modelling}
Zheng-Xin Wang and Qin Li.
\newblock {Modelling the nonlinear relationship between CO2 emissions and
  economic growth using a PSO algorithm-based grey Verhulst model}.
\newblock {\em Journal of Cleaner Production}, 207:214--224, 2019.

\bibitem{meng2018prediction}
Wei Meng, Daoli Yang, and Hui Huang.
\newblock Prediction of china’s sulfur dioxide emissions by discrete grey
  model with fractional order generation operators.
\newblock {\em Complexity}, 2018:1--13, 2018.

\bibitem{xiong2019grey}
Ping-ping Xiong, Wen-jie Yan, Gui-zhi Wang, and Ling-ling Pei.
\newblock Grey extended prediction model based on irls and its application on
  smog pollution.
\newblock {\em Applied Soft Computing}, 80:797--809, 2019.

\bibitem{tan2000}
Guan-Jun Tan.
\newblock The structure method and application of background value in grey
  system {GM (1, 1)} model {(I)}.
\newblock {\em Systems Engineering-Theory \& Practice}, 20(4):98--103, 2000.

\bibitem{zbcie2018}
Bo~Zeng and Chuan Li.
\newblock Improved multi-variable grey forecasting model with a dynamic
  background-value coefficient and its application.
\newblock {\em Computers \& Industrial Engineering}, 118:278--290, 2018.

\bibitem{maxingmco}
Xin Ma and Zhibin Liu.
\newblock The {GMC (1, n)} model with optimized parameters and its application.
\newblock {\em Journal of grey system}, 29(4):122--138, 2017.

\bibitem{xie2009}
Nai-ming Xie and Si-feng Liu.
\newblock Discrete grey forecasting model and its optimization.
\newblock {\em Applied Mathematical Modelling}, 33(2):1173--1186, 2009.

\bibitem{ngdm2013}
Nai-Ming Xie, Si-Feng Liu, Ying-Jie Yang, and Chao-Qing Yuan.
\newblock On novel grey forecasting model based on non-homogeneous index
  sequence.
\newblock {\em Applied Mathematical Modelling}, 37(7):5059--5068, 2013.

\bibitem{dgm1n2015}
Xin Ma and Zhibin Liu.
\newblock Predicting the oil field production using the novel discrete {GM (1,
  N)} model.
\newblock {\em Journal of Grey System}, 27(4):63--73, 2015.

\bibitem{rdgm2016}
Xin Ma and Zhi-bin Liu.
\newblock Research on the novel recursive discrete multivariate grey prediction
  model and its applications.
\newblock {\em Applied Mathematical Modelling}, 40(7-8):4876--4890, 2016.

\bibitem{DING2019}
Song Ding.
\newblock {A novel discrete grey multivariable model and its application in
  forecasting the output value of Chin's high-tech industries}.
\newblock {\em Computers and Industrial Engineering}, 127:749--760, 2019.

\bibitem{ds2018}
Song Ding, Keith~W Hipel, and Yao-guo Dang.
\newblock Forecasting {China's} electricity consumption using a new grey
  prediction model.
\newblock {\em Energy}, 149:314--328, 2018.

\bibitem{maxin2019}
Xin Ma.
\newblock A brief introduction to the grey machine learning.
\newblock {\em Journal of Grey System}, 31(1):1--12, 2019.

\bibitem{wzxdatagroup}
Zheng-Xin Wang, Qin Li, and Ling-Ling Pei.
\newblock Grey forecasting method of quarterly hydropower production in {China}
  based on a data grouping approach.
\newblock {\em Applied Mathematical Modelling}, 51:302--316, 2017.

\bibitem{changmega}
Che-Jung Chang, Liping Yu, and Peng Jin.
\newblock A mega-trend-diffusion grey forecasting model for short-term
  manufacturing demand.
\newblock {\em Journal of the Operational Research Society}, 67(12):1439--1445,
  2016.

\bibitem{dengbook}
Julong Deng.
\newblock {\em Elements on Grey Theory}.
\newblock Huazhong University of Science and Technology Press, Wuhan, 2002.
\newblock in Chinese.

\bibitem{wy1}
Yong Wang and Xiangyi Yi.
\newblock Transient pressure behavior of a fractured vertical well with a
  finite-conductivity fracture in triple media carbonate reservoir.
\newblock {\em Journal of Porous Media}, 20(8):707--722, 2017.

\bibitem{wy2}
Yong Wang and Xiangyi Yi.
\newblock Flow modeling of well test analysis for a multiple-fractured
  horizontal well in triple media carbonate reservoir.
\newblock {\em International Journal of Nonlinear Sciences and Numerical
  Simulation}, 2018.

\bibitem{dingxf2018}
Xianfeng Ding, Dan Qu, and Haiyan Qiu.
\newblock A new production prediction model based on taylor expansion formula.
\newblock {\em Mathematical Problems in Engineering}, 2018:1--12, 2018.

\bibitem{hysspe}
Yisheng Hu, E.~MacKay, O.~Vazquez, and O.~Ishkov.
\newblock Streamline simulation of barium sulfate precipitation occurring
  within the reservoir coupled with analyses of observed
  produced-water-chemistry data to aid scale management.
\newblock {\em SPE Production and Operations}, 33(1):85--101, 2018.

\bibitem{wy3}
Yong Wang, Chen Zhang, Tao Chen, and Xin Ma.
\newblock Modeling the nonlinear flow for a multiple-fractured horizontal well
  with multiple finite-conductivity fractures in triple media carbonate
  reservoir.
\newblock {\em Journal of Porous Media}, 21(12):1283--1305, 2018.

\bibitem{LIANG2019315}
Yan Liang, Weiguang Cai, and Minda Ma.
\newblock {Carbon dioxide intensity and income level in the Chinese megacities'
  residential building sector: Decomposition and decoupling analyses}.
\newblock {\em Science of The Total Environment}, 677:315 -- 327, 2019.

\bibitem{wwq2018math}
Wenqing Wu, Yinghui Tang, Miaomiao Yu, Ying Jiang, and Hui Liu.
\newblock {Reliability analysis of a k-out-of-n: G system with general repair
  times and replaceable repair equipment}.
\newblock {\em Quality Technology \& Quantitative Management}, 15(2):274--300,
  2018.

\bibitem{chenngbm11}
Chun-I Chen, Hong~Long Chen, and Shuo-Pei Chen.
\newblock Forecasting of foreign exchange rates of {Taiwans} major trading
  partners by novel nonlinear grey {Bernoulli} model {NGBM}(1, 1).
\newblock {\em Communications in Nonlinear Science and Numerical Simulation},
  13(6):1194 -- 1204, 2008.

\bibitem{wangzx2011}
Zheng-Xin Wang, Keith~W. Hipel, Qian Wang, and Sha-Wei He.
\newblock An optimized {NGBM}(1,1) model for forecasting the qualified
  discharge rate of industrial wastewater in {China}.
\newblock {\em Applied Mathematical Modelling}, 35(12):5524 -- 5532, 2011.

\bibitem{wang2014}
Zheng-Xin Wang.
\newblock Nonlinear grey prediction model with convolution integral {NGMC} and
  its application to the forecasting of {China's} industrial emissions.
\newblock {\em Journal of Applied Mathematics}, 2014, 2014.

\bibitem{wang2017}
Zheng-Xin Wang and De-Jun Ye.
\newblock Forecasting {Chinese} carbon emissions from fossil energy consumption
  using non-linear grey multivariable models.
\newblock {\em Journal of Cleaner Production}, 142, Part 2:600 -- 612, 2017.

\bibitem{ma2019bernoulli}
Xin Ma, Zhibin Liu, and Yong Wang.
\newblock {Application of a novel nonlinear multivariate grey Bernoulli model
  to predict the tourist income of China}.
\newblock {\em Journal of Computational and Applied Mathematics}, 347:84--94,
  2019.

\bibitem{wulifeng2013}
Lifeng Wu, Sifeng Liu, Ligen Yao, Shuli Yan, and Dinglin Liu.
\newblock Grey system model with the fractional order accumulation.
\newblock {\em Communications in Nonlinear Science and Numerical Simulation},
  18(7):1775--1785, 2013.

\bibitem{fgmweapon}
Shili Fang, Lifeng Wu, Zhigeng Fang, and Xiaojun Guo.
\newblock Using fractional {GM (1, 1)} model to predict the maintenance cost of
  weapon system.
\newblock {\em Journal of Grey System}, 25(3), 2013.

\bibitem{fgmgas}
Lifeng Wu, Sifeng Liu, Ding Chen, Ligen Yao, and Wei Cui.
\newblock Using gray model with fractional order accumulation to predict gas
  emission.
\newblock {\em Natural hazards}, 71(3):2231--2236, 2014.

\bibitem{wubeijing}
Lifeng Wu, Nu~Li, and Yingjie Yang.
\newblock {Prediction of air quality indicators for the Beijing-Tianjin-Hebei
  region}.
\newblock {\em Journal of Cleaner Production}, 196:682 -- 687, 2018.

\bibitem{wujjj}
Lifeng Wu and Hongying Zhao.
\newblock {Using FGM(1,1) model to predict the number of the lightly polluted
  day in Jing-Jin-Ji region of China}.
\newblock {\em Atmospheric Pollution Research}, 2018.

\bibitem{mao2015co2}
Mingyun Gao, Shuhua Mao, Xinping Yan, and Jianghui Wen.
\newblock {Estimation of Chinese CO$_2$ Emission Based on A Discrete Fractional
  Accumulation Grey Model.}
\newblock {\em Journal of Grey System}, 27(4):114--130, 2015.

\bibitem{zeng2018buffer}
Bo~Zeng, Huiming Duan, Yun Bai, and Wei Meng.
\newblock Forecasting the output of shale gas in {China} using an unbiased grey
  model and weakening buffer operator.
\newblock {\em Energy}, 151:238--249, 2018.

\bibitem{zeng2017self}
Bo~Zeng and Si-Feng Liu.
\newblock A self-adaptive intelligence gray prediction model with the optimal
  fractional order accumulating operator and its application.
\newblock {\em Mathematical Methods in the Applied Sciences},
  40(18):7843--7857, 2017.

\bibitem{duan2018}
Huiming Duan, Guang~Rong Lei, and Kailiang Shao.
\newblock Forecasting crude oil consumption in china using a grey prediction
  model with an optimal fractional-order accumulating operator.
\newblock {\em Complexity}, 2018:1--12, 2018.

\bibitem{wwq2018}
Wenqing Wu, Xin Ma, Bo~Zeng, Yong Wang, and Wei Cai.
\newblock {Application of the novel fractional grey model FAGMO (1, 1, k) to
  predict China's nuclear energy consumption}.
\newblock {\em Energy}, 165:223--234, 2018.

\bibitem{ma2019energy}
Xin Ma, Xie Mei, Wenqing Wu, Xinxing Wu, and Bo~Zeng.
\newblock {A novel fractional time delayed grey model with Grey Wolf Optimizer
  and its applications in forecasting the natural gas and coal consumption in
  Chongqing China}.
\newblock {\em Energy}, 178:487 -- 507, 2019.

\bibitem{wu2019bernoulli}
Wenqing Wu, Xin Ma, Bo~Zeng, Yong Wang, and Wei Cai.
\newblock {Forecasting short-term renewable energy consumption of China using a
  novel fractional nonlinear grey Bernoulli model}.
\newblock {\em Renewable Energy}, 140:70 -- 87, 2019.

\bibitem{wu2018gmc}
Li-Feng Wu, Xiao-Hui Gao, Yan-Li Xiao, Ying-Jie Yang, and Xiang-Nan Chen.
\newblock Using a novel multi-variable grey model to forecast the electricity
  consumption of {Shandong} {Province} in {China}.
\newblock {\em Energy}, 157:327--335, 2018.

\bibitem{ma2019apm}
Xin Ma, Mei Xie, Wenqing Wu, Bo~Zeng, Yong Wang, and Xinxing Wu.
\newblock The novel fractional discrete multivariate grey system model and its
  applications.
\newblock {\em Applied Mathematical Modelling}, 70:402 -- 424, 2019.

\bibitem{reduceerror}
Lifeng Wu, Sifeng Liu, Ligen Yao, Ruiting Xu, and Xunping Lei.
\newblock Using fractional order accumulation to reduce errors from inverse
  accumulated generating operator of grey model.
\newblock {\em Soft Computing}, 19(2):483--488, 2015.

\bibitem{fndgm}
Li-Feng Wu, Si-Feng Liu, Wei Cui, Ding-Lin Liu, and Tian-Xiang Yao.
\newblock Non-homogenous discrete grey model with fractional-order
  accumulation.
\newblock {\em Neural Computing and Applications}, 25(5):1215--1221, 2014.

\bibitem{frelation}
LF~Wu, SF~Liu, LG~Yao, and L~Yu.
\newblock Fractional order grey relational analysis and its application.
\newblock {\em Scientia Iranica. Transaction E, Industrial Engineering},
  22(3):1171--1178, 2015.

\bibitem{xdy2016}
Yang Yang and Dingy{\"u} Xue.
\newblock Continuous fractional-order grey model and electricity prediction
  research based on the observation error feedback.
\newblock {\em Energy}, 115:722--733, 2016.

\bibitem{xdy2017}
Yang Yang and Dingy{\"u} Xue.
\newblock An actual load forecasting methodology by interval grey modeling
  based on the fractional calculus.
\newblock {\em ISA transactions}, 2017.

\bibitem{weakbuffer}
Lifeng Wu, Sifeng Liu, Yingjie Yang, Lihua Ma, and Hongxia Liu.
\newblock Multi-variable weakening buffer operator and its application.
\newblock {\em Information Sciences}, 339:98--107, 2016.

\bibitem{mwfgm}
Wei Meng, Qian Li, and Bo~Zeng.
\newblock Study on fractional order grey reducing generation operator.
\newblock {\em Grey Systems: Theory and Application}, 6(1):80--95, 2016.

\bibitem{liu2016new}
Sifeng Liu, Yingjie Yang, Naiming Xie, and Jeffrey Forrest.
\newblock New progress of grey system theory in the new millennium.
\newblock {\em Grey Systems: Theory and Application}, 6(1):2--31, 2016.

\bibitem{cfd2014}
Roshdi Khalil, M~Al~Horani, Abdelrahman Yousef, and Mohammad Sababheh.
\newblock A new definition of fractional derivative.
\newblock {\em Journal of Computational and Applied Mathematics}, 264:65--70,
  2014.

\bibitem{cfd02}
I~Abu~Hammad and R~Khalil.
\newblock Fractional {Fourier} series with applications.
\newblock {\em Am. J. Comput. Appl. Math}, 4(6):187--191, 2014.

\bibitem{cfd03}
M~Abu~Hammad and R~Khalil.
\newblock Conformable fractional heat differential equation.
\newblock {\em Int. J. Pure Appl. Math}, 94(2):215--221, 2014.

\bibitem{cfd04}
R~Khalil and H~Abu-Shaab.
\newblock Solution of some conformable fractional differential equations.
\newblock {\em Int. J. Pure Appl. Math}, 103(4):667--673, 2015.

\bibitem{cfd05}
M~Abu Hammad and R~Khalil.
\newblock Legendre fractional differential equation and {Legendre} fractional
  polynomials.
\newblock {\em International Journal of Applied Mathematics Research},
  3(3):214--219, 2014.

\bibitem{mao2016novel}
Shuhua Mao, Mingyun Gao, Xinping Xiao, and Min Zhu.
\newblock A novel fractional grey system model and its application.
\newblock {\em Applied Mathematical Modelling}, 40(7-8):5063--5076, 2016.

\bibitem{liubook}
Sifeng Liu, Yi~Lin, and Jeffrey Yi~Lin Forrest.
\newblock {\em Grey systems: theory and applications}.
\newblock Springer, 2010.

\bibitem{Xiao2019AE}
Qinge Xiao, Congbo Li, Ying Tang, Lingling Li, and Li~Li.
\newblock A knowledge-driven method of adaptively optimizing process parameters
  for energy efficient turning.
\newblock {\em Energy}, 166:142--156, 2019.

\bibitem{yang2019}
Wendong Yang, Jianzhou Wang, Tong Niu, and Pei Du.
\newblock A hybrid forecasting system based on a dual decomposition strategy
  and multi-objective optimization for electricity price forecasting.
\newblock {\em Applied Energy}, 235:1205 -- 1225, 2019.

\bibitem{timedelayed}
Xin Ma and Zhibin Liu.
\newblock Application of a novel time-delayed polynomial grey model to predict
  the natural gas consumption in china.
\newblock {\em Journal of Computational and Applied Mathematics}, 324:17--24,
  2017.

\bibitem{timeseries}
Christoph Bergmeir, Rob~J Hyndman, and Bonsoo Koo.
\newblock A note on the validity of cross-validation for evaluating
  autoregressive time series prediction.
\newblock {\em Computational Statistics \& Data Analysis}, 120:70--83, 2018.

\end{thebibliography}
\end{document}